\documentclass[sigconf]{acmart}
\AtBeginDocument{%
  }

\copyrightyear{2026}
\acmYear{2026}
\setcopyright{cc}
\setcctype{by}
\acmConference[KDD '26]{Proceedings of the 32nd ACM SIGKDD Conference on Knowledge Discovery and Data Mining V.2}{August 09--13, 2026}{Jeju Island, Republic of Korea}
\acmBooktitle{Proceedings of the 32nd ACM SIGKDD Conference on Knowledge Discovery and Data Mining V.2 (KDD '26), August 09--13, 2026, Jeju Island, Republic of Korea}
\acmDOI{10.1145/3770855.3817961}
\acmISBN{979-8-4007-2259-2/2026/08}

\usepackage{amsmath}
\usepackage{amsthm}
\usepackage{graphicx}
\usepackage{outlines}
\usepackage{xcolor}
\usepackage{algorithm}
\usepackage{algpseudocode}
\usepackage{verbatim}
\usepackage{hyperref}
\usepackage[noabbrev,capitalize]{cleveref}
\crefname{equation}{Eq.}{Eqs.}
\usepackage{breqn}
\usepackage{bbm}
\usepackage{booktabs}
\usepackage{dsfont}
\usepackage{varwidth}
\usepackage{multirow}
\usepackage{tikz}
\usepackage{textcomp}
\usepackage{listings}
\usepackage{pifont}
\usepackage{xspace}
\usepackage{colortbl}
\usepackage{wrapfig}
\usepackage{tabularx}
\usepackage{makecell}
\usepackage{caption}
\usepackage{tikz}
\usetikzlibrary{tikzmark}
\usetikzlibrary{positioning}
\usepackage{subcaption}

\usepackage{tikz}
\usetikzlibrary{arrows.meta}

\usepackage{thmtools,thm-restate}
\theoremstyle{acmplain}

\usepackage{array} % for better column formatting

\usepackage{enumitem}
\setitemize{noitemsep,topsep=0pt,parsep=0pt,partopsep=0pt}
\setenumerate{noitemsep,topsep=0pt,parsep=0pt,partopsep=0pt}

\AtBeginDocument{%
  \setlength{\abovedisplayskip}{6pt plus 2pt minus 2pt}
  \setlength{\belowdisplayskip}{6pt plus 2pt minus 2pt}
  \setlength{\abovedisplayshortskip}{2pt plus 1pt minus 1pt}
  \setlength{\belowdisplayshortskip}{4pt plus 2pt minus 2pt}
}

\newcommand{\cut}[1]{}
\algrenewcommand\algorithmicrequire{\textbf{Input:}}
\algrenewcommand\algorithmicensure{\textbf{Output:}}

\newtheorem{lemma}{Lemma}

\newtheorem{definition}{Definition}

\newtheorem{example}{Example}
\newtheorem{claim}{Claim}
\newtheorem{observation}{Observation}
\newtheorem{fact}{Fact}

\newcommand{\val}{\ensuremath{U}}
\newcommand{\seedset}{T} %{\ensuremath{Seeds}}
\newcommand{\seed}{t}
\newcommand{\nodeset}{\ensuremath{V}}
\newcommand{\node}{v}
\newcommand{\neigh}{N}
\newcommand{\shap}{Shap}
\newcommand{\actprob}{p}
\newcommand{\Sset}{S}

\newcommand{\activeset}{A}
\newcommand{\timestep}{\tau}
\newcommand{\timeconst}{\mathbb{K}}

\newcommand{\Ssize}{k}

\newcommand{\valsim}{\hat{\val}}

\newcommand{\rrsize}{n^{\prime}}
\newcommand{\rrset}{\boldsymbol{R}}
\newcommand{\onestep}{Single-step\xspace}
\newcommand{\fixedstep}{$\timeconst$-steps\xspace}
\newcommand{\fullstep}{Complete\xspace} %{Full-steps\xspace}
\usepackage{adjustbox}

\usepackage{soul}

\newcommand{\var}{\mathrm{Var}}
\newcommand{\subcnt}{\gamma}
\newcommand{\subcntmat}{\Gamma}

\newcommand{\dpsingleopt}{\mathsf{ExactSingleStep}}
\newcommand{\permuteMC}{\mathsf{ApproxPermuteMC}}

\newcommand{\algrrset}{\mathsf{ApproxRRset}}

\newcommand{\algliveedge}{\mathsf{ApproxLiveEdge}}
\newcommand{\dpsingleoptexp}{\textsc{ExactSingleStep}\xspace}
\newcommand{\permuteMCexp}{\textsc{ApproxPermuteMC}\xspace}

\newcommand{\rrsetexp}{\textsc{ApproxRRset}\xspace}

\newcommand{\bruteforceexp}{\textsc{BruteForce}\xspace}
\newcommand{\algliveedgeexp}{\textsc{ApproxLiveEdge}\xspace}
\newcommand{\dblpdata}{DBLP Coauthor\xspace}
\newcommand{\congressdata}{Congressional Twitter\xspace}
\newcommand{\facebookdata}{Facebook\xspace}
\newcommand{\twitterdata}{Twitter\xspace}
\newcommand{\livejournaldata}{LiveJournal\xspace}
\newcommand{\orkutdata}{Orkut\xspace}

\newcommand{\delextra}[1]{}
\usepackage{etoolbox}
\newtoggle{fullversion}
\togglefalse{fullversion} % Uncomment this line to compile the conference version (with appendix)

\usepackage{balance}

\begin{document}
\title{Efficient Shapley-Based Influence Attribution in Social Networks}

\author{Fangzhu Shen}
\email{fangzhu.shen@duke.edu}
\affiliation{%
  \institution{Duke University}
  \city{Durham}\state{NC}\country{USA}}

\author{Amir Gilad}
\email{amirg@cs.huji.ac.il}
\affiliation{%
  \institution{The Hebrew University}
  \city{Jerusalem}\country{Israel}}

\author{Sudeepa Roy}
\email{sudeepa@cs.duke.edu}
\affiliation{%
  \institution{Duke University}
  \city{Durham}\state{NC}\country{USA}}

\renewcommand{\shortauthors}{Shen et al.}

\begin{abstract}
The ubiquity of social platforms has reshaped the way information, behaviors, and advertisements diffuse across networks, with influence propagation often initiated by a small set of ``seed'' users. While much of the literature emphasizes optimizing seed selection to maximize spread, a critical yet underexplored question remains: how to fairly estimate the contributions of individual seeds ``ex-ante'', i.e., before the diffusion process occurs? This capability is essential for budget allocation, influencer pricing, and fair, privacy-preserving credit distribution under uncertainty, without relying on ex-post cascade logs that capture only a single execution of influence propagation. We introduce a framework for ex-ante influence attribution based on Shapley values from cooperative game theory, which capture each seed's marginal impact in a principled and equitable manner. Adapting Shapley values to influence propagation raises unique computational challenges due to the stochastic nature of diffusion and the intricate dependencies across network structures. To address these challenges, we design polynomial-time algorithms for the special case of single-step activation that is of independent practical interest, establish a sharp tractability boundary by proving $\#P$-hardness for any propagation beyond one step, and develop approximation algorithms with provable guarantees for the standard IC model as well as time-bounded variants. Empirical evaluation on real-world and synthetic networks demonstrates that our methods are both efficient and effective, offering a practical mechanism for ex-ante influence attribution.
\end{abstract}


\begin{CCSXML}
<ccs2012>
<concept>
<concept_id>10002951.10003260.10003282.10003292</concept_id>
<concept_desc>Information systems~Social networks</concept_desc>
<concept_significance>500</concept_significance>
</concept>
<concept>
<concept_id>10003752.10010070.10010099.10010102</concept_id>
<concept_desc>Theory of computation~Solution concepts in game theory</concept_desc>
<concept_significance>300</concept_significance>
</concept>
<concept>
<concept_id>10003752.10003809.10003635</concept_id>
<concept_desc>Theory of computation~Graph algorithms analysis</concept_desc>
<concept_significance>100</concept_significance>
</concept>
</ccs2012>
\end{CCSXML}

\ccsdesc[500]{Information systems~Social networks}
\ccsdesc[300]{Theory of computation~Solution concepts in game theory}
\ccsdesc[100]{Theory of computation~Graph algorithms analysis}

\keywords{Social networks; Influence propagation; Influence diffusion; Influence attribution; Cooperative game theory; Shapley values; Independent cascade model; \#P-hardness; Approximation algorithms}

\maketitle

\newcommand\kddavailabilityurl{https://github.com/fangzhushen/Shapley-value-Influence-attribution}
\ifdefempty{\kddavailabilityurl}{}{
\begingroup\small\noindent\raggedright\textbf{Resource Availability:}\\
The source code of this paper has been made publicly available at \url{\kddavailabilityurl}.
\endgroup
}

\section{Introduction}\label{sec:intro}
The ever-growing connectivity of social media platforms and online communities has fundamentally transformed the way information, behaviors, and advertisements spread through networks.
In these widely used online environments, content such as news and advertisements can rapidly disseminate through peer-to-peer interactions and word-of-mouth effects. This phenomenon, commonly referred to as \emph{influence propagation} (or \emph{influence diffusion})~\cite{kempe2003maximizing, domingos2001mining,leskovec2007cost}, drives applications such as viral marketing, public health messaging, and the spread of news or ideas. Typically, an influence propagation process is initiated by a small set of strategically selected users, known as {\em seed nodes}. These nodes act as the origin points of the diffusion process, triggering cascades of activations or adoptions throughout the network.

While a vast body of research has focused on influence maximization~\cite{kempe2003maximizing,kimura2006tractable,chen2010scalable,kim2013scalable, borgs2014maximizing,tang2015influence,tang2014influence,demaine2014influence,chen2024link}, which aims to identify the optimal seed set to maximize future influence, a novel and equally important question has not yet been explored to our knowledge: \textbf{given a fixed set of seed nodes, how can we fairly estimate each node's contribution to the overall diffusion outcome ``before'' execution?}
We refer to this problem as \emph{``ex-ante influence attribution''}.
If we can accurately and efficiently attribute the contribution of each seed node, we can enable a broad range of downstream tasks: fair compensation of influencers in marketing campaigns, performance-based resource allocation in public health interventions, and identification of key contributors in the spread of news or misinformation. Attributing contributions ex-ante also allows us to `compare' two seed nodes, e.g., which influencer should be paid more, which influencer provides redundant coverage with respect to others, and who adds unique reach to many. 

The problem of influence attribution arises directly in practice. Consider a company launching a product campaign on social media platforms by hiring multiple influencers. These influencers serve as seed nodes that share the product with their followers, who may, in turn, recommend it to their own connections, creating successive waves of influence propagation over time.
Before launching, the company must decide how to compensate these influencers. A natural but naive approach is to estimate each seed node's contribution independently, for example, using traditional centrality measures such as out-degree~\cite{freeman2002centrality} or PageRank~\cite{brin1998anatomy}, or metrics like follower engagement rate.
However, influence in a social network is fundamentally combinatorial. The contribution of a seed node depends on which other seed nodes have already been hired, because their followers may overlap and their diffusion paths may interact. Two seed nodes with heavily overlapping follower bases provide redundant coverage, whereas a seed node that reaches an otherwise unreachable community provides a unique contribution to the overall propagation process. Consequently, independent measures cannot capture these interaction effects, leading to unfair compensation, i.e., overpaying or underpaying influencers.

To address the influence attribution problem in social networks where multiple seed nodes jointly drive a diffusion process, we adopt the concept of \emph{Shapley values}~\cite{shapley1953value} from cooperative game theory.
Shapley values provide a principled method for distributing a collective payoff among players by quantifying each player's contribution as their average marginal gain across all possible orders of player participation.
Notably, the Shapley value is the unique solution that satisfies four desirable axioms: (1) {\em efficiency} -- the total payoff is fully distributed among all players; (2) {\em symmetry} -- players with identical marginal contributions receive the same value; (3) {\em linearity} -- values are additive across games; and (4) {\em null player} -- a player with zero contribution receives zero value.
Together, these axioms provide a fair and theoretically grounded foundation for attributing influence among seed nodes.

Importantly, our framework operates in an \emph{ex-ante} manner: it measures each influencer's expected contribution based on the network structure and influence probabilities \emph{before} executing any campaign.
This stands in contrast to \emph{ex-post} attribution by Zhu et al.~\cite{zhu2021analysis}, which analyzes contributions after a campaign concludes based on a single observed cascade.
While ex-post attribution is useful for retrospective analysis, it cannot guide budget allocation or pricing decisions {\em `before'} campaigns run.
Moreover, ex-post methods require detailed cascade logs that are often unavailable due to privacy constraints, tracking limitations, or scale.
Ex-ante attribution addresses these limitations by enabling proactive decision-making based solely on network structure and diffusion probabilities, and also does not depend on one specific execution of the influence propagation process.

However, applying Shapley values to influence propagation presents significant computational challenges.
The standard Shapley formula requires summing over exponentially many coalitions.
Moreover, unlike ex-post attribution where influence values can be computed by counting in deterministic logs~\cite{zhu2021analysis}, estimating marginal contributions in the ex-ante setting requires computing the expected influence under probabilistic propagation from network structure and activation probabilities. Since the propagation through overlapping paths is non-linear, this computation is also challenging, demanding techniques specifically designed for the probabilistic setting.

\smallskip
\noindent
\textbf{Termination models.~}
In general, influence propagation can continue for a given number of steps $\timeconst$ -- referred to as \emph{\fixedstep} termination, or until no new nodes can be activated, known as \emph{\fullstep} termination (when $\timeconst = $ the number of non-seed nodes). In addition, we examine a practical special case, termed \emph{\onestep} termination, which models scenarios such as one-hop referral programs and affiliate marketing. For instance, a company may send free product samples to a set of reviewers, who then publish reviews that reach their audiences for promotion; however, audience members do not receive samples themselves and have no particular incentive to review further, so influence terminates after one step. 
Even in this seemingly simpler scenario, measuring contributions remains non-trivial, since multiple reviewers may target an overlapping set of audiences, leading to redundancy in which the marginal contribution of each additional reviewer diminishes. Understanding how to quantify contributions in this constrained setting is therefore of both practical and theoretical significance.% \red{We study all of three settings and proposes efficient algorithms.}

\smallskip
\noindent
\textbf{Our contributions.~}We present a comprehensive study of Shapley value computation for ex-ante influence attribution. \Cref{tab:results_summary} summarizes our algorithmic and complexity results. Specifically:
\begin{itemize}[leftmargin=*]
    \item We formalize the problem of ex-ante influence attribution as the computation of Shapley values under the independent cascade model~\cite{kempe2003maximizing} with different termination criteria {\bf (\Cref{sec:model})}. 
    \sloppy
        \item For \onestep termination, we design a novel dynamic-programming algorithm that computes \emph{exact} Shapley values in polynomial time, circumventing the exponential complexity of the standard Shapley formula (row~1 of \Cref{tab:results_summary}; {\bf \Cref{sec:single_step_case}}).
    \item We prove that exact computation becomes \#P-hard for 
$\timeconst \geq 2$ steps, establishing a tractability boundary: 
while $2$-step expected influence remains polynomial-time computable, 
the corresponding Shapley computation is \#P-hard. This is the first such 
hardness result specific to IC-based Shapley values and does not follow from general cooperative game hardness (row~2 of \Cref{tab:results_summary}; \textbf{\Cref{sec:hardness}}).~

\item For \fullstep termination (standard IC model) and \fixedstep termination, we develop two approximation algorithms with theoretical guarantees (rows~2-3 of \Cref{tab:results_summary}). $\algliveedge$ uses a novel credit-splitting technique on live-edge graphs, and $\algrrset$ adapts reverse reachable sets~\cite{tang2014influence} for Shapley-based influence attribution. Both significantly outperform a standard Monte Carlo baseline and enable efficient attribution on large-scale networks {\bf (\Cref{sec:multisteps})}.

    \item Finally, we conduct extensive experiments on six real-world datasets and synthetic networks, demonstrating the high accuracy and scalability of our proposed algorithms under different termination criteria {\bf (\Cref{sec:experiments})}.
\end{itemize}

\begin{table}[t]
    \centering
    \caption{Summary of our results for Shapley value computation for ex-ante influence attribution.}
    \label{tab:results_summary}
    \small
    \setlength{\tabcolsep}{2pt}
    \renewcommand{\arraystretch}{1.15}
    \begin{adjustbox}{max width=\columnwidth}
    \begin{tabular}{@{}l l l l@{}}
    \toprule
    \textbf{Termination} & \textbf{Algorithm} & \textbf{Guarantee} & \textbf{Complexity} \\
    \midrule
    \onestep ($\timeconst=1$)
        & \makecell[l]{$\dpsingleopt$\\ (Sec.~\ref{sec:single_step_case})}
        & Exact
        & $O(|\seedset|^3 |\nodeset|)$ \\
    \midrule
    \multirow{3}{*}{\shortstack[l]{\fixedstep ($\timeconst \ge 2$),\\ \fullstep (standard IC)\\ (\#P-hard, Thm.~\ref{thm:hardness})}}
        & \makecell[l]{$\algliveedge$ \\ (Sec.~\ref{subsec:live_edge})}
        & $(\epsilon, \delta)$-approx
        & $O\bigl(\frac{|\nodeset|^2}{\epsilon^2} \ln\frac{|\seedset|}{\delta} \cdot |E|\bigr)$ \\
    \cmidrule(l){2-4}
        & \makecell[l]{$\algrrset$ \\ (Sec.~\ref{subsec:rrset_approx_algo})}
        & $(\epsilon, \delta)$-approx %\makecell[l]{Mult.\ $\epsilon$ (top-$k$),\\ Additive (rest)}
        & Nearly linear in $|E|$ \\
    \bottomrule
    \end{tabular}
    \end{adjustbox}
\end{table}

\section{Related Work}\label{sec:related_work}
\textbf{Shapley values for influence and advertising attribution.~}
\citet{zhu2021analysis} computed Shapley values for ex-post influence attribution where a single deterministic cascade has been fully observed.
In contrast, our work addresses ex-ante influence attribution, where Shapley values are computed based on the {\em expected influence before observing the diffusion outcome}, enabling essential capabilities for planning and executing reward mechanisms in real-world campaigns. This ex-ante formulation leads to a distinct computational challenge: coalition values are expectations over stochastic diffusion rather than deterministic functions of a realized cascade. Therefore, the algorithms from~\cite{zhu2021analysis} do not transfer to our problem. 
\citet{singal2019shapley} studied Shapley-based attribution in a causal setting of online advertising, modeling the user journey as a Markov chain. These formulations differ from ours in the underlying process model.

\smallskip
\noindent
\textbf{Influence maximization and network centrality.~}
Much of the work on influence diffusion in social networks has focused on \emph{influence maximization} \cite{domingos2001mining,kempe2003maximizing}.
Various efficient heuristics and approximation algorithms have been proposed \cite{kimura2006tractable,chen2010scalable,kim2013scalable, borgs2014maximizing,tang2015influence,tang2014influence,demaine2014influence,chen2024link,gasko2023identification}. Notably, the use of reverse reachable sets~\cite{borgs2014maximizing,tang2015influence,tang2014influence} significantly improves the scalability; \cite{gasko2023identification} uses Shapley values as a heuristic for selecting influential seeds.
Nevertheless, this line of work targets prospective seed selection, rather than equitable attribution among the given seeds.
Another related stream of research investigates \emph{network centrality measures}~\cite{bonacich1972factoring,freeman2002centrality,brin1998anatomy,ghosh2014interplay,narayanam2011shapley,michalak2013efficient,chen2017interplay}, which aim to identify the structural importance of nodes.
Prior work~\cite{michalak2013efficient, chen2017interplay} has applied Shapley values to measure network centrality, rather than attributing influence among a fixed set of seeds.

\smallskip
\noindent
\textbf{General Shapley value estimation and applications.~} 
\sloppy
Efficient Shapley estimators are also developed for general cooperative games, including Monte Carlo sampling~\cite{castro2009polynomial,mitchell2022sampling,jia2019towards}, stratified sampling~\cite{burgess2021approximating,zhang2023efficient}, and the differential matrix~\cite{pang2025shapley}. These methods assume cheap oracle access to coalition values, whereas our setting requires computing expected influence under stochastic diffusion, which itself is $\#P$-hard for general graphs~\cite{kempe2003maximizing}, necessitating tailored algorithms.
Beyond social networks, Shapley values have been widely used in machine learning and databases to measure or explain contributions of different entities, such as features or data points in ML models~\cite{arenas2021tractability,lundberg2017unified,van2022tractability,ghorbani2019data,fryer2021shapley,akkas2024gnnshap}, or tuples or constraints in databases~\cite{livshits2021shapley,livshits2022shapley,deutch2022computing,reshef2020impact,kara2024shapley,deutch2021explanations}. These applications share the Shapley value framework but operate on fundamentally different value functions that do not involve stochastic diffusion over networks.

\section{Preliminaries}\label{sec:model}
In this section, we review related concepts and define the problem.

\subsection{Network and Independent Cascade Model}\label{subsec:influence_models}

A (social) network is modeled as a directed graph $G=(\nodeset, E)$, where $\nodeset$ is the set of {\em nodes} (users) and $E \subseteq \nodeset \times \nodeset$ is the set of directed {\em edges} (connections). We denote the set of {\em out-neighbors} of a node $u$ as $\neigh^{+}(u)$ and the set of {\em in-neighbors} of a node $u$ as $\neigh^{-}(u)$, i.e., $\neigh^{+}(u) = \{v \in \nodeset \mid (u,v) \in E\}$ and $\neigh^{-}(u) = \{v \in \nodeset \mid (v,u) \in E\}$.

\smallskip
\noindent
{\bf Influence propagation and the independent cascade model.~} 
In {\em influence propagation} or {\em diffusion}, at any timestep (step), each node exists in one of two states: \emph{active} or \emph{inactive}.  
An active node has adopted the influence (e.g., an idea or behavior) and stays active forever. In turn, it can influence its out-neighbors to become active. This process is referred to as \emph{activation}.

We consider the {\em independent cascade (IC) model}~\cite{kempe2003maximizing}. Formally, the input to the IC model is $(G(V, E), T, \actprob)$ where $G(V, E)$ is a network, $T \subseteq V$ is the set of {\em seed nodes} that are active initially (at step $\timestep=0$), and $\actprob: E \to [0,1]$ denotes the activation probability $\actprob(u,v)$ for each edge $(u, v) \in E$. 
In each step $\timestep \geq 1$, each node $u$ that is newly activated in the previous step $\timestep-1$ gets a single chance to activate each of its currently inactive out-neighbors $v$ with $\actprob(u,v)$. If successful, $v$ becomes active in step $\timestep$.

\smallskip
\noindent
{\bf Termination and scope.} 
Under the standard IC model, diffusion continues until no new nodes are 
activated. We call this \textbf{\fullstep} termination. In addition, we analyze two special cases that arise in practice and admit stronger computational results:

\begin{itemize}[leftmargin=*,topsep=2pt,itemsep=2pt]
\item \textbf{\onestep} terminates after $\timeconst = 1$ step, i.e., nodes are activated by their neighboring seed nodes and then the process stops. As described in Section~\ref{sec:intro}, \onestep termination models the practical scenario where only seed nodes promote products to their followers, and the followers have no incentive to share.
%when influence propagation stops after one step, 
%This models practical scenarios, but 
Although computing Shapley values even for $\timeconst= 1$ step remains challenging due to probabilistic overlaps with shared neighboring nodes, 
%. Thus, we design 
we show that in this case, Shapley values can be exactly computed in polynomial time in \Cref{sec:single_step_case}. 
\item \textbf{\fixedstep} terminates after $\timeconst$ steps for a given value of $\timeconst \geq 2$. This termination model generalizes the standard IC model (\fullstep termination) by setting $\timeconst = |V| - |\seedset|$, since the diffusion process continues only if at least one new node is activated per step. 
Therefore, our approximation algorithms (\Cref{sec:multisteps}) apply to both \fixedstep termination and \fullstep termination. Practically, \fixedstep\ termination models the situation where the influence up to a certain distance is measured in a large graph.

\end{itemize}

\subsection{Shapley Values for Influence Diffusion}
\label{subsec:shapley_influence}
We adapt the notion of Shapley values from cooperative game theory \cite{shapley1953value}, where the seed nodes $\seedset$ form the set of players in the cooperative game framework (See Appendix~\ref{sec:model_appendix}).

\smallskip
\noindent
{\bf Value Function.~}We define the value function as the expected number of non-seed nodes that get activated when the diffusion terminates.
Seed nodes are excluded from the value function -- they are already activated at step 0 and therefore do not contribute to the value collected by the diffusion process.
The %definition of the 
value function $\val(S)$ for a subset of seed nodes $S \subseteq T$ uses a subgraph $G_{T, S}$, where we remove $\seedset \setminus \Sset$ (seed nodes not in $S$) and their incident edges (incoming and outgoing). Then, the value is given by the expected number of non-seed nodes from $V \setminus T$ at the end of the diffusion process starting with $S$ as the seed nodes in $G_{T, S}$.

\begin{definition}[Value Function of a Set of Seed Nodes]\label{def:value_function} %[Influence of a Set of Seed Nodes]
Given $(G(V, E), \actprob, T)$ and a subset $S \subseteq T$, let $G_{T, S}(V', E')$ be the subgraph of $G$ where %vertices 
$V' = V \setminus (T \setminus S)$ and $E'$ is formed by removing all incoming and outgoing edges of $T \setminus S$ from $E$. Then, given a termination time $\timeconst$,
the \emph{value function} $\val: 2^{\seedset} \rightarrow \mathbb{R}$ for any subset $\Sset \subseteq \seedset$ is:
\begin{equation}\label{eq:seed-value}
    \begin{aligned}
    \val(\Sset) = &\mathbb{E} \Big[ \big| \{ u \in \nodeset \setminus \seedset \mid u \text{ is activated in } (G_{T, S}, p, S)  \} \big| \Big]
    \end{aligned}
    \end{equation}
under \onestep\, \fixedstep\ or \fullstep termination (the expectation $\mathbb E$ is over $p$).
\end{definition}
A natural property of this formulation is that $\val(\emptyset) = 0$, making the value function valid for Shapley values.

\begin{definition}[Shapley value of a Seed Node]\label{def:shapley_influence}
\sloppy
Given $(G(V, E), \actprob, T)$, the Shapley value of a seed node $\seed \in \seedset$ is:
\begin{align}
    \shap(\seed) &= \frac{1}{|\seedset|!} \sum_{\pi \in \Pi(\seedset)} \left[ \val(\Sset_{\pi, \seed} \cup \{\seed\}) - \val(\Sset_{\pi, \seed}) \right] \label{eq:shapley_value_1}\\
    &= \sum_{\Sset \subseteq \seedset \setminus \{\seed\}} \frac{|\Sset|! \left( |\seedset| - |\Sset| - 1 \right)!}{|\seedset|!} \left[ \val(\Sset \cup \{\seed\}) - \val(\Sset) \right] \label{eq:shapley_value_2}.
\end{align}
\end{definition}
\begin{example}\label{exm:prelim}
    %\sloppy
    Consider the network $G$ shown in \Cref{fig:ic_network_example}, with nodes $\nodeset = \{\seed_1, \seed_2, a, b, c, d\}$, set of seed nodes $\seedset = \{\seed_1, \seed_2\}$, where each edge has activation probability $0.5$.
    Under either \onestep or \fixedstep termination,
    $\shap(\seed_1) =  \frac{1}{2} \left[(\val(\{\seed_1\}) - \val(\emptyset)) + (\val(\{\seed_2, \seed_1\}) - \val(\{\seed_2\})) \right]$.
\end{example}

\smallskip
\noindent
{\bf Complexity overview.} 
When there is only one seed node $T = \{t\}$, it follows from the efficiency property of Shapley values that the sum of Shapley values of all players equals the value of the grand coalition~\cite{shapley1953value}, so we have $\shap(t) = \val(\{t\})$. Intuitively, this single seed node is responsible for all influence diffusion and thus receives full credit for it.
The problem of computing Shapley values becomes complex for $|T| > 1$ in an arbitrary graph $G$, since multiple seed nodes may be responsible for the same non-seed node being activated. In the following sections, we show that $\shap(t)$ for $t \in T$ can be computed in polynomial time for \onestep\ termination, but the problem becomes \#P-hard~\cite{valiant1979complexity} for \fixedstep\ with $\timeconst \geq 2$, for which we present efficient algorithms as practical solutions.

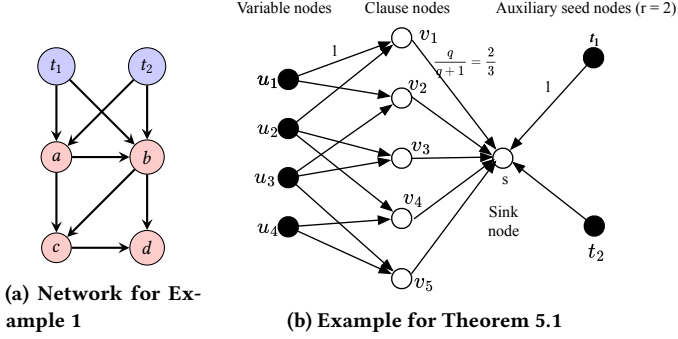
\begin{figure}[t]
    \begin{subfigure}{0.3\columnwidth}
    \centering
        \begin{tikzpicture}[>=stealth, node distance=1.5cm, auto]
            % Nodes
            \node[draw, circle, fill=blue!20, scale=0.8] (s) {\(\seed_1\)};
            \node[draw, circle, fill=blue!20, scale=0.8] (g) [right of=s] {\(\seed_2\)};
            \node[draw, circle, fill=red!20, scale=0.8] (a) [below of=s] {\(a\)};
            \node[draw, circle, fill=red!20, scale=0.8] (b) [right of=a] {\(b\)};
            \node[draw, circle, fill=red!20, scale=0.8] (c) [below of=a] {\(c\)};
            \node[draw, circle, fill=red!20, scale=0.8] (d) [below of=b] {\(d\)};
            % Edges with activation probabilities
            \draw[->, thick] (s) -- (a);
            \draw[->, thick] (s) -- (b);
            \draw[->, thick] (g) -- (a);
            \draw[->, thick] (g) -- (b);
            \draw[->, thick] (a) -- (b);
            \draw[->, thick] (a) -- (c);
            \draw[->, thick] (b) -- (d);
            \draw[->, thick] (b) -- (c);
            \draw[->, thick] (c) -- (d);
        \end{tikzpicture}
        %\vspace{-2mm}
        \caption{Network for \Cref{exm:prelim}}
        %arrows indicate edges. Activation probabilities of all edges are $0.5$.}
        \label{fig:ic_network_example}
    \end{subfigure}
    %\hfill
    %\hspace{-6mm}
    \begin{subfigure}{0.69\columnwidth}
    \centering
        \includegraphics[width=1.15\columnwidth]{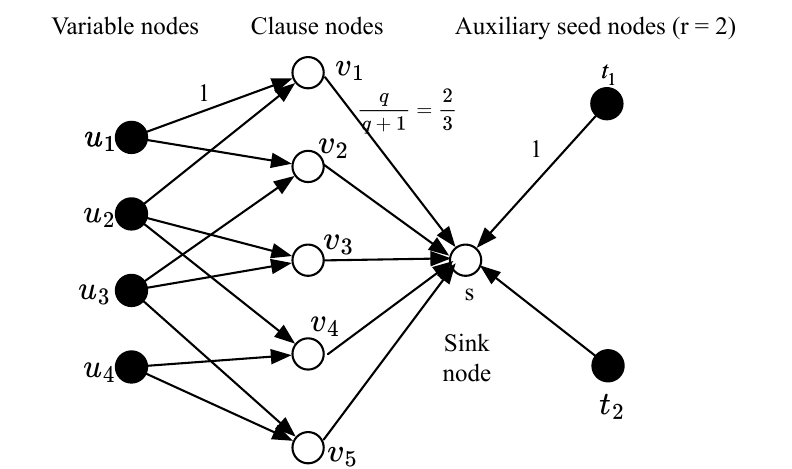}
        \caption{Example for \Cref{thm:hardness}}\label{fig:fixed_stand_graph_mainbody}
    \end{subfigure}
    \label{fig:combined_figures}
    \vspace{-1em}
    \caption{(a) shows a toy network of the IC model for \Cref{exm:prelim}. (b) shows an example of graph $G_{(r,q)}$ with $r=2, q=2$ for \Cref{thm:hardness}, where the monotone 2-CNF formula is $\phi = (x_1\lor x_2) \land (x_1 \lor x_3) \land (x_2\lor x_3) \land (x_2 \lor x_4) \land (x_3 \lor x_4)$.}
    \end{figure}
\section{\onestep Termination: Poly-time}\label{sec:single_step_case}
Computing Shapley values requires summing over exponentially-many terms in $|V|$, which is intractable even for moderate seed set sizes. In this section, we show that under the practical case of \onestep termination, this exponential barrier can be bypassed. We derive a closed-form expression with a novel dynamic programming recurrence, enabling exact computation in polynomial time.

\begin{restatable}{theorem}{singlepoly}\label{thm:single_poly}
    \sloppy
    Given $(G(\nodeset, E), \seedset, \actprob)$,  $\dpsingleopt$ (see \Cref{alg:single_step_main_opt} in Appendix~\ref{sec:singlestep_appendix}) %correctly 
    computes the Shapley values for all seed nodes under \onestep termination in $O(|\nodeset| \cdot |\seedset|^3)$ time. %can be computed in polynomial time in $|\nodeset|$.
\end{restatable}
%$O(|\nodeset|^4)$
We outline the key ideas and analyze $\dpsingleopt$ below. See Appendix~\ref{sec:singlestep_appendix} for the complete proof and pseudocode (\Cref{alg:single_step_main_opt}).

\smallskip
\noindent
\textbf{Closed-form expression.~}
Under \onestep termination, each non-seed node $u$ can only be activated by its direct seed neighbors, and these activation attempts are mutually independent. This allows us to express the marginal contribution of seed $\seed$ as a sum over its non-seed out-neighbors: for each such $u$, the contribution of $\seed$ to $u$ is the probability that $\seed$ activates $u$ while all other seeds in the coalition fail to do so. Substituting into Eq.~(\ref{eq:shapley_value_2}) yields:
\begin{equation}\label{eq:shapley_sum_mainbody}
    \shap(\seed) = \sum_{\Ssize = 0}^{|\seedset| - 1} \frac{\Ssize!(|\seedset| - \Ssize - 1)!}{|\seedset|}  \cdot \left[\sum_{u \in (\nodeset \setminus \seedset) \cap \neigh^{+}(\seed)} \actprob_{\seed,u} \cdot \alpha_{u, \seed}(\seedset \setminus \{\seed\}, \Ssize)\right],
\end{equation}
where $\alpha_{u,\seed}(L,k) = \sum_{\substack{\Sset \subseteq L \\ |\Sset|=k}} \left( \prod_{w \in \neigh^{-}(u)\cap \Sset}(1-\actprob_{w,u}) \right)$ aggregates the failure probabilities over all $\Ssize$-sized subsets of $L$, for any $L \subseteq \seedset \setminus \{\seed\}$.

\smallskip
\noindent
\textbf{Dynamic programming recurrence.~}
The key challenge is computing $\alpha_{u,\seed}(\seedset \setminus \{\seed\}, \Ssize)$ for all $\Ssize$ values efficiently, as the na\"ive approach enumerates $\binom{\seedset \setminus \{\seed\}}{\Ssize}$ subsets. We observe that for any subset $L \subseteq \seedset \setminus \{ \seed \}$ and any seed $n \in L$, we can partition $\Ssize$-sized subsets of $L$ into those containing $n$ and those that do not, yielding:
\begin{equation}\label{eq:dp_decomposition_formula_mainbody}
    \alpha_{u,\seed}(L,\Ssize) = (1-\actprob_{n,u}) \cdot \alpha_{u,\seed}(L \setminus \{n\}, \Ssize-1) + \alpha_{u,\seed}(L \setminus \{n\}, \Ssize),
\end{equation}
with base cases $\alpha_{u,\seed}(L,0) = 1$ and $\alpha_{u,\seed}(L,\Ssize) = 0$ for $\Ssize > |L|$. 
This recurrence computes $\alpha_{u,\seed}(L, \Ssize)$ for all $\Ssize \in \{0, \ldots, |L|\}$ in $O(|L|^2)$ time.

\smallskip
\noindent
\textbf{Local decomposition for scalability.~}
To optimize the computation, we decompose
the Shapley value computation in graph $G(V,E)$ into multiple subgraphs for parallel computation. For each non-seed node $u$, let $\seedset(u)$ denote the set of seed in-neighbors of $u$. Then we form a local subgraph $G_{u}(V_u, E_u)$ where $V_u = \seedset(u) \cup \{u\}$, and $E_u = \{(t,u) \mid t \in \seedset(u)\}$. Since the activation of each non-seed node under \onestep termination is independent of other activations, the value function $\val(\Sset)$ can be decomposed as the sum of value functions in all local subgraphs. By the linearity property:
\begin{equation}\label{eq:local_decomp}
\shap(\seed) = \sum_{u\in \{\nodeset \setminus \seedset\} \cap \neigh^{+}(\seed)} \;\shap_u(\seed),
\end{equation}
where $\shap_u(\seed)$ is the Shapley value of $\seed$ in subgraph $G_u$.

\smallskip
\noindent
\textbf{Poly-time algorithm $\dpsingleopt$.~}
For each non-seed node $u$, the algorithm first constructs a local subgraph $G_u$ and precomputes the binomial coefficients from Eq.~(\ref{eq:shapley_sum_mainbody}). Then for each $\seed \in \seedset(u)$, it uses the dynamic programming recurrence in Eq.~(\ref{eq:dp_decomposition_formula_mainbody}) to compute $\alpha_{u,\seed}(L, \Ssize)$ for all $\Ssize$ values, and applies Eq.~(\ref{eq:shapley_sum_mainbody}) to obtain the local Shapley values $\shap_{u}(\seed)$. Finally, the algorithm aggregates all local contributions via Eq.~(\ref{eq:local_decomp}) to obtain the final Shapley values.

\smallskip
\noindent
\textbf{Complexity analysis.}
For each seed $\seed \in \seedset(u)$ in subgraph $G_u$, computing $\alpha_{u,\seed}(\seedset(u) \setminus \{\seed\}, \Ssize)$ for all $\Ssize \in \{0, \ldots, |\seedset(u)|-1\}$ takes $O(|\seedset(u)|^2)$ time via the recurrence in Eq.~(\ref{eq:dp_decomposition_formula_mainbody}). Since there are $|\seedset(u)|$ seeds in each subgraph, computing $\shap_u(\seed)$ for all $\seed \in \seedset(u)$ takes $O(|\seedset(u)|^3)$ time per subgraph. There are at most $|\nodeset| - |\seedset|$ non-seed nodes. Since $|\seedset(u)| \leq |\seedset|$ for each subgraph, the total time complexity is $O\bigl((|\nodeset| - |\seedset|) \cdot |\seedset|^3\bigr) = O(|\nodeset| \cdot |\seedset|^3)$. In practice, we parallelize the computation across subgraphs to further improve efficiency.

\section{$\#P$-Hardness for $2$-step termination}\label{sec:hardness}
We show that the problem is $\#P$-hard even for $\timeconst = 2$ steps, which establishes a sharp tractability boundary and justifies the need for approximation algorithms (\Cref{sec:multisteps}). This hardness result does not follow from the known $\#P$-hardness of computing expected influence under \fullstep termination~\cite{chen2010scalable}, since expected $2$-step influence is polynomial-time computable. Nor does the exponential summation in the Shapley value definition by itself imply hardness for any specific value function, as we already showed a poly-time exact algorithm in \Cref{sec:single_step_case}. In this section, we outline the key ideas of the proof. See Appendix~\ref{sec:appendix_fixedstep} for the complete proof.

\begin{restatable}{theorem}{fixedstephardness}\label{thm:hardness}
Computing $\shap(t)$ for a seed node $t \in T$ is $\#P$-hard for \fixedstep\ termination with $\timeconst \ge 2$.
\end{restatable}

We give a reduction from the $\#P$-complete problem of counting satisfying assignments for a monotone 2-CNF formula  (denoted as \#2-CNF)~\cite{valiant1979complexity}. The constructed graph has maximum directed path-length 2, so the reduction 
applies to all \fixedstep terminations with $\timeconst \ge 2$. 

Let $\phi$ be a monotone $2$-CNF formula with $n$ variables $x_1,\dots,x_n$ and $m$ clauses $C_1,\dots,C_m$. 
Although the satisfiability of monotone 2-CNF formulas can be decided in poly-time, the 
counting version $\#$2-CNF is $\#P$-hard~\cite{valiant1979complexity}. Here  
each clause $C_j$ is a disjunction of two unnegated variables. For each $C_j$, define $\var(C_j) \;=\; \{\,x_i \mid x_i \text{ appears in } C_j\}$. 
The proof proceeds in two main parts: (1) constructing a family of influence graphs $G_{(r, q)}$ using parameters $r, q$ that we define below, and deriving a set of Shapley values on these graphs, and (2) recovering \#2-CNF from these Shapley values. %showing how to 
The reduction is a general Turing reduction and is not {\em parsimonious} (one-to-one) since it uses the oracle to compute Shapley values multiple times to solve \#2-CNF. 

\smallskip
\noindent
\textbf{(1) Graph Construction and Shapley Value Expression.~} 
Given the monotone 2-CNF formula $\phi$, we build a %two-step influence instance 
graph $G_{(r,q)} = (V_{(r,q)},E_{(r,q)})$ (e.g., \Cref{fig:fixed_stand_graph_mainbody}) for every ordered pair $(r,q)\in [n] \times [m]$. The construction introduces a \emph{variable seed node} $u_i$ for each variable $x_i$, a \emph{clause node} $v_j$ for each clause $C_j$, a \emph{sink node} $s$, and $r$ \emph{auxiliary seed nodes} $t_1,\dots ,t_r$. 
We create edges with probability $1$ from $u_i$ to $v_j$ whenever $x_i\in \var(C_j)$, edges with probability $\frac{q}{q+1}$ from every $v_j$ to $s$, and edges with probability $1$ from every auxiliary seed node $t_\ell$ to $s$. The construction of $G_{(r, q)}$ takes polynomial time, and we construct $m \cdot n$ such graphs $G_{(r, q)}$ in total. % network graphs

We use the Shapley values of the first auxiliary seed node $t_1$ from graphs $G_{(r, q)}$, which we denote by $\shap_{(r, q)}(t_1)$, in the second step of our hardness proof ($t_1$ is chosen arbitrarily). For %the %first auxiliary seed node 
$t_1$, the only chance to influence $s$ is through the edge $(t_1,s)$ at step $1$, so the Shapley value $\shap_{(r,q)}(\seed_1)$ of $\seed_1$ in graph $G_{(r,q)}$ can be simplified as follows ($\mathbbm{1}$ is the indicator function): 
\begin{small}
\begin{equation}\label{eq:shap_t1}
\begin{aligned}
&\shap_{(r,q)}(\seed_1) =\\
&\frac{1}{n+r} + \sum_{k=1}^{n} \frac{k!(n+r-k-1)!}{(n+r)!} \cdot  \sum_{S \subseteq U \atop |S|=k} (\frac{1}{q+1})^{\sum_{v_j\in V}\mathbbm{1}[\exists u_i \in S \text{ s.t. } (u_i, v_j) \in E_{(r,q)}]}
\end{aligned}
\end{equation}
\end{small}

\smallskip
\noindent
\textbf{(2) Computing the number of satisfying assignments of $\phi$:~}
We denote the number of satisfying assignments of the 2-CNF formula $\phi$ by $\#_{\phi}$. We show that if one can compute $\shap_{(r,q)}(\seed_1)$ in poly-time for every pair $(r,q)$, then $\#_{\phi}$ can be computed in poly-time. 

First, we express both $\#_{\phi}$ and $\shap_{(r,q)}(\seed_1)$ in terms of $\subcnt_{k,c}$, which denotes 
%Define $\subcnt_{k,c}$ as the number of size-$k$ subsets of variables that satisfy exactly $c$ clauses in $\phi$.
% Define $\subcnt_{k,c}$ as 
the number of size-$k$ subsets of variables such that the assignment setting these $k$ variables to true (and all others to false) satisfies exactly $c$ clauses in $\phi$.
%Denote $\subcnt_{k,c}$ as the number of assignments satisfied that: (1) setting $k$ variables to true and all others to false; (2) satisfying exactly $c$ clauses in $\phi$;
Then $\#_{\phi}$ can be expressed as:
%\begin{small}\begin{equation}\label{eq:satisfying_assignments_mainbody}
    $\#_{\phi} = \sum_{k=1}^{n}\subcnt_{k,m}$,
%\end{equation}
%\end{small}
and $\shap_{(r,q)}(\seed_1)$ can be expressed as:
%\begin{small}
\begin{equation}\label{eq:shapley_nck_mainbody}
    \shap_{(r,q)}(\seed_1) = \frac{1}{n+r} + \sum_{k=1}^{n} \sum_{c=1}^{m} \frac{k!(n+r-k-1)!}{(n+r)!} \cdot (\frac{1}{q+1})^{c} \cdot \subcnt_{k,c}
\end{equation}
%\end{small}
%Then we form a matrix system and solve for $\subcnt_{k,c}$. 
Collecting \Cref{eq:shapley_nck_mainbody} for all pairs $(r,q)$ gives the %following 
matrix system:
%\begin{equation}\label{eq:system_linear}
    $A \times \subcntmat \times B = D$; 
%\end{equation}
where:
%\begin{itemize}[leftmargin=*]
 %   \item 
 $A$ is an $n \times n$ matrix with entries $a_{r,k} = \frac{k!(n+r-k-1)!}{(n+r)!}$;% for $k,p \in [1,n]$
        % \begin{equation}\label{eq:matrix_A}
        % A = \begin{pmatrix}
        %     \frac{1!(n-1)!}{(n+1)!} & \frac{2!(n-2)!}{(n+1)!} & \cdots & \frac{n!0!}{(n+1)!} \\  
        %     \frac{1!(n)!}{(n+2)!} & \frac{2!(n-1)!}{(n+2)!} & \cdots & \frac{n!1!}{(n+2)!} \\  
        %     \vdots & \vdots & \ddots & \vdots \\
        %     \frac{1!(2n-2)!}{(2n)!} & \frac{2!(2n-3)!}{(2n)!} & \cdots & \frac{n!(n-1)!}{(2n)!}
        %     \end{pmatrix}
        % \end{equation}
 %   \item 
 $B$ is an $m \times m$ matrix with entries $b_{c,q} = (\frac{1}{q+1})^{c}$;% for $c,q \in [1,m]$:
    % \begin{equation}\label{eq:matrix_B}
    %     B =\begin{pmatrix}
    %     (\frac{1}{2})^{1} & (\frac{1}{3})^{1} & \cdots & (\frac{1}{m+1})^{1} \\
    %     \vdots & \vdots & \ddots & \vdots \\
    %     (\frac{1}{2})^{m} & (\frac{1}{3})^{m} & \cdots & (\frac{1}{m+1})^{m}
    %     \end{pmatrix}
    % \end{equation}
 %   \item 
  $D$ is an $n \times m$ matrix with entries $d_{r,q} = \shap_{(r,q)}(\seed_1) - \frac{1}{n+r}$;% for $p \in [1,n]$ and $q \in [1,m]$:
    %     \begin{equation}\label{eq:matrix_D}
    %         D = \begin{pmatrix}
    %         \shap_{(1,1)}(\seed_1) - \frac{1}{n+1} & \cdots & \shap_{(1,m)}(\seed_1) - \frac{1}{n+1} \\
    %         \vdots & \ddots & \vdots \\
    %         \shap_{(n,1)}(\seed_1) - \frac{1}{2n} & \cdots & \shap_{(n,m)}(\seed_1) - \frac{1}{2n}
    %         \end{pmatrix}
    %     \end{equation}
 %   \item 
 $\subcntmat$ is an $n \times m$ matrix with entries $\subcnt_{k,c}$. % for $k \in [1,n]$ and $c \in [1,m]$:
    % \begin{equation}\label{eq:matrix_subcnt}
    %     \subcntmat = \begin{pmatrix}
    %     \subcnt_{1,1} & \subcnt_{1,2} & \cdots & \subcnt_{1,m} \\
    %     \vdots & \vdots & \ddots & \vdots \\
    %     \subcnt_{n,1} & \subcnt_{n,2} & \cdots & \subcnt_{n,m}
    %     \end{pmatrix}
    % \end{equation}
%\end{itemize}
Both coefficient matrices $A$ and $B$ are non-singular, %\red{(details in the full version~\cite{fullversion})}, 
so we can solve for all $\subcnt_{k,c}$ in polynomial time if the Shapley values $\shap_{(r,q)}(\seed_1)$ are known. Finally, we sum $\subcnt_{k,m}$ over all $k$ to compute $\#_{\phi}$ in polynomial time, completing the reduction.

\section{Approximation Algorithms}\label{sec:multisteps}
In this section, we present three approximation algorithms for computing $\shap(\seed)$ under \fixedstep\ and \fullstep terminations:
(a) a baseline $\permuteMC$ adapted from the standard permutation-sampling Monte-Carlo (MC) algorithm~\cite{castro2009polynomial}; 
(b) a novel estimator $\algliveedge$ (\Cref{subsec:live_edge}) that avoids permutation sampling by introducing a new \emph{credit-splitting identity} on live-edge graphs; and
(c) $\algrrset$ (\Cref{subsec:rrset_approx_algo}), which adapts the Reverse Reachable (RR) set method~\cite{{tang2015influence,chen2017interplay}} to our problem of computing Shapley values for a given set of seed nodes and requires key modifications.
Both $\algliveedge$ and $\algrrset$ significantly improve efficiency while maintaining strong theoretical guarantees.

\smallskip
\noindent
\textbf{Baseline: Monte-Carlo Permutation Estimator}
We adapt the permutation-sampling framework \cite{castro2009polynomial,michalak2013efficient}
to obtain an unbiased $(\epsilon,\delta)$-approximation estimator.
$\permuteMC$ samples $n_{\pi} = O\left( \frac{|V|^2}{\epsilon^2} \ln (\frac{|\seedset|}{\delta}) \right)$ random permutations of seed set $\seedset$. Each permutation is a specific ordering of seed nodes. For each permutation $\pi_i$ and each seed $\seed_j$, it estimates the marginal contribution $\valsim(S_{\pi_i, j} \cup \{\seed_j\}) - \valsim(S_{\pi_i, j})$, where $S_{\pi_i, j}$ contains seeds preceding $\seed_j$ in $\pi_i$, and $\valsim(\Sset)$ estimates the value function by computing the average number of nodes activated by any seed set $\Sset$ over $n_{\text{MC}} = O\left( \frac{|V|^2}{\epsilon^2} \ln (\frac{|V|^2 |\seedset|^2}{\epsilon^2 \delta}) \right)$ propagation simulations. The final estimated Shapley value averages these marginal contributions across all sampled permutations.
However, its nested structure leads to a high runtime complexity $O\left( \frac{|V|^4}{\epsilon^4} \ln( \frac{|\seedset|}{\delta} + \frac{|V|^2 |\seedset|^2}{\epsilon^2 \delta}) \cdot |\seedset| \cdot |E| \right)$.
It simulates influence propagation $n_{\text{MC}}$ times per marginal contribution, and the process is repeated over $|\seedset|$ seeds and $n_{\pi}$ permutations, which makes it computationally expensive for large networks. Theoretical results of $\permuteMC$ are deferred to Appendix~\ref{sec:permute_mc_appendix}.

\subsection{Live-Edge Graph Estimator}\label{subsec:live_edge}
The $\permuteMC$ relies on sampling permutations of seeds and requires a nested loop to simulate the marginal contributions for each seed, creating a computational bottleneck that limits the scalability. To address this, we propose a novel approach $\algliveedge$, which generates live-edge graphs, i.e., realizations of the diffusion process by sampling all edges. Then it directly distributes influence on each live-edge graph realization. This approach bypasses the computational redundancy inherent in the baseline and offers significant performance gains.

The key insight of $\algliveedge$ is that for any single realization of the diffusion process (a sampled live-edge graph), the influence process is deterministic. Thus if a non-seed node $v$ is reached by a set of $k$ seeds in a sampled live-edge graph, each of the $k$ seeds has an equal probability of $1/k$ of being the first one to activate $v$ under a uniformly random permutation of seeds. Therefore, each of those $k$ seeds should receive $1/k$ credit for the activation of $v$ in this realization.

\sloppy
\noindent
\textbf{Algorithm Overview.~}$\algliveedge$ samples $n = O\left(\frac{|V|^2}{\epsilon^2} \ln(\frac{|\seedset|}{\delta})\right)$ live-edge graphs and processes each in two phases: (1) it performs a single multi-source BFS from all seeds to construct a reachability map $B_{g_i}$, recording which seeds can reach each non-seed node, and (2) it distributes unit credit for each non-seed node $v$ equally among all seeds in $B_{g_i}[v]$. Finally, it averages the credits across all samples to obtain the Shapley value estimates. The pseudocode is given in \Cref{alg:live_edge_graph} in \Cref{sec:algliveedge_appendix}.

\sloppy
Despite its structural simplicity, $\algliveedge$ achieves the same unbiasedness and approximation guarantees as $\permuteMC$. Moreover, since it traverses the graph only once per sampled live-edge graph, eliminating the need for nested simulation loops, the runtime complexity significantly improves from $|V|^4$ to $|V|^2$ (see \Cref{prop:liveedge_property} in \Cref{sec:algliveedge_appendix}).

\subsection{Reverse-Reachable Set Estimator}\label{subsec:rrset_approx_algo}
We present a third approximation algorithm, which is based on {\em Reverse-Reachable (RR)} sets. An RR set for a node $v$ consists of all nodes that can reach $v$ in a random live-edge graph $g$, which is sampled from $G(V,E)$ by removing each edge $e$ with $1-p(e)$ probability. Then a \emph{random RR set} is generated by first selecting a node $v \in V$ uniformly at random and then generating an RR set for $v$ on a live-edge graph $g$ randomly sampled from $G$.
$\algrrset$ avoids the permutations and nested-loop simulations by sampling reachability \emph{backwards}, from random non-seed nodes to seeds, and it also leads to better scalability compared to $\permuteMC$.

\smallskip
\noindent
\textbf{Adaptations for Shapley value estimation.~} RR sets were originally developed for influence maximization problems~\cite{borgs2014maximizing, tang2014influence, tang2015influence}, and have also been applied to Shapley centrality estimation~\cite{chen2017interplay}.
However, these prior works address different problems: influence maximization seeks to select an optimal seed set, while Shapley centrality measures the importance of all nodes in the network using Shapley values (no seed nodes). In contrast, our goal is to compute the Shapley value of each node in a \emph{fixed} seed set $\seedset$ with respect to a \emph{value function} $\val(\cdot)$ that counts only non-seed activations.

This difference necessitates modifications to both RR set construction and the input graph: %two key algorithmic adaptations:
(1) all RR sets are rooted at non-seed nodes rather than all nodes, and (2) all incoming edges to seed nodes are removed to avoid influence propagation between seeds.
As shown in the following \Cref{lem:shapley_rr}, these modifications ensure that the Shapley value of $\seed$ equals the sum of expected influence on each non-seed node, where the credit for each activated non-seed node is distributed equally among all seeds that could have reached it in a random RR set.

%the proof is in the full version \cite{fullversion}\hideappendix{, see \Cref{sec:rr_analysis_appendix}}):
 \begin{restatable}{lemma}{shapleyRr}\label{lem:shapley_rr}
    \sloppy
    Let $G^{\prime} = (V, E^{\prime})$ be the subgraph of $G$ obtained by removing all incoming edges to the set of seed nodes $\seedset$, and $\rrset$ be a random RR set generated in $G^{\prime}$ by selecting a non-seed node $v$ uniformly at random from $V \setminus \seedset$. Then, for any seed node $\seed \in \seedset$:
    \begin{small}
    \begin{equation}
    \shap(\seed) = \rrsize \cdot \mathbb{E}_{\rrset}\left[ \frac{ \mathbb{I}\{ \seed \in \rrset^{\prime} \} }{ | \rrset^{\prime} | } \right]
    \end{equation}
\end{small}
    where $\rrsize = |V \setminus \seedset|$, $\rrset^{\prime} = \rrset \cap \seedset$, and $\mathbb{I}\{ \cdot \}$ is the indicator function.
\end{restatable}

\smallskip
\noindent
\textbf{Our Algorithm.~}Based on Lemma~\ref{lem:shapley_rr}, we propose $\algrrset$, adapted from the adaptive-sampling framework of \cite {chen2017interplay}
with three parameters: $\varepsilon$ sets the multiplicative error, $\ell$ determines the confidence guarantee $1-\frac{1}{n^\ell}$, and $k$ sets the threshold such that the multiplicative error bound holds for the $k$ largest Shapley values. It runs in three phases (pseudocode is in \Cref{alg:shapley_rr} in Appendix~\ref{sec:rr_analysis_appendix}).

\begin{itemize}[leftmargin=*]
    \item \textbf{Phase 0: Graph Modification.} Remove all incoming edges to seed nodes from $G$ to obtain the modified graph $G^{\prime} = (V, E^{\prime})$. Denote the number of non-seed nodes as $\rrsize = |V \setminus \seedset|$.
    \item \textbf{Phase 1: Parameter Estimation.~}Estimate the required number of RR sets $\theta$ to achieve the $(\varepsilon, \delta)$-guarantee, following the adaptive-sampling framework of~\cite{chen2017interplay}. See Appendix~\ref{sec:rr_analysis_appendix} for details.
    \item \textbf{Phase 2: Shapley Value Estimation.}
    Generate $\theta$ RR sets, and for each RR set $R_j$ intersecting $\seedset$, increment the estimated contribution $\boldsymbol{est}(\seed)$ of each $\seed \in R_j \cap \seedset$ by $1/|R_j \cap \seedset|$, representing the probability that $\seed$ activates the root node of $R_j$ in this RR set. Finally, we get the estimated contribution for each seed $\seed$: $\boldsymbol{est}(t) = \sum_{j=1}^{\theta} \frac{\mathbb{I}\{\seed \in R_j \cap \seedset\}}{|R_j \cap \seedset|},$
    then Shapley value estimates are obtained by normalizing them as $\widehat{\shap}(t) =  \rrsize \cdot \frac{\boldsymbol{est}_t}{\theta}$

\end{itemize}

Our adaptation preserves the strong theoretical guarantees of the RR-set framework by adapting the proof structure from~\cite{chen2017interplay}. For accuracy, with high probability $1 - 1/n^{\ell}$, $\algrrset$ returns estimates that satisfy a multiplicative error bound relative to their own Shapley values for the top-$k$ largest Shapley values, and an additive error relative to the $k$-th largest value for the rest. For runtime, the expected runtime of $\algrrset$ is nearly linear in the number of edges. See Appendix~\ref{sec:rr_analysis_appendix} for the formal statements.

\section{Experiments}\label{sec:experiments}
In this section, we examine the necessity of Shapley values for influence attribution and the performance of proposed algorithms. The experiments demonstrate that: (1) common heuristics (e.g., degree centrality, greedy IM selection order) are poor proxies for measuring influence contributions.
(2) For the \onestep termination, $\dpsingleopt$ is exceptionally scalable, handling networks with millions of nodes and seeds in minutes.
(3) For the \fullstep and \fixedstep terminations, $\algrrset$ provides unique scalability benefits for large seed sets, while $\algliveedge$ offers a robust balance of accuracy and efficiency. Both consistently outperform the baseline $\permuteMC$, making large-scale influence attribution practical.

% Full experiments are in the full version~\cite{fullversion}. \hideappendix{See Appendix~\ref{sec:appendix_experiments}.}

\subsection{Setup}\label{subsec:experiment_setup}
We implement algorithms in Rust and conduct experiments on a machine with an EPYC CPU \footnote{AMD EPYC 7R13 48-Core Processor @2.6GHz and 256GiB RAM.}. Code and data are available.\footnote{\url{https://github.com/fangzhushen/Shapley-value-Influence-attribution}}

\smallskip
\noindent
\textbf{Datasets.~}We use six real-world networks (\Cref{tab:datasets}) varying in size and density\footnote{The \congressdata network represents Twitter interactions among members of the 117\textsuperscript{th} U.S.\ Congress (Feb--June 2022); We construct the \dblpdata network as a co-authorship network from the Proceedings of VLDB 2024 dataset.}, along with synthetic Erd\H{o}s-R\'enyi random graphs\footnote{Given the number of nodes $n$ and average degree $d$, the number of edges is $n \times d$, and the graph is randomly chosen from the collection of all satisfiable graphs.}~\cite{erdds1959random}. Results on synthetic graphs are averaged over 5 random instances.

\begin{table}[t]
    \caption{Real-world networks used in the experiments.}
    \label{tab:datasets}
    \centering
    %\begin{small}
    \begin{tabular}{|l|r|r|}
    \hline
    \textbf{Dataset} & \textbf{\# Nodes} & \textbf{Avg. Degree} \\
    \hline
    \congressdata~\cite{fink2023centrality,fink2023congressional} & 475 & 27.98 \\
    \dblpdata~\cite{dblp.xml.2025-04-01} & $\approx$1.7K & 7.99 \\
    \facebookdata~\cite{leskovec2012learning} & $\approx$4K & 43.70 \\
    \twitterdata~\cite{leskovec2012learning} & $\approx$81K & 12.75 \\
    \orkutdata~\cite{yang2012defining}& $\approx$3.1M & 76.28 \\
    \livejournaldata~\cite{backstrom2006group,leskovec2009community} & $\approx$4.8M & 14.23 \\
    \hline
    \end{tabular}
    %\end{small}
\end{table}

\begin{figure}[!t]
    \centering
    % --- Table part ---
    \captionof{table}{Top 10 members in the \congressdata network selected by out-degree. Each entry shows ``rank (actual value)'' by different metrics.}
    \renewcommand{\arraystretch}{1}
    \adjustbox{width=0.95\columnwidth,center}{
    %\footnotesize
    \begin{tabular}{|l|c|c|c|}
    \hline
    \textbf{Username} & \textbf{Out-Deg} & \textbf{PageRank} & \textbf{Shapley} \\ \hline
    SpeakerPelosi & 1 (210) & 1 (0.0167) & 2 (1.00) \\ \hline
    GOPLeader & 2 (157) & 9 (0.0013) & 3 (0.95) \\ \hline
    RepBobbyRush & 3 (111) & 5 (0.0036) & 4 (0.94) \\ \hline
    \rowcolor{gray!30} SenSchumer & 4 (97) & 3 (0.0053) & \textbf{7 (0.38)} \\ \hline
    \rowcolor{gray!30} SteveScalise & 5 (89) & 4(0.0049) & \textbf{1 (1.04)} \\ \hline
    RepMarkTakano & 6 (85) & 7 (0.0023) & 5 (0.67) \\ \hline
    rosadelauro & 7 (84) & 8 (0.0014) & 6 (0.46) \\ \hline
    \rowcolor{gray!30} LeaderHoyer & 8 (79) & 10 (0.0006) & \textbf{10 (0.24)} \\ \hline
    RepJimBanks & 9 (75) & 6 (0.0031) & 9 (0.25) \\ \hline
    SenWarren & 10 (71) & 2 (0.0068) & 8 (0.34) \\ \hline
    \end{tabular}%
    }
    \label{tab:centrality_vs_shapley}

    \vspace{0.5em}

    % --- Figure part ---
    \includegraphics[width=\linewidth]{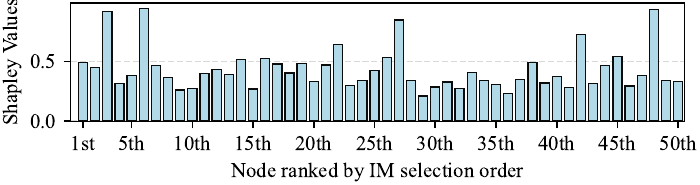}
    \captionof{figure}{Shapley values for top-50 seeds selected by greedy IM algorithm in the \congressdata network.}
    \label{fig:congress_shapley_values_im_top50}
\end{figure}

\smallskip
\noindent
\textbf{Assigning probability to edges.~} For \congressdata network, we used empirically learned influence probabilities that reflect the likelihood of influence based on observed Twitter interactions (e.g., retweets, mentions)~\cite{fink2023congressional}. For the other real-world and synthetic networks, we assign edge probabilities using the standard Weighted Cascade (WC) model that for an edge $(u,v)$, $p_{uv} = \frac{1}{d_{in}(v)}$, where $d_{in}(v)$ is the in-degree of node $v$~\cite{kempe2003maximizing}.

\smallskip
\noindent
\textbf{Seed selection.~}Unless otherwise specified, we use a degree-based selection to select the top-$k$ nodes by out-degree ~\cite{zhu2021analysis}. For the case study (Sec.~\ref{subsec:case_study}), we also use a \emph{greedy IM-based selection}.

\smallskip
\noindent
\textbf{Algorithms and parameter settings.} For \onestep termination, we evaluate \dpsingleoptexp (Sec.~\ref{sec:single_step_case}) against a \bruteforceexp baseline that enumerates all permutations. For \fullstep and \fixedstep terminations, we evaluate \permuteMCexp, \algliveedgeexp and \rrsetexp (Sec.~\ref{sec:multisteps}). Since no prior algorithm exists for ex-ante Shapley attribution in influence propagation, we therefore use \permuteMCexp as our baseline. We also compare against centrality metrics (Out-Degree, PageRank) in our case study (Sec.~\ref{subsec:case_study}).

\sloppy
For experiments in Sec.~\ref{subsec:full_step_experiments}, we fix number of samples for each approximation algorithm to achieve a balance between accuracy and computational efficiency: 500 permutations with 500 MC samples for \permuteMCexp, 5,000 live-edge graphs for \algliveedgeexp, and 500,000 RR-sets for \rrsetexp.
We adopt fixed-sample comparisons rather than fixing theoretical guarantee parameters (e.g., $\varepsilon$, $\delta$) because the three algorithms have fundamentally different per-sample costs and guarantee types, making identical theoretical guarantee parameters incomparable.
A parameter sensitivity study in Appendix~\ref{subsec:param_tuning}
examines how accuracy varies with sample size for each algorithm.

Since exact computation for \fixedstep termination is impractical, we establish ground truth using \rrsetexp with high-accuracy parameters ($\epsilon = 0.01, \ell = 1$), which provides theoretical guarantees on approximation quality. We run each experiment 5 times and report the average.

\smallskip
\noindent
\textbf{Evaluation Metric.~}
To assess the accuracy of algorithms, we report the \emph{average relative error}: $\frac{1}{|\seedset|} \sum_{i=1}^{|\seedset|} \frac{|\widehat{\shap}(\seed_i) - \shap(\seed_i)|}{\shap(\seed_i)}$,
where $\widehat{\shap}(\seed_i)$ is the estimated and $\shap(\seed_i)$ the ground-truth Shapley value. We also report \emph{runtime} to examine efficiency.

\subsection{Shapley Values vs. Other Methods}
\label{subsec:case_study}
\begin{figure}[t]
    \centering
    \includegraphics[width=0.9\columnwidth]{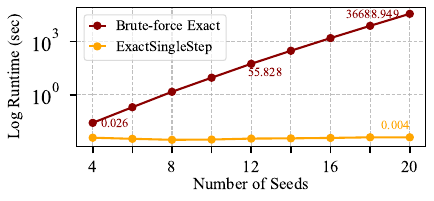}
    \caption{Runtime for \onestep termination on \facebookdata}
    \label{fig:exact_singlestep_vs_dp_single_runtime}
\end{figure}

\begin{table}[t]
\centering
\caption{Runtime (seconds) of \dpsingleoptexp.}
\label{tab:exact_single_step_runtime}
\begin{subtable}{\columnwidth}
\centering
%\small
\caption{Varying number of seeds on real-world networks. Dash (--) indicates the number of seeds larger than network size.}
\label{tab:exact_single_step_realworld}
%\small
%\footnotesize
\renewcommand{\arraystretch}{1}
\begin{tabular}{|l|rrrr|}
\hline
\multirow{2}{*}{\textbf{Dataset}} & \multicolumn{4}{c|}{\textbf{Number of Seeds}} \\
\cline{2-5}
& \textbf{10} & \textbf{1K} & \textbf{100K} & \textbf{1M} \\
\hline
\facebookdata & 0.004 & 0.013 &  -- & -- \\
\twitterdata & 0.007  & 0.079 &  -- & -- \\
\livejournaldata & 0.033  & 0.252 & 13.719 & 403.847 \\
\orkutdata & 0.101  & 1.455 & 14.211 & 26.546 \\
\hline
\end{tabular}
\end{subtable}
\vfill
\vspace{2mm} 
\begin{subtable}{\columnwidth}
    \centering
    \caption{Varying network structure parameters (default setting: 5k nodes 10 avg.degree, 500 seed nodes)}
    \label{tab:exact_single_step_synthetic}
    %\small
    %\footnotesize
    \renewcommand{\arraystretch}{1}
    \begin{tabular}{|l|cccc|}
    \hline
    %\textbf{Parameter} & \multicolumn{4}{c}{\textbf{Values}} \\
    %\hline
    \textbf{Number of nodes} & \textbf{1K} & \textbf{10K} & \textbf{100K} & \textbf{1M} \\
    \hline
    Runtime (sec) & 0.0024 & 0.0043 & 0.0058 & 0.0066 \\
    \hline
    \textbf{Average degree} & \textbf{5} & \textbf{50} & \textbf{100} & \textbf{200} \\
    \hline
    Runtime (sec) & 0.0031 & 0.0072 & 0.011 & 0.014 \\
    \hline
    \end{tabular}
\end{subtable}
\end{table}

Common methods to select seeds include centrality-based heuristics (e.g., degree, PageRank), and greedy approximation algorithms for influence maximization (IM). However, these metrics guide {\em which nodes to select}, not {\em how to measure contributions within a selected set}. Using the \congressdata network, we show that rankings by these selection methods diverge significantly from Shapley values, as they fail to account for redundant and synergistic influence among seeds. This underscores the need for a theoretically grounded attribution method such as Shapley values.

\smallskip
\noindent
\textbf{Centrality metrics vs. Shapley values.~}In \Cref{tab:centrality_vs_shapley}, we select the top-10 nodes by out-degree as seeds and report their out-degree, PageRank, and Shapley value rankings. Significant discrepancies emerge across all three metrics. Most notably, {\em SteveScalise} ranks 1\textsuperscript{st} in Shapley value despite being only 5\textsuperscript{th} in out-degree and 4\textsuperscript{th} in PageRank, indicating his influence reaches regions not well-covered by other high-centrality seeds. Conversely, {\em SenSchumer} ranks 4\textsuperscript{th} in out-degree and 3\textsuperscript{rd} in PageRank but only 7\textsuperscript{th} in Shapley value, as his influence pathways largely overlap with other seeds. These patterns show that centrality metrics, whether local (degree) or global (PageRank), cannot capture marginal contribution within a specific coalition where influence overlaps must be accounted for.

\smallskip
\noindent
\textbf{Greedy IM-based selection vs. Shapley values.~}
Figure~\ref{fig:congress_shapley_values_im_top50} shows Shapley values for top-50 seeds selected by the greedy IM algorithm~\cite{tang2015influence}, which iteratively chooses nodes maximizing marginal influence spread using reverse reachable sets. Although nodes selected earlier have higher incremental influence at selection time, this does not correspond to higher Shapley values: the 48\textsuperscript{th} selected node achieves the 2\textsuperscript{nd} highest Shapley value, surpassing the 1\textsuperscript{st} selected node (ranked 11\textsuperscript{th} in Shapley). This occurs because greedy IM evaluates marginal contribution at selection time with an incomplete seed set, while Shapley values measure contributions within the {\em complete seed set}, accounting for all interactions that emerge after all seeds are selected.

\subsection{\onestep Termination Results}\label{subsec:single_step_runtime}

\begin{figure}[t]
    \centering
    \begin{minipage}{1\columnwidth}
    \begin{subfigure}[b]{0.5\columnwidth}
        \centering
        \includegraphics[width=\textwidth]{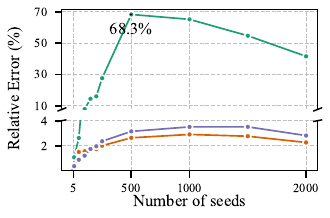}
        \caption{\facebookdata}
        \label{fig:facebook_error_completestep}
    \end{subfigure}
    \hfill
    \begin{subfigure}[b]{0.49\columnwidth}
    \centering
    \includegraphics[width=\textwidth]{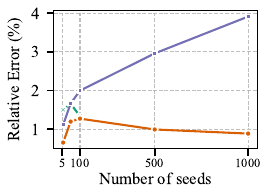}
    \caption{\twitterdata}
    \label{fig:twitter_error_completestep}
    \end{subfigure}
    \caption{Relative error vs. \# seeds for \fixedstep termination}
    \label{fig:error_vary_seeds_completestep}
    \end{minipage}
    \vfill
        \begin{minipage}{1\columnwidth}
        \begin{subfigure}[b]{0.49\columnwidth}
            \centering
            \includegraphics[width=\textwidth]{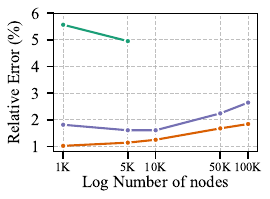}
            \caption{Varying number of nodes}
            \label{fig:synthetic_varysize_10avgdeg_10degseeds_shapley}
        \end{subfigure}
            \begin{subfigure}[b]{0.49\columnwidth}
            \centering
            \includegraphics[width=\textwidth]{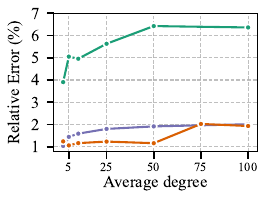}
            \caption{Varying average degree}
            \label{fig:synthetic_vary_deg_shapley_5000size_10degseeds_shapley}
        \end{subfigure}
        \end{minipage}
        \begin{minipage}{\columnwidth}
        \centering
        \includegraphics[width=\columnwidth]{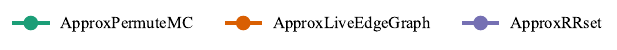}
    \end{minipage}
        \caption{Relative error vs. network structure for \fixedstep termination (default: 5k nodes, 500 seed nodes, 10 avg. degree).}
        \label{fig:quality_analysis}
    \end{figure}

    \begin{table*}[!htbp]
        \centering
        \caption{Runtime (seconds) vs. number of seeds for \fullstep termination. Dashes (--) indicate timeout (>12 hours)}
        \label{tab:runtime_vary_seed}
        %\small
        \setlength{\tabcolsep}{0.5pt}
        \renewcommand{\arraystretch}{1}
        \begin{subtable}{\textwidth}
            \centering
            \begin{tabularx}{\textwidth}{|l|*{4}{>{\centering\arraybackslash}X}|*{4}{>{\centering\arraybackslash}X}|*{4}{>{\centering\arraybackslash}X}|*{4}{>{\centering\arraybackslash}X}|}
            \hline
            \multirow{2}{*}{\makecell[l]{\textbf{Algorithm / Dataset}}}
                & \multicolumn{4}{c|}{\textbf{\facebookdata}}
                & \multicolumn{4}{c|}{\textbf{\twitterdata}}
                & \multicolumn{4}{c|}{\textbf{\livejournaldata}}
                & \multicolumn{4}{c|}{\textbf{\orkutdata}} \\
            \cline{2-17}
            & \textbf{10} & \textbf{100} & \textbf{1,000} & \textbf{2,000}
            & \textbf{10} & \textbf{100} & \textbf{1,000} & \textbf{5,000}
            & \textbf{1K} & \textbf{10K} & \textbf{100K} & \textbf{1M}
            & \textbf{1K} & \textbf{10K} & \textbf{100K} & \textbf{1M} \\
            \hline
            \rrsetexp
                & 0.615  & 0.317 & 0.147 & 0.078
                & 1.365 & 1.050 & 0.628 & 0.397
                & 45.65 & 39.48 & 32.90 & 13.81
                & 115.95 & 109.83 & 89.54 & \textbf{41.25} \\
            \algliveedgeexp
                & 0.255 & 0.277 & 0.358 & 0.434
                & 4.565 & 4.995 & 6.223 & 7.194
                & 560.48 & 591.52 & 1353.80 & 2655.62
                & 1051.06 & 1969.37 & 2294.51 & 3387.17 \\
            $\permuteMC$
                & 19.18 & 1322.52 & 15236 & 24508
                & 79.96 & 8022.48 & -- & --
                & -- & -- & -- & --
                & -- & -- & -- & -- \\
            \hline
            \end{tabularx}
        \end{subtable}
    \end{table*}

We first evaluate \dpsingleoptexp, our exact algorithm for the \onestep termination. As proven in \Cref{thm:single_poly}, it runs in polynomial time in the size of the network and seed set. Our experiments validate that this theoretical bound translates to exceptional practical performance, enabling exact computation on large-scale networks where brute-force approaches are infeasible.

Figure~\ref{fig:exact_singlestep_vs_dp_single_runtime} compares the runtime of \dpsingleoptexp against \bruteforceexp. Both obtain the same Shapley values, but the runtime of \bruteforceexp grows exponentially with the number of seeds, requiring over 10 hours for just 20 seeds. In contrast, \dpsingleoptexp completes in 0.004 seconds for 20 seeds, demonstrating orders-of-magnitude speedup.

Table~\ref{tab:exact_single_step_realworld} shows \dpsingleoptexp remains remarkably efficient as the number of seed nodes increases on large-scale real networks. On \orkutdata (3.1M nodes) and \livejournaldata (4.8M nodes) networks, it processes 1M seeds in 27 seconds and 7 minutes, respectively. Table~\ref{tab:exact_single_step_synthetic} shows graceful scaling with network size and density on synthetic networks. Increasing nodes from 1K to 1M increases runtime from only 0.0024s to 0.0066s, while increasing the average degree from 5 to 200 causes modest polynomial growth.

\subsection{Approximation Algorithms Results}\label{subsec:full_step_experiments}
We evaluate accuracy and efficiency under \fullstep termination, where exact computation is \#P-hard, varying the number of seeds, network sizes, and densities on real-world and synthetic datasets. Additional experiments varying the step bound $\timeconst$ under \fixedstep termination are in Appendix~\ref{subsec:impact_k_appendix}.

\noindent
\textbf{Accuracy analysis.~}\Cref{fig:error_vary_seeds_completestep} shows relative error vs.\ number of seeds on \facebookdata and \twitterdata networks. We observe: (1) \permuteMCexp performs poorly (up to $68.3\%$ error), while both \algliveedgeexp and \rrsetexp stay below 4\% error across all settings. (2) On \facebookdata network, a peak-and-decline pattern emerges for all algorithms. This is because as seeds increase, true Shapley values shrink faster than absolute error due to increased overlap among seeds, so relative error first rises (peaking around 500--1K seeds) then declines. (3) \algliveedgeexp is consistently slightly more accurate than \rrsetexp on both networks.

Figure~\ref{fig:quality_analysis} demonstrates that both \algliveedgeexp and \rrsetexp achieve high accuracy across diverse network structures, with \algliveedgeexp performing slightly better (less than 1\% difference). When varying the number of nodes (i.e., network size), both algorithms exhibit modest error growth while maintaining strong accuracy ($<2.6\%$ and $<1.8\%$ respectively). When varying the average degree (i.e., network density), \rrsetexp error increases to ~2.0\% in dense graphs, whereas \algliveedgeexp remains stable until avg. degree is 75, then rises to ~2.0\%. Both algorithms substantially outperform \permuteMCexp (up to ~6.5\% error), indicating robust scalability.

\smallskip
\noindent
\textbf{Runtime Analysis.~}
\Cref{tab:runtime_vary_seed} presents runtime vs.\ number of seeds on real-world datasets. First, \rrsetexp shows superior scalability with the runtime decreasing as the number of seeds grows (e.g., from $115.95$s to $41.25$s on the dense \orkutdata network), because we sample a fixed number of RR sets on $G^{\prime}$ obtained by removing all incoming edges to seed nodes. As the number of seeds grows, $G^{\prime}$ becomes sparser, reducing the average cost of the reverse graph traversals required per RR set. In contrast, the runtime of \algliveedgeexp shows a slight sublinear increase, as its cost is dominated by sampling live-edge graphs and performing a multi-source BFS from all seeds. While adding more seeds increases the number of BFS sources, the overall graph traversal cost does not grow proportionally. Thus it remains efficient across all datasets, completing its run on \livejournaldata with 1M seeds in under an hour. \permuteMCexp runtime explodes due to its nested sampling design.

\Cref{tab:performance_comparison} shows scalability with network sizes and densities. As the number of nodes grows (with 500 seeds and average degree 10), the higher accuracy of \algliveedgeexp comes at the cost of longer runtime. Both \algliveedgeexp and \rrsetexp scale near-linearly, and \rrsetexp is more efficient for large-scale networks, completing 1M-node graphs in ~12s versus \algliveedgeexp's 112s. \permuteMCexp times out beyond 1K nodes. Then when varying the average degree, \rrsetexp slows more rapidly (from 0.071s to 1.830s) and performs worse than \algliveedgeexp on dense graphs. Finally, \permuteMCexp is impractical and times out beyond 1K nodes or moderate density.

\smallskip
\noindent
\textbf{Summary of Findings.~}
Combining relative error and runtime results reveals distinct trade-offs between our algorithms: (1) \algliveedgeexp offers the best balance, maintaining high accuracy while achieving efficiency across almost all settings. For a large set of seed nodes, it provides high accuracy within a practical time. The effectiveness and efficiency make it a great choice for large-scale, real-world applications.
(2) \rrsetexp provides superior scalability with a slight increase of error, making it suitable for large-scale social networks with a large number of seed nodes (e.g., 41s on 4.8M-node \livejournaldata with 1M seeds).
(3) Both algorithms achieve similar accuracy within practical runtime, significantly outperforming \permuteMCexp, which fails to scale beyond small problem instances.

\begin{table}[t]
    \caption{Runtime (seconds) vs. network structure parameters for \fullstep termination. Dashes (--) indicate timeout (>12 hours).}
    \small
    \label{tab:performance_comparison}
    \centering
    %\footnotesize
    \setlength{\tabcolsep}{2.5pt}
    \renewcommand{\arraystretch}{1}
    \begin{tabular}{|l|cccc|cccc|}
    \hline
    \multirow{2}{*}{\textbf{Algorithm}} & \multicolumn{4}{c|}{\textbf{Number of Nodes}} & \multicolumn{4}{c|}{\textbf{Average Degree}} \\
    \cline{2-9}
    & \textbf{1K} & \textbf{10K} & \textbf{100K} & \textbf{1M} & \textbf{5} & \textbf{50} & \textbf{100} & \textbf{200} \\
    \hline
    \makecell[l]{\textsc{Approx-}\\\textsc{RRset}} & 0.064 & 0.215 & 1.557 & 11.575 & 0.071 & 0.452 & 0.927 & 1.830 \\
    \hline
    \makecell[l]{\textsc{ApproxLive-}\\\textsc{EdgeGraph}} & 0.082 & 0.543 & 6.428 & 111.916 & 0.251 & 0.473 & 0.753 & 1.318 \\
    \hline
    \makecell[l]{\textsc{Approx-}\\\textsc{PermuteMC}} & 26989 & -- & -- & -- & 2132.7 & 13948 & 27461 & -- \\
    \hline
    \end{tabular}
\end{table}

\section{Conclusion}\label{sec:conclusions}
In this paper, we studied the Shapley value computation for ex-ante influence attribution under the independent cascade model, covering various practical termination settings. We proposed a polynomial-time exact algorithm for \onestep termination, proved $\#P$-hardness for $\timeconst \geq 2$ steps, developed two scalable approximation algorithms with provable guarantees for arbitrary $\timeconst$ steps including \fullstep termination, and experimented with both synthetic and real datasets.
Our framework provides a foundational step towards influence attribution in social networks, opening several important future research directions.
While our current model assumes truthful participation, future work can integrate it with principles from mechanism design to analyze the incentive compatibility of seeds and prevent strategic manipulation. 
Extending to other diffusion models such as the Linear Threshold (LT) model is a non-trivial direction, since both our \onestep algorithm and the 
$\#P$-hardness result rely on graph structures specific to IC; thus analogous constructions for LT remain open. Finally, this framework can also be extended to dynamic networks with temporal changes, robustness to estimation error in edge probabilities, and enhanced by advanced sampling strategies to further improve scalability.

\begin{acks}
%\section{Acknowledgments}
The work of Amir Gilad was funded by the Israel Science Foundation (ISF) under grant 1702/24, the Scharf-Ullman Endowment, and the Alon Scholarship.
We thank Prof. Jian Pei (Duke University) for insightful discussions on this project.
We also thank Danyu Sun for participating in initial discussions on this project.
 %This work was supported by the [...] Research Fund of [...] (Number [...]). Additional funding was provided by [...] and [...]. We also thank [...] for contributing [...].
\end{acks}

%\clearpage

\bibliographystyle{ACM-Reference-Format}
\bibliography{bibtex}
\appendix
\section{Details From Section~\ref{sec:model}}\label{sec:model_appendix}

{\bf Shapley Values.~}
Given a set of players $N$ and a value (or utility) function $\val: 2^N \to \mathbb{R}$ with $\val(\emptyset) = 0$, where $\val(S)$ denotes the value raised by the subset $S$ of players, the Shapley value for each player $i \in N$ is defined as the average marginal contribution of $i$ over all possible permutations of the player set:
\begin{small}
\begin{equation}\label{eq:shapley-1}
    \shap_{N,\val}(i) = \frac{1}{|N|!} \sum_{\pi \in \Pi(N)} \left[ \val(\Sset_{\pi, i} \cup \{i\}) - \val(\Sset_{\pi, i}) \right]
\end{equation}
    \end{small}

Here $\Pi(N)$ is the set of all permutations of $N$, and $\Sset_{\pi, i}$ is the set of players in $N$ that precede $i$ in permutation $\pi$.
An equivalent formula for Shapley values is the following:
\begin{small}
\begin{equation}\label{eq:shapley-2}
    \shap_{N,\val}(i) = \sum_{S \subseteq N \setminus \{i\}} \frac{|S|!(|N|-|S|-1)!}{|N|!} \left[ \val(S \cup \{i\}) - \val(S) \right]
\end{equation}
    \end{small}
%\section{Appendix for Section~\ref{sec:single_step_case}}\label{sec:appendix_single_step}
\section{Details From Section~\ref{sec:single_step_case}}\label{sec:singlestep_appendix}
\subsection{Algorithm $\dpsingleopt$}
\begin{algorithm}[!htbp]
    \caption{$\dpsingleopt(G, \seedset, \actprob)$}
    \label{alg:single_step_main_opt}
    \begin{algorithmic}[1]
    \Require Network $G(\nodeset, E)$, seed set $\seedset$, activation probabilities $\actprob$
    \Ensure Shapley values $\shap(\seed)$ for each $\seed\in \seedset$
    \ForAll{non-seed node $u\in \nodeset\setminus \seedset$} \label{line:iterate_u}
    \State $\seedset(u) \gets \{t \in \seedset \mid (t,u) \in E\}$  \label{line:seedset_u}
        \If{$\seedset(u) \not= \emptyset$} \label{line:check_nonempty}
        \State Build subgraph $G_{u}(V_u, E_u)$ where $V_u = \seedset(u) \cup \{u\}$ and $E_u$ consists of edges between $\seedset(u)$ and $u$ \label{line:build_subgraph}
        \State Compute coefficients $C[\Ssize] \gets \frac{\Ssize!\,(|\seedset(u)|-\Ssize-1)!}{|\seedset(u)|!}$ for $0 \leq \Ssize \leq |\seedset(u)|-1$ \label{line:compute_coef}
        \ForAll{$t\in \seedset(u)$} \label{line:iterate_t}
            \State Compute $\alpha_{u,t}[\Ssize]$ for all $\Ssize$ by bottom-up DP using Eq.~(\ref{eq:dp_decomposition_formula}) \label{line:compute_alpha}
            \State $\shap_u(t) \gets \actprob_{t,u}\sum_{\Ssize=0}^{|\seedset(u)|-1} (C[\Ssize] \cdot \alpha_{u,t}[\Ssize])$ \label{line:compute_local_shapley}
        \EndFor
    \EndIf
    \EndFor
    \State $\shap(\seed) \gets \sum_{u\in \{\nodeset\setminus \seedset\} \cap \neigh^{+}(\seed)} \shap_u(\seed)$ for each $\seed \in \seedset$ \label{line:update_shapley_global}
    \State \Return $\shap(\seed)$ for each $\seed\in \seedset$
    \end{algorithmic}
\end{algorithm}

We present the poly-time algorithm $\dpsingleopt$ for \onestep\ in \Cref{alg:single_step_main_opt}. %\noindent
%\paragraph{\textbf{Poly-time algorithm}}
%$\dpsingleopt$ (pseudocode is in \Cref{alg:single_step_main_opt} in Appendix~\ref{sec:singlestep_appendix}) 
$\dpsingleopt$ processes each non-seed node $u$ independently (line~\ref{line:iterate_u}). For each $u$, we construct a local subgraph $G_u$ containing only $u$ and $\seedset(u)$ if its seed in-neighbors $\seedset(u)$ is not empty (lines~\ref{line:check_nonempty}, \ref{line:build_subgraph}). We first precompute all binomial coefficients for this local subgraph  as a coefficient array $C$ (line~\ref{line:compute_coef}), then for each $t \in \seedset(u)$, we compute the summation array $\alpha_{u,t}[\Ssize]$ for all $\Ssize$ values using the dynamic programming recurrence (line~\ref{line:compute_alpha}): we initialize $\alpha_{u,\seed}[0] = 1$ and $\alpha_{u,\seed}[k] = 0$ for $k > 0$, and consider nodes in $\seedset(u) \setminus \{\seed\}$ by an arbitrary order: $\{n_1, n_2, \cdots, n_{|\seedset(u)| - 1}\}$. We iterate over $i$ from $1$ to $|\seedset(u)| - 1$, and for each node $n_i$ in thie sequence, we update the $\alpha_{u,\seed}$ array by iterating $k$ from $i$ down to $1$:
\begin{equation}\label{eq:dp_decomposition}
    \alpha_{u,\seed}[k] \leftarrow \alpha_{u,\seed}[k] + (1-p_{n_i,u}) \cdot \alpha_{u,\seed}[k-1]. %\quad \text{for } k \text{ from } i \text{ to } 1
\end{equation}
Since there is only one non-seed node $u$ in the subgraph $G_{u}(V_u, E_u)$, we only need to compute one summary arrary and then we obtain the local Shapley value $\shap_u(t)$ as the dot product of coefficient array $C$ and the summation array $\alpha_{u,\seed}$, multiplied by $\actprob_{\seed, u}$ (line~\ref{line:compute_local_shapley}). Finally, we aggregate all local Shapley values to obtain the Shapley values in the original graph $G$ (line~\ref{line:update_shapley_global}).

\subsection{Proof of \Cref{thm:single_poly}}

\singlepoly*

Theorem~\ref{thm:single_poly} can be proved by combining the results of the following lemmas, which establish the correctness and polynomial-time complexity of the algorithm $\dpsingleopt$. Specifically, Lemma~\ref{lem:non_zero_marginal} simplifies the marginal contribution under \onestep\ termination, Lemma~\ref{lem:shapley_single_step} derives a closed-form expression for the Shapley value, Lemma~\ref{lem:decompose_eq} establishes the dynamic programming recurrence for efficient computation, Lemma~\ref{lem:recursive_general} analyzes its time complexity, and Lemma~\ref{lem:local_decomp} proves the decomposition property that enables the optimization. We provide the detailed proofs of these lemmas in the subsequent subsections, and present the complete proof of Theorem~\ref{thm:single_poly} in \Cref{subsubsec:complete_proof_single_poly}.

\subsubsection{Marginal contribution simplification}
\begin{restatable}{lemma}{nonzeromarginal}\label{lem:non_zero_marginal}
    Under the \onestep\ termination, the marginal contribution of seed $t \in \seedset \setminus \Sset$ to any coalition $\Sset \subseteq \seedset$ is:
    \begin{equation}\label{eq:marginval_simplified_appendix}
    \val(\Sset\cup\{t\}) - \val(\Sset)
    = \sum_{u \in (\nodeset \setminus \seedset) \cap \neigh^{+}(t)} \left[ \actprob_{t,u} \cdot \prod_{w \in \neigh^{-}(u)\cap \Sset}(1-\actprob_{w,u}) \right].
    \end{equation}
    \end{restatable}

\begin{proof}
    First, since we only consdier \onestep activation, by the \Cref{def:value_function} that $\val(\Sset)$ is the expected number of non-seed nodes activated solely by paths originating from $\Sset$, it can be rewritten as the sum of the probability that each non-seed node $u$ is activated by coalition $\Sset$ at step 1:
    \begin{equation}\label{eq:value_function_single}
    \begin{aligned}
        \val(\Sset) &= \sum_{u \in \nodeset \setminus \seedset} \Pr[u \in \activeset_1 \mid \text{paths originating from } \Sset]. \\
        %& = \sum_{u \in \nodeset \setminus \seedset}\ap_1(u, \Sset) \\
        %& = \sum_{u \in \nodeset \setminus \seedset} (1-\ap_0(u, \Sset))(1-\prod_{w \in \neigh(u)}(1-\actprob_{u,w}\ap_0(w, \Sset)))
    \end{aligned}
    \end{equation}

Since each active node independently attempts to activate its neighbors, the activation probability of each non-seed node under the \onestep activation equals to $1$ minus the probability that being failed to be activated by all of its neighbors in $\Sset$:
    \begin{equation}\label{eq:activation_prob}
    \begin{aligned}
        &\Pr[u \in \activeset_1 \mid \text{paths originating from } \Sset]\\ 
        =& 1 - \prod_{w \in \neigh^{-}(u)}(1-\actprob_{u,w}\cdot \mathbbm{1}[w \in \Sset])\\
        =& 1- \prod_{w \in \neigh^{-}(u)\cap \Sset}(1-\actprob_{u,w}).
    \end{aligned}
    \end{equation}

Moreover, only non-seed nodes with incoming edges from seed nodes within $\Sset$ have the chance to be activated. Thus, equation \eqref{eq:value_function_single} can be simplified as:
    \begin{equation}\label{eq:value_function_single_simplified}
        \val(\Sset) = \sum_{u \in \{\nodeset \setminus \seedset\} \cap \neigh^{+}(\Sset)} (1-\prod_{w \in \neigh^{-}(u)\cap \Sset}(1-\actprob_{u,w})).
    \end{equation}

Then we substitute \Cref{eq:value_function_single_simplified} into the formula of marginal contribution for node $\seed$ to $\Sset$ as follows:

\begin{equation}\label{eq:marginval}
\footnotesize
    \begin{aligned}
        & \val(\Sset\cup\{\seed\}) - \val(\Sset)\\
        =&  \sum_{u \in \{\nodeset \setminus \seedset\} \cap \neigh^{+}(\Sset \cup \{\seed\})} (1-\prod_{w \in \neigh^{-}(u)\cap (\Sset \cup \{\seed\})}(1-\actprob_{w,u}))\\
        &- \sum_{u \in \{\nodeset \setminus \seedset\} \cap \neigh^{+}(\Sset)} (1-\prod_{w \in \neigh^{-}(u)\cap \Sset}(1-\actprob_{w,u}))\\
        =& \sum_{u \in \{\nodeset \setminus \seedset\} \cap \neigh^{+}(\{\seed\})} (\prod_{w \in \neigh^{-}(u)\cap \Sset}(1-\actprob_{w,u}) -\prod_{w \in \neigh^{-}(u)\cap (\Sset \cup \{\seed\})}(1-\actprob_{w,u}))\\
        +& \sum_{u \in \{\nodeset \setminus \seedset\} \cap \neigh^{+}(\Sset)\setminus \neigh^{+}(\{\seed\})} (\prod_{w \in \neigh^{-}(u)\cap \Sset}(1-\actprob_{w,u}) -\prod_{w \in \neigh^{-}(u)\cap (\Sset \cup \{\seed\})}(1-\actprob_{w,u}))\\
    \end{aligned}
    \end{equation}
    In the above \Cref{eq:marginval}, we divide the marginal contribution into two disjoint summation terms. 
    In the first term, we consider $u \in \{\nodeset \setminus \seedset\} \cap \neigh^{+}(\{\seed\})$, which are the non-seed out-neighbor nodes of $v$, so we have 
    \begin{equation}
        \prod_{w \in \neigh^{-}(u)\cap (\Sset \cup \{\seed\})}(1-\actprob_{w,u}) = \left[\prod_{w \in \neigh^{-}(u)\cap \Sset}(1-\actprob_{w,u})\right] \cdot (1-\actprob_{\seed,u}),
    \end{equation}
    which follows the fact that adding $\seed$ to the set of $u$'s seed neighbors introduces the term $(1-\actprob_{\seed,u})$ to the product.
    
    In the second term, we consider $u \in\{\nodeset \setminus \seedset\} \cap \neigh^{+}(\Sset)\setminus \neigh^{+}(\{\seed\})$, which are non-seed nodes that are out-neighbor of seed nodes within $\Sset$ but are not $\seed$'s out-neighbor, so we have 
    \begin{equation}
        \prod_{w \in \neigh^{-}(u)\cap (\Sset \cup \{\seed\})}(1-\actprob_{w,u}) = \prod_{w \in \neigh^{-}(u)\cap \Sset}(1-\actprob_{w,u}),
    \end{equation}
    Therefore, the two products are identical, causing the second term to cancel out.

    Simplifying \Cref{eq:marginval}, we obtain Eq.~\ref{eq:marginval_simplified_appendix}. Note that $\actprob_{u,v} \in (0,1]$ for all $(u,v) \in E$, all terms in the summation are non-negative, we have $\val(\Sset \cup \{v\}) - \val(\Sset) \geq 0 $. Therefore, $\val(\Sset)$ is monotonically non-decreasing.
\end{proof}

\subsubsection{Closed-form expression}
\begin{restatable}{lemma}{shapleysinglesimplify}\label{lem:shapley_single_step}
    For any seed node $\seed \in \seedset$, under \onestep\ termination,
    %its Shapley value $\shap(\seed)$ under \onestep\ termination is:
    \begin{small}
    \begin{equation}\label{eq:shapley_sum}
        \begin{aligned}
        \shap(\seed) &= \sum_{\Ssize = 0}^{|\seedset| - 1} \frac{\Ssize!(|\seedset| - \Ssize - 1)!}{|\seedset|!} 
        %&\quad 
        \ \cdot \left[\sum_{u \in \{\nodeset \setminus \seedset\} \cap \neigh^{+}(\{\seed \})} \actprob_{\seed,u} \cdot \alpha_{u, \seed}(\seedset \setminus \{\seed\}, \Ssize) \right]
        \end{aligned}
    \end{equation}
    \end{small}
    \begin{small}
    \begin{equation}\label{eq:alpha_def}
    \text{where}\ \ \ \   \alpha_{u, \seed}(\seedset \setminus \{\seed\}, \Ssize) = \sum_{\substack{\Sset \subseteq \seedset \setminus \{\seed\} \\ |\Sset|=\Ssize}} \left( \prod_{w \in \neigh^{-}(u)\cap \Sset}(1-\actprob_{w,u}) \right).
    \end{equation}
    \end{small}
\end{restatable}

\begin{proof}
    Let $k$ denote the size of a subset $\Sset$. By the definition of Shapley value in \Cref{def:shapley_influence}:
    \begin{equation} \label{eq:shap_single_t1}
        \begin{aligned}
            &\shap(\seed)\\
            =& \sum_{\Sset\subseteq \seedset \setminus \{ \seed\} } \frac{|\Sset|! (|\seedset|-|\Sset| -1)!}{|\seedset|!} \cdot [\val(\Sset \cup \{ \seed \}) - \val(\Sset) ]\\
            =& \sum_{\Ssize=0}^{|\seedset|-1}\sum_{\Sset \subseteq \seedset\setminus t \atop |\Sset|=\Ssize}\frac{\Ssize!(|\seedset|-\Ssize-1)!}{|\seedset|!}\cdot [\val(\Sset \cup \{t\}) - \val(\Sset) ]\\
            = &\sum_{\Ssize=0}^{|\seedset|-1}\frac{\Ssize!(|\seedset|-\Ssize-1)!}{|\seedset|!}\cdot \Big(\sum_{\Sset \subseteq \seedset\setminus t \atop |\Sset|=\Ssize}[\val(\Sset \cup \{t\}) - \val(\Sset)]\Big)
        \end{aligned}
    \end{equation}
    Substituting the marginal contribution from \Cref{lem:non_zero_marginal} and switch the order of summation, we obtain:
    \begin{equation}
        \begin{aligned}
            &\shap(\seed) \\
            =&\sum_{\Ssize=0}^{|\seedset|-1}\frac{\Ssize!(|\seedset|-\Ssize-1)!}{|\seedset|!}\\
            &\cdot \left[\sum_{\Sset \subseteq \seedset\setminus \seed \atop |\Sset|=\Ssize} \sum_{u \in \{\nodeset \setminus \seedset\} \cap \neigh^{+}(\{\seed\})} [\actprob_{\seed,u} \cdot \prod_{w \in \neigh^{-}(u)\cap \Sset}(1-\actprob_{w,u})]\right]\\
            =&\sum_{\Ssize=0}^{|\seedset|-1}\frac{\Ssize!(|\seedset|-\Ssize-1)!}{|\seedset|!}\\
            &\cdot \left[\sum_{u \in \{\nodeset \setminus \seedset\} \cap \neigh^{+}(\{\seed\})} \actprob_{\seed,u} \cdot \sum_{\Sset \subseteq \seedset\setminus \seed \atop |\Sset|=\Ssize} [\prod_{w \in \neigh^{-}(u)\cap \Sset}(1-\actprob_{w,u})]\right]\\
            %=&p\cdot \sum_{k=0}^{|V|-1}\frac{k!(|V|-k-1)!}{|V|!}\cdot \sum_{i\in N_{t_1}}\sum_{S \subseteq V\setminus t_1 \atop |S|=k}[1-\activestatus_0(i)](1-p)^{\sum_{v\in S}e_{i,v}\activestatus_0(v)}\\
            %=&\actprob\sum_{\Ssize=0}^{|\seedset|-1}\frac{\Ssize!(|\seedset|-\Ssize-1)!}{|\seedset|!}\cdot \Big(\sum_{u\in \neigh(\seed)}[1-\activestatus_0(u)]\Big(\sum_{\Sset \subseteq \seedset\setminus \seed \atop |\Sset|=\Ssize}(1-\actprob)^{\sum_{i\in \Sset}\edge_{(u,i)}\activestatus_0(i)}\Big)\Big)\\
        \end{aligned}
    \end{equation}
\end{proof}

\subsubsection{Efficient recurrence for $\alpha_{u,t}(L,k)$}
\begin{restatable}{lemma}{decomposeeq}\label{lem:decompose_eq}
    Given a seed node $\seed \in \seedset$, one of its non-seed out-neighbor $u \in \{\nodeset \setminus \seedset\} \cap \neigh^{+}(\{\seed\})$ and an integer $k \in \{1, \ldots, |\seedset| - 1\}$, the following holds for any subset $L \subseteq \seedset \setminus \{ \seed \}$: %and any node $n \in L$:
     \begin{small}
    \begin{equation}\label{eq:dp_decomposition_formula}
    \begin{aligned}
        \alpha_{u,\seed}(L,k) = 
    \begin{cases}
        1 & \text{if } k = 0\\
        (1-\actprob_{n,u}) \cdot \alpha_{u,\seed}(L \setminus \{n\}, k-1) & \\
        + \alpha_{u,\seed}(L \setminus \{n\}, k) & \forall n \in L, \text{if } 1 \leq k \leq |L| \\
        0 & \text{otherwise},
    \end{cases}
    \end{aligned}
    \end{equation}
    \end{small}
    % \begin{equation}\label{eq:dp_decomposition_formula}
    % \begin{aligned}
    %     \alpha_{u,\seed}(L,k) = & (1-\actprob_{n,\seed}) \cdot \alpha_{u,\seed}(L \setminus \{n\}, k-1) \\
    %                    & + \alpha_{u,\seed}(L \setminus \{n\}, k)
    % \end{aligned}
    % \end{equation}
where $\alpha_{u,\seed}(L,k) = \sum_{\substack{\Sset \subseteq L \\ |\Sset|=k}} \left( \prod_{w \in \neigh^{-}(u)\cap \Sset}(1-\actprob_{w,u}) \right)$.
\end{restatable}

\begin{proof}
    We partition the subsets of $L$ of size $\Ssize$ into those that include $n$ and those that do not.

For the sums, we have following observations:
\begin{itemize}[leftmargin=*]
    \item $\Ssize \in \{0, 1, \cdots, |L|\}$, limited by the size of $L$. If $\Ssize > |L|$, we set the value to $0$ as it is meaningless.
    \item When $\Ssize = 0$, then subset $\Sset=\emptyset$ for any $L$. Thus $\alpha_{u,t}(L,0) = 1$ for $\forall L \subseteq \seedset\setminus \{ \seed \}$.
    \item When $\Ssize = 1$ and $|L|=1$, denote $L=\{v^{\prime}\}$ without loss of generality. The summation can also be divided into two parts: a subset including the only node $v^{\prime}$ and an empty subset. Then the decomposition equation holds:
    \begin{equation}
        \begin{aligned}
            & \sum_{\Sset \subseteq \{v^{\prime}\} \atop |\Sset|=1}[\prod_{w \in \neigh^{-}(u)\cap \Sset}(1-\actprob_{w,u})]\\
            &=[\sum_{\Sset \subseteq \emptyset \atop |\Sset|=0}1 \cdot (1-\actprob_{v^{\prime},u})] + \sum_{\Sset \subseteq \emptyset \atop |\Sset|=1}[\prod_{w \in \neigh^{-}(u)\cap \Sset}(1-\actprob_{w,u})]\\
            &=(1-\actprob_{v^{\prime},u})
    \end{aligned}
    \end{equation}
    %for $\forall v^{\prime} \in \seedset\setminus \{ \seed \}$ and $L = \{v^{\prime} \}$, $\alpha_{u,t}(L,1) = \prod_{w \in \neigh^{-}(u)\cap \{v^{\prime}\}}(1-\actprob_{w,u}) = (1-\actprob_{v^{\prime},u})$, which is the base case.
    \item When $|L| \geq 2$, for all $\Ssize = 1, \cdots, |L|$, the sum over possible subsets $\Sset$ can be decomposed as two parts: subsets $\Sset$ including $n$ and subsets $\Sset$ not including $n$, where $n$ denotes an arbitrary node in $L$: 
    \begin{equation}
        \begin{aligned}
            & \sum_{\Sset \subseteq L \atop |\Sset|=k}[\prod_{w \in \neigh^{-}(u)\cap \Sset}(1-\actprob_{w,u})]\\
            &=\sum_{\Sset \subseteq L \setminus \{n\} \atop |\Sset|=\Ssize-1}[\prod_{w \in \neigh^{-}(u)\cap \Sset}(1-\actprob_{w,u})\cdot (1-\actprob_{n,u})]\\
            &+ \sum_{\Sset \subseteq L \setminus \{n\} \atop |\Sset|=\Ssize}[\prod_{w \in \neigh^{-}(u)\cap \Sset}(1-\actprob_{w,u})]
    \end{aligned}
    \end{equation}

Moving the term $1-\actprob_{n,u}$ outside the summation, then we get the decomposition equation.
\end{itemize}
%Note that when $\Ssize = |L|$, the latter sum is meaningless so it is set to 0. It means that if $\Ssize=|L|$, every subset $\Sset$ consist the arbitrary node $n$.
\end{proof}

\begin{restatable}{lemma}{recursivegeneral}\label{lem:recursive_general}
    Given any seed $t \in \seedset$, any non-seed out-neighbor $u \in (\nodeset \setminus \seedset)\cap \neigh^{+}(t)$, and any $L \subseteq \seedset \setminus\{t\}$, all values $\alpha_{u,t}(L,k)$ for $k \in \{0,\ldots,|L|\}$ can be computed in $O(|L|^2)$ time.
\end{restatable}

\begin{proof}
    Using \Cref{lem:decompose_eq}, each $\alpha_{u,\seed}(L,k)$ can be decomposed recursively. Let $L = \{n_1, n_2, \ldots, n_{|L|} \}$ be an arbitrary ordering of the elements in $L$. We can compute the values $\alpha_{u,\seed}(L_i, k)$ for $L_i = \{n_1, \ldots, n_i\}$ and $k \in \{0, \ldots, i\}$ iteratively for $i = 0, \ldots, |L|$.

    For a given $i$, $\alpha_{u,\seed}(L_i, k)$ for all $k \in \{0, \ldots, i\}$ can be computed using the values computed for $L_{i-1}$ by \Cref{lem:decompose_eq}:
    \begin{equation}
        \alpha_{u,\seed}(L_i, k) =
    \begin{cases}
        1 & k = 0\\
        (1-\actprob_{n_i,u})\cdot \alpha_{u,\seed}(L_{i-1}, k-1) + \alpha_{u,\seed}(L_{i-1}, k) & 1 \leq k \leq i \\
        0 & \text{otherwise}
    \end{cases}
    \end{equation}
    %This process builds up a recursive tree till $\alpha_{u,\seed}(L, \Ssize)$.
    Filling this DP table for all $i=0,\ldots,|L|$ and $k=0,\ldots,i$ takes $O(|L|^2)$ operations since each entry is computed in $O(1)$ time.

    %Moreover, since we compute $\alpha_{u,\seed}(L, k)$ for all possible $k$ values,we can reuse the values of $\alpha_{u,\seed}(L_{i-1}, k)$ for the computing of both $\alpha_{u,\seed}(L_i, k)$ and $\alpha_{u,\seed}(L_i, k+1)$. Combining recursion trees for each possible values of $\Ssize$ together, we get a DP table with height $|L|$ and width $|L|+1$ in \Cref{fig:tree2}. Because each node in the DP table computes in $O(1)$ time, the total time to compute all values up to $L_{|L|}$ is $O(|L|^2)$.

% \begin{figure*}[!h]
%     \includegraphics[width=\textwidth]{figures/recursion_tree2_new.png}
%     \caption{Recursion Tree for all sums}\label{fig:tree2}
%     \end{figure*}
\end{proof}

\subsubsection{Optimization}
We first prove the decomposition property which supports the optimization:
\begin{restatable}{lemma}{localdecomp}\label{lem:local_decomp}
    Given $(G(\nodeset, E), \seedset, \actprob)$, the Shapley value of each seed $\seed \in \seedset$ can be computed as:
    \begin{equation}
    \shap(\seed) = \sum_{u\in \{\nodeset \setminus \seedset\} \cap \neigh^{+}(\seed)} \;\shap_u(\seed)
    \end{equation}
    %where $\shap_u(\seed)$ is the Shapley value of $\seed$ in the subgraph $G_u(V_u, E_u)$, where $V_u = \{u\} \cup \seedset(u)$ and $E_u = \{(v,u) \in E \mid v \in \seedset(u)\}$.
\end{restatable}

\begin{proof}
Under \onestep, whether $u$ is activated depends only on attempts along incoming edges to $u$ in one round; these events are independent across different $u$'s, so $\val$ decomposes as a sum of local value functions $\val_u$. Shapley values are linear in the value function, hence the second claim.
\end{proof}

\begin{proof}
    For each $u\in\nodeset\setminus\seedset$, define the indicator random variable
    \begin{equation}
    X_u(\Sset)\ :=\ \mathbbm{1}\{u \text{ is activated in one step from }\Sset\}.
    \end{equation}
    By Definition~\ref{def:value_function}, the value function is the expected
    number of activated non-seed nodes, hence by linearity of expectation
    \begin{equation}\label{eq:U-sum-ind}
    \val(\Sset)\ =\ \mathbb{E} \Big[\sum_{u\in\nodeset\setminus\seedset} X_u(\Sset)\Big]
    \ =\ \sum_{u\in\nodeset\setminus\seedset} \mathbb{E}[X_u(\Sset)].
    \end{equation}
    
    For a fixed $u$, consider the subgraph $G_u(V_u, E_u)$ induced by $u$ and $\seedset(u)$. The value function of this subgraph depends on the only non-seed node $u$ and for $\seed \in \Sset \setminus seedset (u)$, it has no chance to activate $u$ under \onestep:
    \begin{equation}
    \val_u(\Sset)= \mathbb{E} [X_u(\seedset(u) \cap \Sset)] = \mathbb{E} [X_u(\Sset)]
    \end{equation}

    Therefore, the value function of the entire graph decomposes as a sum of the value functions of the subgraphs:
    \begin{equation}
    \val(\Sset)\ =\ \sum_{u\in\nodeset\setminus\seedset} \val_u(\Sset).
    \end{equation}

    Shapley values are linear in the value function, hence we can exchange the order of summation and get:
    \begin{equation}
        \begin{aligned}
            \shap(\seed) &= \frac{1}{|\seedset|!}\sum_{\pi\in\Pi(\seedset)} \Big(\val(S_{\pi,\seed}\cup\{\seed\})-\val(S_{\pi,\seed})\Big)\\
            &= \sum_{u\in(\nodeset\setminus\seedset)\cap\neigh^{+}(\seedset)} \frac{1}{|\seedset|!}\sum_{\pi\in\Pi(\seedset)} \Big(\val_u(S_{\pi,\seed}\cup\{\seed\})-\val_u(S_{\pi,\seed})\Big)\\
            &= \sum_{u\in(\nodeset\setminus\seedset)\cap\neigh^{+}(\seedset)} \shap_u(\seed)
        \end{aligned}
    \end{equation}
    where $\shap_u(\seed)$ is the Shapley value of $\seed$ in the local subgraph $G_u(V_u, E_u)$. Finally, $\val_u(\cdot)$ depends only on $\neigh^{-}(u)\cap\seedset$; hence if $\seed\notin\neigh^{-}(u)$, then $\val_u(S\cup\{\seed\})=\val_u(S)$ for all $S$ and $\shap_u(\seed)=0$. Thus the sum can be restricted to
    $u\in(\nodeset\setminus\seedset)\cap\neigh^{+}(\seed)$, which proves the lemma.
    \end{proof}

\subsubsection{Completing the Proof of Theorem~\ref{thm:single_poly}}\label{subsubsec:complete_proof_single_poly}
\begin{proof}[Proof of Theorem~\ref{thm:single_poly}]
We establish both correctness and polynomial-time complexity of $\dpsingleopt$ (\Cref{alg:single_step_main_opt}).

\textit{Correctness:} 
By Lemma~\ref{lem:non_zero_marginal}, the marginal contribution under \onestep\ has the simplified form in Eq.~\eqref{eq:marginval_simplified_appendix}. Substituting this into the Shapley definition and grouping coalitions by size yields the closed-form expression in Lemma~\ref{lem:shapley_single_step} (Eq.~\eqref{eq:shapley_sum}). 

Algorithm~\ref{alg:single_step_main_opt} implements this expression as follows. For each non-seed node $u$, Line~\ref{line:build_subgraph} constructs the local subgraph $G_u(V_u, E_u)$ with $V_u = \seedset(u) \cup \{u\}$, and Line~\ref{line:compute_coef} precomputes the binomial coefficients $C[\Ssize] = \frac{\Ssize!(|\seedset(u)|-\Ssize-1)!}{|\seedset(u)|!}$. For each seed $t \in \seedset(u)$, Line~\ref{line:compute_alpha} computes the array $\{\alpha_{u,t}[\Ssize]\}_{\Ssize=0}^{|\seedset(u)|-1}$ via the dynamic programming recurrence established in Lemma~\ref{lem:decompose_eq}, and Line~\ref{line:compute_local_shapley} computes the local Shapley value $\shap_u(t)$ by evaluating Eq.~\eqref{eq:shapley_sum} restricted to subgraph $G_u$. Finally, by Lemma~\ref{lem:local_decomp}, Line~\ref{line:update_shapley_global} correctly aggregates all local Shapley values to obtain the global Shapley values via $\shap(\seed) = \sum_{u \in (\nodeset \setminus \seedset) \cap \neigh^{+}(\seed)} \shap_u(\seed)$.

\textit{Time Complexity:}
There are at most $|\nodeset| - |\seedset|$ non-seed nodes, and hence at most $|\nodeset| - |\seedset|$ subgraphs to process in Lines~\ref{line:iterate_u}--\ref{line:check_nonempty}. For each subgraph $G_u$, the coefficient computation in Line~\ref{line:compute_coef} takes $O(|\seedset(u)|) \subseteq O(|\seedset|)$ time. For each seed $t \in \seedset(u)$, computing the DP array $\alpha_{u,t}[\cdot]$ in Line~\ref{line:compute_alpha} takes $O(|\seedset(u)|^2)$ time by Lemma~\ref{lem:recursive_general}, and the dot product in Line~\ref{line:compute_local_shapley} takes $O(|\seedset(u)|)$ time. Since $|\seedset(u)| \leq |\seedset|$, processing all seeds in $\seedset(u)$ takes $O(|\seedset(u)| \cdot |\seedset(u)|^2) = O(|\seedset(u)|^3) \subseteq O(|\seedset|^3)$ time per subgraph. Summing over all subgraphs, the total time complexity is $O\bigl((|\nodeset| - |\seedset|) \cdot |\seedset|^3\bigr) = O(|\nodeset| \cdot |\seedset|^3)$. 
%Since $|\seedset| \leq |\nodeset|$, this simplifies to $O(|\nodeset|^4)$ in the worst case.
\end{proof}
\section{Details From Section~\ref{sec:hardness}}\label{sec:appendix_fixedstep}
This section provides the complete proof of the following \Cref{thm:hardness}:

\fixedstephardness*

% \begin{restatable}{lemma}{twostephardness}\label{lem:2step_hardness}
%     Computing the Shapley value for any seed node in the $2$-step termination is $\#P$-hard.%, even when the value function $U(S)$ can be computed in polynomial time.
% \end{restatable}

%\proof
We establish the hardness result for \fixedstep termination with $\timeconst = 2$ steps via a reduction from the $\#P$-complete problem of counting satisfying assignments for a monotone 2-CNF formula~\cite{valiant1979complexity}. In this graph, the maximum directed path-length is 2 and hence the hardness extends to \fullstep termination and \fixedstep termination with $\timeconst \ge 2$ steps. %this finding can be extended to any \fixedstep termination with $\timeconst \ge 2$ steps and the \fullstep termination. 

Let $\phi$ be a monotone 2-CNF formula with $n$ variables $x_1,\dots,x_n$ and $m$ clauses $C_1,\dots,C_m$. Each clause $C_j$ is a disjunction of two unnegated variables. For each clause $C_j$, define $\var(C_j) \;=\; \{\,x_i \mid x_i \text{ appears in } C_j\}$. Denote the number of satisfying assignments of $\phi$ by $\#_{\phi}$. The proof proceeds in the following two main parts.

%Graph Construction and 
\subsection{\textbf{Shapley Value Expression}}\label{subsec:graph_construction}
%For each pair $(r,q)$ with $r\in\{1,\ldots,n\}$ and $q\in\{1,\ldots,m\}$, we build a graph instance $G_{(r,q)}=(V_{(r,q)},E_{(r,q)})$ in the 2-step activation model. We fix one special seed node $\seed_1$ and derive its Shapley value in each graph instance, denoted by $\shap_{(r,q)}(\seed_1)$ (see \Cref{lemma:shap_t1}).
%\paragraph{\textbf{Graph Construction and Shapley Values}}
\textbf{Marginal contribution of $\seed_1$.~} In \Cref{sec:hardness} we have explained the construction of each graph $G_{(r,q)}$. We now show an key lemma of the marginal contribution of the fixed seed node $\seed_1$ in each instance $G_{(r,q)}$.% which establishes the expression of the Shapley value in \Cref{eq:shap_t1}.

\begin{lemma}\label{lemma:marginal_t1}
For each graph $G_{(r,q)}$ where $(r,q)\in [n] \times [m]$, let $\Sset \subseteq U \cup \seedset_r \setminus \{\seed_1\}$ be any subset of seed nodes. The marginal contribution of $\seed_1$ to $\Sset$ is:
\begin{equation}\label{eq:marginal_t1}
\begin{aligned}
    &\val(\Sset \cup \seed_1) - \val(\Sset) \\
    &= 
    \begin{cases}
    1 & \text{if } \Sset = \emptyset\\
    (\frac{1}{q+1})^{\sum_{v_j \in V}  \mathbbm{1}[\exists u_i \in \Sset \text{ s.t. } (u_i, v_j) \in E_{(p,q)}]} & \text{if } \Sset \subseteq U \text{ and } \Sset \not= \emptyset\\
    %(\frac{1}{q+1})^{\sum_{v_j \in V}  \mathbbm{1}[\exists u_i \in \Sset \text{ such that } x_i \in \var(C_j)]} & \text{if } \Sset \subseteq U \text{ and } \Sset \not= \emptyset\\
    0 & \text{otherwise}
    \end{cases}
\end{aligned}
\end{equation}
\end{lemma}

\begin{proof}
The marginal contribution of $\seed_1$ is $\val(\Sset \cup \{\seed_1\}) - \val(\Sset)$. By construction, $\seed_1$ has no path to any clause node $v \in V$. Thus, its addition only affects the sink node $s$. The marginal contribution simplifies to the change in activation probability of $s$:
\begin{small}
\begin{equation}
\val(\Sset \cup \{\seed_1\}) - \val(\Sset) = \Pr[s \text{ is active} | \Sset \cup \{\seed_1\}] - \Pr[s \text{ is active} | \Sset]
\end{equation}
\end{small}

Due to the edge $(\seed_1, s)$ having probability 1, $\Pr[s \text{ is active} | \Sset \cup \{\seed_1\}] = 1$ for any $\Sset$. The expression becomes $1 - \Pr[s \text{ is active} | \Sset]$. We consider three cases for the subset $\Sset$:
\begin{enumerate}[leftmargin=*]
    \item If $\Sset$ contains any auxiliary seed $\seed_i \in \seedset_r \setminus \{\seed_1\}$, then $s$ is already activated with probability 1 via the edge $(\seed_i, s)$. Thus, $\Pr[s \text{ is active} | \Sset] = 1$, and the marginal contribution is $1 - 1 = 0$.
    \item If $\Sset = \emptyset$, no node is initially active, so $\Pr[s \text{ is active} | \emptyset] = 0$. The marginal contribution is $1 - 0 = 1$.

    \item If $\Sset \subseteq U$ is non-empty, $s$ is activated only via clause nodes. The probability that $s$ is \emph{not} activated by $\Sset$ is the probability that for every clause node $v_j$ activated by $\Sset$, the edge $(v_j, s)$ fails. This is $(\frac{1}{q+1})^{\sum_{v_j \in V} \mathbbm{1}[v_j \text{ is activated by } \Sset]}$.
    
    Since a clause node $v_j$ is activated if and only if $\exists u_i \in \Sset$ such that $(u_i, v_j) \in E_{(r,q)}$, we have $\Pr[s \text{ is active} | \Sset] = 1 - (\frac{1}{q+1})^{\sum_{v_j \in V} \mathbbm{1}[\exists u_i \in \Sset \text{ s.t. } (u_i, v_j) \in E_{(r,q)}]}$.
    Thus we obtain the marginal contribution as \Cref{eq:marginal_t1}.
\end{enumerate}
These three cases together establish the lemma.
\end{proof}

\textbf{Shapley value of $\seed_1$.~} Then we establish the expression of the Shapley value in \Cref{eq:shap_t1} of \Cref{sec:hardness} in the following lemma.

%Now that we have derived the marginal contributions of $\seed_1$ for different subsets, we can compute its Shapley value. For each instance $G_{(r,q)}$, denote the Shapley value of $\seed_1$ as $\shap_{(r,q)}(\seed_1)$.

\begin{lemma}\label{lemma:shap_t1}
    For each graph $G_{(r,q)}$, where $r \in \{1,\cdots, n\}$ and $q \in \{1, \cdots, m\}$, the Shapley value of $\seed_1$ is:
    \begin{equation}\label{eq:shap_t1_appendix}
    \begin{aligned}
    &\shap_{(r,q)}(\seed_1) = \frac{1}{n+r} \\
    &+ \sum_{k=1}^{n} \frac{k!(n+r-k-1)!}{(n+r)!} \cdot  \sum_{S \subseteq U \atop |S|=k} (\frac{1}{q+1})^{\sum_{v_j\in V}\mathbbm{1}[\exists u_i \in S \text{ s.t. } (u_i, v_j) \in E_{(r,q)}]}
    \end{aligned}
    \end{equation}
\end{lemma}

\begin{proof}
By the definition of Shapley value in \Cref{def:shapley_influence}, we have:
\begin{equation}
    \begin{aligned}
    &\shap_{(r,q)}(\seed_1) \\
    =& \sum_{S \subseteq U\cup \seedset_r \setminus \{\seed_1\}} \frac{|S|!(n+r-|S|-1)!}{(n+r)!} \cdot [\val(S \cup \seed_1) - \val(S)] \\
%    &= \sum_{k=0}^{n+p} \frac{k!(n+p-k)!}{(n+p+1)!} \cdot \sum_{S \subseteq U \cup \seedset_p \atop |S|=k} [\val(S \cup \seed_1) - \val(S)]\\
    %&= \sum_{k=0}^{n+r} \frac{k!(n+r-k)!}{(n+r+1)!} \cdot  \sum_{S \subseteq U \atop |S|=k} (\frac{1}{\lambda_r})^{\text{number of clauses satisfied by variables in } S}
    \end{aligned}
\end{equation} 
%We first analyze the marginal contributions of $\seed_1$ and simplify the expression of the Shapley value of $\seed_1$ in each instance $G_{(p,q)}$.
%To relate the Shapley value to counting satisfying assignments, we analyze the marginal contributions of $t_1$ and connect it to the number of satisfied clauses by varible nodes.

First, we reorganize the sum by grouping subsets $\Sset$ by their size:
\begin{equation}
\begin{aligned}
&\shap_{(r,q)}(\seed_1) \\
&= \sum_{k=0}^{n+r-1} \frac{k!(n+r-k-1)!}{(n+r)!} \cdot \sum_{\Sset \subseteq U \cup \seedset_r \atop |\Sset|=k} [\val(\Sset \cup\{\seed_1\}) - \val(\Sset)]\\
\end{aligned}
\end{equation} 

Then we substitute the marginal contributions from \Cref{lemma:marginal_t1} for different subsets $\Sset$:
\begin{itemize}[leftmargin=*]
    \item When $k=0$: The empty set contributes $\frac{0!(n+r-1)!}{(n+r)!} = \frac{1}{n+r}$
    \item When $1 \leq k \leq n$, we consider two subcases:
    \begin{itemize}[leftmargin=*]
        \item For subsets $\Sset \subseteq U$ (containing only variable nodes), the marginal contribution is $(\frac{1}{q+1})^{\sum_{v_j \in V}  \mathbbm{1}[\exists u_i \in \Sset \text{ s.t. } (u_i, v_j) \in E_{(r,q)}]}$
        \item Otherwise, the marginal contribution is $0$.
    \end{itemize}
    \item When $k > n$: Since there are only $n$ variable nodes, all such sets contain at least one additional seednode from $\seedset_r \setminus \{\seed_1\}$. Therefore, their marginal contribution is $0$.
\end{itemize}
Collecting all the terms, we yield \Cref{eq:shap_t1_appendix}.
\end{proof}

\subsection{\textbf{Computing the Number of Satisfying Assignments}}\label{subsec:recovery_sat}
We show that if one can compute $\shap_{(r,q)}(\seed_1)$ for every pair $(r,q)$ in polynomial time, then $\#_{\phi}$ can be obtained in polynomial time. The argument proceeds in following steps.

\textbf{Step 1: Express $\#_{\phi}$ in terms of $\subcnt_{k,c}$.~} We begin by noting that there exists a bijection between subsets of variables and assignments to the variables in $\phi$:
\begin{observation}\label{obs:subsets_assignments}
    For a monotone 2-CNF formula $\phi$, there is a one-to-one correspondence between subsets of variables $\{x_1,\dots,x_n\}$ and assignments to the variables. Specifically:
    \begin{enumerate}[leftmargin=*]
        \item Every subset of variables uniquely maps to an assignment that setting variables in this subset to $\textsc{True}$ and other variables to $\textsc{False}$. 
        \item Conversely, any valid assignment to $\phi$ can be viewed as picking out those variables that are set to \textsc{True} to form a subset of variables.
    \end{enumerate}
    \end{observation}

    By \Cref{obs:subsets_assignments}, we say that a subset of variables \emph{satisfies} $c$ clauses of $\phi$ if the corresponding assignment (setting variables in this subset to \textsc{True} and other variables to \textsc{False}) makes exactly $c$ clauses \textsc{True}.

    Therefore, we define $\subcnt_{k,c}$ as the number of size-$k$ subsets of variables that satisfies exactly $c$ clauses in $\phi$, where $c\in\{0,1,\dots,m\}$ and $k\in\{0,1,\dots,n\}$. Note that since $\phi$ is a monotone 2-CNF formula, all literals are unnegated, thus all variables are set to $\textsc{False}$ ($k=0$) is equivalent to no clauses are satisfied ($c=0$). Therefore, we have $\subcnt_{0,0}=1$,  $\subcnt_{k,0}=0$ for all $k\ge 1$ and $\subcnt_{0,c}=0$ for all $c\ge 1$.  
    Then we can express $\#_{\phi}$ using $\subcnt_{k,c}$:
    \begin{observation}\label{obs:satisfying_assignments}
        The number of satisfying assignments $\#_{\phi}$ equals the sum of size-$k$ subsets of variables that satisfy all $m$ clauses over all possible subset sizes $k$:
        \begin{equation}\label{eq:satisfying_assignments}
            \#_{\phi} = \sum_{k=1}^{n}\subcnt_{k,m}
        \end{equation}
    \end{observation}
Thus, if we can compute $\subcnt_{k,c}$ for all $c\in\{1,\dots,m\}$ and $k\in\{1,\dots,n\}$, we can obtain $\#_{\phi}$ in polynomial time.

\textbf{Step 2: Express $\shap_{(r,q)}(\seed_1)$ in terms of $\subcnt_{k,c}$.~} 
We relate the subsets of variable nodes in our construction to subsets of variables in $\phi$ (see \Cref{claim:clause_satisfaction}), showing that $\shap_{(r,q)}(\seed_1)$ can be rewritten using $\subcnt_{k,m}$ (see \Cref{claim:shapley_nck}). So far, we connect the counting of satisfying assignments to the Shapley value of $\seed_1$ using $\subcnt_{k,m}$.
\begin{claim}\label{claim:clause_satisfaction}
    Given a subset of variable nodes $\Sset \subseteq U$, for any clause node $v_j \in V$:
    \begin{equation}\label{eq:clause_satisfaction}
        \mathbbm{1}[\exists u_i \in \Sset \text{ s.t. } (u_i, v_j) \in E_{(r,q)}] = \mathbbm{1}[\exists x_i \in \var(\Sset) \text{ s.t. } x_i \in \var(C_j)]
        %\mathbbm{1}[C_j \text{ is satisfied by } \var(\Sset)]
    \end{equation}
    where $\var(\Sset)$ denotes the set of variables corresponding to variable nodes in $\Sset$ for each subset $S \subseteq U$. 
    
    %$C_j$ is satisfied by $\var(\Sset)$.
    % \begin{equation}
    %     \mathbbm{1}[\exists u_i \in \Sset \text{ s.t. } (u_i, v_j) \in E_{(p,q)}] = \mathbbm{1}[\exists x_i \in \var(\Sset) \text{ s.t. } x_i \in \var(C_j)] = \mathbbm{1}[C_j \text{ is satisfied by } \var(\Sset)]
    % \end{equation}
\end{claim}

\begin{proof}
    By the construction of each $G_{(r,q)}$, an edge $(u_i, v_j)$ exists if and only if variable $x_i$ appears in clause $C_j$ in $\phi$. Therefore, $\exists u_i \in \Sset$ such that $(u_i, v_j) \in E_{(r,q)}$ if and only if $\exists x_i \in \var(\Sset)$ such that $x_i \in \var(C_j)$.
\end{proof}    

    Using \Cref{claim:clause_satisfaction}, the Shapley value of $\seed_1$ in each $G_{(r,q)}$ can be reduced to \Cref{eq:shapley_nck}, in terms of $\subcnt_{k,c}$:

    \begin{claim}\label{claim:shapley_nck}
        The Shapley value of $\seed_1$ in each instance $G_{(r,q)}$ can be expressed as:
        \begin{equation}\label{eq:shapley_nck}
            \begin{aligned}
            &\shap_{(r,q)}(\seed_1) = \frac{1}{n+r} \\
            & +\sum_{k=1}^{n} \sum_{c=1}^{m} \frac{k!(n+r-k-1)!}{(n+r)!} \cdot (\frac{1}{q+1})^{c} \cdot \subcnt_{k,c}
            \end{aligned}
        \end{equation}
    \end{claim}
    \begin{proof}
        By \Cref{claim:clause_satisfaction}, the inner sum of $\shap_{(r,q)}(\seed_1)$ in each instance $G_{(r,q)}$ can be rewrite as:
        \begin{equation}
            \begin{aligned}
            &\sum_{S \subseteq U \atop |S|=k} (\frac{1}{q+1})^{\sum_{j} \mathbbm{1}[\exists u_i \in S \text{ s.t. } (u_i, v_j) \in E_{(r,q)}]}\\
            &= \sum_{\var(\Sset) \subseteq X \atop |\var(\Sset)|=k} (\frac{1}{q+1})^{\sum_{j} \mathbbm{1}[\exists x_i \in \var(\Sset) \text{ s.t. } x_i \in \var(C_j)]}
            \end{aligned}
        \end{equation}
        where $X$ denotes the set of variables $\{x_1, \dots, x_n\}$ in $\phi$.
    
        Since a clause $C_j$ in $\phi$ is satisfied if and only if at least one of its variables is set to \textsc{True}. Therefore, having $\exists x_i \in \var(S) \text{ s.t. } x_i \in \var(C_j)$ precisely captures when clause $C_j$ is satisfied by the assignment where variables in $\var(S)$ are set to \textsc{True} and remaining variables are set to \textsc{False}. Thus:
        \begin{equation}
            \begin{aligned}
            &\sum_{j} \mathbbm{1}[\exists x_i \in \var(S) \text{ s.t. } x_i \in \var(C_j)] \\
            & = \sum_{j} \mathbbm{1}[C_j \text{ is satisfied by } \var(S)]\\
            %&= |\{C_j \in \phi : \exists x_i \in \var(S) \text{ s.t. } x_i \in \var(C_j)\}|\\
            & = |\{C_j \in \phi: C_j \text{ is satisfied by } \var(S)\}|
            \end{aligned}
        \end{equation}
        which counts the number of clauses satisfied by the variables in $\var(\Sset)$.
    
        For any fixed size $k\in [1,n]$, we can partition these set of variables $\var(\Sset)$ with $|\var(\Sset)|=k$ according to the number of clauses they satisfy.  Note that since at least one variable is set to \textsc{True} ($k\geq 1$), at least one clause is satisfied. Then we can simplify the inner sum as:
        \begin{equation}\label{eq:shapley_nck_inner}
        \begin{aligned}
            &\sum_{\var(\Sset) \subseteq X \atop |\var(\Sset)|=k} (\frac{1}{q+1})^{|\{C_j \in \phi: C_j \text{ is satisfied by } \var(S)\}|}\\
            &= \sum_{c=1}^m \sum_{\var(\Sset) \subseteq X,|\var(\Sset)|=k \atop |\{C_j \in \phi: C_j \text{ is satisfied by } \var(S)\}| = c} (\frac{1}{q+1})^c \\
            %&= \sum_{c=1}^m \sum_{S \subseteq U,|S|=k \atop |\{C_j : C_j \text{ satisfied by } \var(S)\}| = c} (\frac{1}{q+1})^c \\
            &= \sum_{c=1}^{m} \subcnt_{k,c} (\frac{1}{q+1})^{c}
        \end{aligned}
        \end{equation}
        %where $\subcnt_{c,k}$ is the number of size-$k$ subsets of variables that satisfy exactly $c$ clauses.

        Substituting \Cref{eq:shapley_nck_inner} back into \Cref{eq:shap_t1_appendix} from \Cref{lemma:shap_t1}, we can rewrite the Shapley value expression as:
        \begin{equation}
            \begin{aligned}
            &\shap_{(r,q)}(\seed_1) = \frac{1}{n+r} \\
            &+ \sum_{k=1}^{n} \frac{k!(n+r-k-1)!}{(n+r)!} \cdot \sum_{c=1}^{m} \subcnt_{k,c} (\frac{1}{q+1})^{c}
            \end{aligned}
        \end{equation}
    Rearranging the terms, we obtain \Cref{eq:shapley_nck}.
    \end{proof}

\textbf{Step 3: Solve for $\subcnt_{k,c}$ via a matrix system.~} For each pair of parameters $(r,q)$, we obtain the $\shap_{(r,q)}(\seed_1)$ of the form \Cref{eq:shapley_nck}. We collect all $m\cdot n$ equations into a matrix system:
\begin{equation}\label{eq:system_linear}
    A \times \subcntmat \times B = D
\end{equation}
where:
\begin{itemize} 
    \item $A$ is an $n \times n$ matrix with entries $a_{p,k} = \frac{k!(n+p-k-1)!}{(n+p)!}$ for $k,p \in [1,n]$:
        \begin{equation}\label{eq:matrix_A}
        A = \begin{pmatrix}
            \frac{1!(n-1)!}{(n+1)!} & \frac{2!(n-2)!}{(n+1)!} & \cdots & \frac{n!0!}{(n+1)!} \\  
            \frac{1!(n)!}{(n+2)!} & \frac{2!(n-1)!}{(n+2)!} & \cdots & \frac{n!1!}{(n+2)!} \\  
            \vdots & \vdots & \ddots & \vdots \\
            \frac{1!(2n-2)!}{(2n)!} & \frac{2!(2n-3)!}{(2n)!} & \cdots & \frac{n!(n-1)!}{(2n)!}
            \end{pmatrix}
        \end{equation}
    \item $B$ is an $m \times m$ matrix with entries $b_{c,q} = (\frac{1}{q+1})^{c}$ for $c,q \in [1,m]$:
    \begin{equation}\label{eq:matrix_B}
        B =\begin{pmatrix}
        (\frac{1}{2})^{1} & (\frac{1}{3})^{1} & \cdots & (\frac{1}{m+1})^{1} \\
        \vdots & \vdots & \ddots & \vdots \\
        (\frac{1}{2})^{m} & (\frac{1}{3})^{m} & \cdots & (\frac{1}{m+1})^{m}
        \end{pmatrix}
    \end{equation}
    \item $D$ is an $n \times m$ matrix with entries $d_{r,q} = \shap_{(r,q)}(\seed_1) - \frac{1}{n+r}$ for $r \in [1,n]$ and $q \in [1,m]$:
        \begin{equation}\label{eq:matrix_D}
            D = \begin{pmatrix}
            \shap_{(1,1)}(\seed_1) - \frac{1}{n+1} & \cdots & \shap_{(1,m)}(\seed_1) - \frac{1}{n+1} \\
            \vdots & \ddots & \vdots \\
            \shap_{(n,1)}(\seed_1) - \frac{1}{2n} & \cdots & \shap_{(n,m)}(\seed_1) - \frac{1}{2n}
            \end{pmatrix}
        \end{equation}
    \item $\subcntmat$ is an $n \times m$ matrix with each entry is $\subcnt_{k,c}$ for $k \in [1,n]$ and $c \in [1,m]$:
    \begin{equation}\label{eq:matrix_subcnt}
        \subcntmat = \begin{pmatrix}
        \subcnt_{1,1} & \subcnt_{1,2} & \cdots & \subcnt_{1,m} \\
        \vdots & \vdots & \ddots & \vdots \\
        \subcnt_{n,1} & \subcnt_{n,2} & \cdots & \subcnt_{n,m}
        \end{pmatrix}
    \end{equation}
\end{itemize}
Now, each $\shap_{(r,q)}(\seed_1)$ can be expressed as:
    \begin{equation}
        \begin{aligned}
            &\shap_{(r,q)}(\seed_1)- \frac{1}{n+r} \\
            &= d_{r,q} \\
            &=\sum_{k=1}^{n} \sum_{c=1}^{m} a_{r,k} \cdot \subcnt_{k,c} \cdot b_{c,q}
        \end{aligned}
    \end{equation}

$A$ and $B$ are known coefficient matrices, $D$ is known if there exists an oracle for computing the Shapley value, and $\subcntmat$ is the matrix of unknown variables $\subcnt_{k,c}$. Then we will show that $A$ and $B$ are non-singular, and thus we can solve for $\subcntmat$ in polynomial time.
\begin{claim}\label{claim:non_singular}
    Both matrices $A$ (\Cref{eq:matrix_A}) and $B$ (\Cref{eq:matrix_B}) are non-singular.
\end{claim}
\begin{proof}
    For the coefficient matrix $A$ in \Cref{eq:matrix_A}, we multiply each row $i \in [1,n]$ by $(n+i)!$; divide each column $j \in [1,n]$ by $j!$; and reverse the order of columns. Through these transformations, we get the following matrix denoted as $A^{\prime}$:
\begin{equation}
    \begin{split}
   A^{\prime} =& \begin{pmatrix}
   0! & 1! & \cdots & (n-1)! \\
   \vdots & \vdots & \ddots & \vdots \\
   (n-1)! & n! & \cdots & (2n-2)!
   \end{pmatrix}
   \end{split}
\end{equation}
The entries of matrix $A^{\prime}$ is $a^{\prime}_{i,j} = (i-1+j-1)!$, and then 
its determinant is $\det(A^{\prime}) = \prod_{i=0}^{n-1} i!i! \neq 0$, thus $A^{\prime}$ is non-singular ~\cite{bacher2002determinants}. Because these algebraic manipulations preserve non-singularity, $A$ is also non-singular.

For the coefficient matrix $B$ in \Cref{eq:matrix_B}, we take the transpose of $B$:
\begin{equation}
    B^\top = \begin{pmatrix}
    (\frac{1}{2})^{1} & (\frac{1}{2})^{2} & \cdots & (\frac{1}{2})^{m} \\
    \vdots & \vdots & \ddots & \vdots \\
    (\frac{1}{m+1})^{1} & (\frac{1}{m+1})^{2} & \cdots & (\frac{1}{m+1})^{m}
    \end{pmatrix}
\end{equation}
Divide each row by $(\frac{1}{i+1})$, we get a Vandermonde matrix denoted as $(B^\top)^{\prime}$:
\begin{equation}
    {B^\top}^{\prime} = \begin{pmatrix}
    1 & \frac{1}{2} & \cdots & (\frac{1}{2})^{m-1} \\
    1 &\frac{1}{3} & \cdots & (\frac{1}{3})^{m-1} \\
    \vdots & \vdots & \ddots & \vdots \\
    1 & \frac{1}{m+1} & \cdots & (\frac{1}{m+1})^{m-1}
    \end{pmatrix}
\end{equation}
Matrix $(B^\top)^{\prime}$ is a Vandermonde matrix, so the determinant is $\det((B^\top)^{\prime}) = \prod_{1 \leq i < j \leq m} (\frac{1}{j+1} - \frac{1}{i+1}) \neq 0$ and it is a non-singular matrix. Thus $\det(B) = \det(B^\top) = \det((B^\top)^{\prime}) \cdot \prod_{i=1}^{m} \frac{1}{i+1} \neq 0$, implying that $B$ is non-singular.
\end{proof}

Finally, we show that we can solve the matrix system \Cref{eq:system_linear} in polynomial time and thus compute the total number of satisfying assignments $\#_{\phi}$:
\begin{claim}\label{claim:solve_subcntmat}
    If we can compute $\shap_{(r,q)}(\seed_1)$ for each $G_{(r,q)}$ in polynomial time, we can compute the number of satisfying assignments $\#_{\phi}$ in polynomial time.
\end{claim}
\begin{proof}
By \Cref{claim:non_singular}, we can solve all $\subcnt_{k,c}$ by:
\begin{equation}\label{eq:solve_subcntmat}
    \subcntmat = A^{-1} \times D \times B^{-1}
\end{equation}
If we can compute $\shap_{(p,q)}(\seed_1)$ for each $G_{(p,q)}$ in polynomial time, we can compute all $\subcnt_{k,c}$ in polynomial time through \Cref{eq:solve_subcntmat}. Then we can sum $\subcnt_{k,m}$ over all $k$ to compute $\#_{\phi}$ in polynomial time.
\end{proof}

\section{Details of $\permuteMC$}\label{sec:permute_mc_appendix}
This section provides precise pseudocode, proofs of unbiasedness and additive error guarantees, and runtime analysis for the Monte–Carlo permutation estimator $\permuteMC$.

\subsection{\textbf{Pseudocode}}
We present the algorithm $\permuteMC$ in \Cref{alg:approx_perm_mc}. The number of permutations $n_{\pi}$ and the number of Monte Carlo simulations $n_{\text{MC}}$ per value estimation are determined by \Cref{prop:permute_mc_property}, based on the approximation parameters $\epsilon$ and $\delta$ (Lines~\ref{line:input_n_pi}--\ref{line:input_n_mc}). The algorithm samples $n_{\pi}$ random permutations of $\seedset$ (Line~\ref{line:sample_perm}). For each permutation $\pi_i$, it iterates over the seed nodes in the order defined by $\pi_i$ (Line~\ref{line:seed_loop}), maintaining a set $S$ of previously visited seeds (initialized to $\emptyset$ in Line~\ref{line:init_S}). At each step, the marginal contribution of the current seed $\seed_j$ is estimated as $\valsim(S \cup \{\seed_j\}) - \valsim(S)$, where $\valsim(\cdot)$ denotes the average influence computed over $n_{\text{MC}}$ simulations (Lines~\ref{line:val_union}--\ref{line:val_S}). These marginal contributions are accumulated in $\textit{est}[\seed_j]$ (Line~\ref{line:update_est}), and $S$ is updated accordingly (Line~\ref{line:update_S}). Finally, the estimated Shapley value for each seed is obtained by averaging over all $n_{\pi}$ permutations (Line~\ref{line:return_permute_mc}).

\begin{algorithm}[!htbp]
    \caption{$\permuteMC(G, \seedset, \actprob, n_{\pi}, n_{\text{MC}})$}
    \label{alg:approx_perm_mc}
    \begin{algorithmic}[1]
    \Require Network $G(V,E)$, seed set $\seedset$, activation probabilities $\actprob$, approximation parameters \label{line:input}
   
    \Ensure Estimated Shapley values $\widehat{\shap(\seed)}$ for all $\seed \in \seedset$ \label{line:output}
    \State Number of permutation of $\seedset$: $n_{\pi} \gets O\left( \frac{|V|^2}{\epsilon^2} \ln (\frac{|\seedset|}{\delta}) \right)$ \label{line:input_n_pi}
    \State Number of simulations per value estimation: $n_{\text{MC}} \gets O\left( \frac{|V|^2}{\epsilon^2} \ln (\frac{|V|^2 |\seedset|^2}{\epsilon^2 \delta}) \right)$ \label{line:input_n_mc}
    \State Initialize $\textit{est}[\seed] \gets 0$ for all $\seed \in \seedset$ \label{line:init}
    \For{$i = 1$ to $n_{\pi}$} \label{line:perm_loop}
        \State Sample a random permutation $\pi_i$ of $\seedset$ \label{line:sample_perm}
        \State Initialize $S \gets \emptyset$ \label{line:init_S}
        \For{each $\seed_j$ in $\pi_i$} \label{line:seed_loop}
            \State $\hat{\val}(S \cup \{\seed_j\}) \gets$ average influence of $S \cup \{\seed_j\}$ from $n_{\text{MC}}$ simulations \label{line:val_union}
            \State $\hat{\val}(S) \gets$ average influence of $S$ from $n_{\text{MC}}$ simulations \label{line:val_S}
            \State $\textit{est}[\seed_j] \mathrel{+}= \hat{\val}(S \cup \{\seed_j\}) - \hat{\val}(S)$ \label{line:update_est}
            \State $S \gets S \cup \{\seed_j\}$ \label{line:update_S}
        \EndFor \label{line:end_seed_loop}
    \EndFor \label{line:end_perm_loop}
    \State \Return $\widehat{\shap(\seed)} \gets \frac{1}{n_{\pi}} \textit{est}[\seed]$ for all $\seed \in \seedset$ \label{line:return_permute_mc}
    \end{algorithmic}
    \end{algorithm}

\subsection{Theoretical guarantees}
Using Hoeffding's inequality~\cite{hoeffding1994probability}, we establish the following theoretical performance guarantees on both accuracy %(\Cref{prop:permute_mc_property}) 
and efficiency.

\begin{restatable}{proposition}{permutemcproperty}\label{prop:permute_mc_property}
Given $(G(\nodeset, E), \seedset, \actprob)$, $\permuteMC$ returns $\widehat{\shap(\seed)}$ for all $\seed \in \seedset$ such that: (1) $\mathbb{E}[\widehat{\shap(\seed)}] = \shap(\seed)$, $\forall \seed \in \seedset$; (2) for any $\epsilon > 0$ and $\delta \in (0,1)$, with probability at least $1 - \delta$, $|\widehat{\shap(\seed)} - \shap(\seed)| \le \epsilon$ holds for all $\seed \in \seedset$; (3) the running time is $O\left( \frac{|V|^4}{\epsilon^4} \ln( \frac{|\seedset|}{\delta} + \frac{|V|^2 |\seedset|^2}{\epsilon^2 \delta}) \cdot |\seedset| \cdot |E| \right)$. 
\end{restatable}

\proof
    The proof consists of the following three parts.

\textbf{(1) Unbiasedness of the Estimator.~}We first show that the estimator of the value function $\valsim$ is unbiased in \Cref{lemma:mc_unbiased}. Then, based on the unbiasedness of $\valsim$, and the uniformity of the permutation sampling, we obtain the unbiasedness of the Shapley value estimator $\widehat{\shap(\seed)}$.
\begin{lemma}[Unbiasedness of value function estimation]\label{lemma:mc_unbiased}
    Given (G, $\seedset$, $\actprob$), for any $S \subseteq \seedset$, the estimator $\valsim(S)$ obtained by $\permuteMC$ is unbiased, i.e., $\mathbb{E}[\valsim(S)] = \val(S)$.
    \end{lemma}
    
    \begin{proof}
    For a fixed set of seed nodes $\Sset$, let $\mathbb{I}_{i,v}(\Sset)$ be the indicator random variable for whether $v \in \nodeset \setminus \seedset$ is activated in the $i$-th simulation that propogates from $\Sset$. Then the estimator $\valsim(\Sset)$ is computed as:
\begin{equation}
    \valsim(\Sset) = \frac{1}{n_{\text{MC}}} \sum_{i=1}^{n_{\text{MC}}} \sum_{v \in \nodeset \setminus \seedset} \mathbb{I}_{i,v}(\Sset)
    \end{equation}
    %where $m$ is the number of Monte Carlo simulations.
    
    Since all simulations are identically distributed, we take the expectation of $\valsim(\Sset)$ and obtain the following equation:
    \begin{align}
        \mathbb{E}[\valsim(S)] &= \mathbb{E}\left[\frac{1}{n_{\text{MC}}} \sum_{i=1}^{n_{\text{MC}}} \sum_{v \in \nodeset \setminus \seedset} \mathbb{I}_{i,v}(S)\right] \label{eq:mc_unbiased_1} \\
        &= \frac{1}{n_{\text{MC}}} \sum_{i=1}^{n_{\text{MC}}} \sum_{v \in \nodeset \setminus \seedset} \mathbb{E}[\mathbb{I}_{i,v}(S)] \label{eq:mc_unbiased_2} \\
        &= \frac{1}{n_{\text{MC}}} \sum_{i=1}^{n_{\text{MC}}} \sum_{v \in \nodeset \setminus \seedset} \Pr(\mathbb{I}_{i,v}(S) = 1) \label{eq:mc_unbiased_3} \\
        &= \sum_{v \in \nodeset \setminus \seedset} \Pr(v \text{ is activated } | S) \label{eq:mc_unbiased_4} \\
        &= \val(S) \label{eq:mc_unbiased_5}
    \end{align}

\end{proof}

Then we prove the unbiasedness of the Shapley value estimator:

\begin{proof}[Proof for \Cref{prop:permute_mc_property} (1)]\label{proof:shapley_unbiased}
    Taking expectations over both the random permutation $\pi$ and the MC randomness, we get:
    \begin{equation}
    \begin{aligned} 
    \mathbb{E}[\widehat{\shap(\seed)}] &= \mathbb{E}\left[\frac{1}{n_{\pi}} \sum_{i=1}^{n_{\pi}} \big(\valsim(\Sset_{\pi_i, \seed} \cup \{\seed\}) - \valsim(\Sset_{\pi_i, \seed})\big)\right]\\
    &= \frac{1}{n_{\pi}} \sum_{i=1}^{n_{\pi}} \mathbb{E}[\valsim(\Sset_{\pi_i, \seed} \cup \{\seed\}) - \valsim(\Sset_{\pi_i, \seed})]
    \end{aligned}
    \end{equation}

    Note that for any single permutation $\pi_i$, randomness comes from two sources: (1) the random sampling of the permutation $\pi_i$ and (2) the random Monte Carlo simulations used to estimate $\val$. Thus, using the law of total expectation:
        \begin{align} 
        &\mathbb{E}[\valsim(\Sset_{\pi_i, \seed} \cup \{\seed\}) - \valsim(\Sset_{\pi_i, \seed})] \\
        &= \mathbb{E}_{\pi_i \sim \Pi(\seedset)}\left[\mathbb{E}_{\valsim}\left[\valsim(\Sset_{\pi_i, \seed} \cup \{\seed\}) - \valsim(\Sset_{\pi_i, \seed})|\pi_i\right]\right] \label{eq:shapley_unbiased_1}\\
        &= \mathbb{E}_{\pi_i \sim \Pi(\seedset)}\left[\val(\Sset_{\pi_i, \seed} \cup \{\seed\}) - \val(\Sset_{\pi_i, \seed})\right] \label{eq:shapley_unbiased_2}\\
        &= \frac{1}{|\seedset|!} \sum_{\pi \in \Pi(\seedset)} [\val(\Sset_{\pi, \seed} \cup \{\seed\}) - \val(\Sset_{\pi, \seed})] \label{eq:shapley_unbiased_3}
        \end{align}
\Cref{eq:shapley_unbiased_2} follows from the unbiasedness of $\valsim$ in \Cref{lemma:mc_unbiased}. \Cref{eq:shapley_unbiased_3} follows from that each permutation $\pi_i$ is sampled uniformly at random from $\Pi(\seedset)$,

    Therefore:
    \begin{equation}
    \begin{aligned}
    &\mathbb{E}[\widehat{\shap(\seed)}] \\
    &=  \frac{1}{n_{\pi}} \sum_{i=1}^{n_{\pi}} \left( \frac{1}{|\seedset|!} \sum_{\pi \in \Pi(\seedset)} [\val(\Sset_{\pi, \seed} \cup \{\seed\}) - \val(\Sset_{\pi, \seed})] \right) \\
    &= \frac{1}{|\seedset|!} \sum_{\pi \in \Pi(\seedset)} [\val(\Sset_{\pi, \seed} \cup \{\seed\}) - \val(\Sset_{\pi, \seed})] \\
    &= \shap(\seed) \label{eq:shapley_unbiased_4}
    \end{aligned}
    \end{equation}
    %The last equality follows from the definition of Shapley value. This completes the proof that $\widehat{\shap(\seed)}$ is an unbiased estimator of $\shap(\seed)$.
\end{proof}

\textbf{(2) Additive approximation error bound.~}
We will use the Hoeffding's inequality to prove \Cref{prop:permute_mc_property} (2), as detailed in \Cref{lemma:shapley_error_bound}.

\begin{fact}[Hoeffding's inequality]\label{fact:hoeffding}
    Given independent random variables $X_1, X_2, \cdots, X_n$ with bounds $a_i \leq X_i \leq b_i$, for all $t>0$ it satisfies:
    \begin{equation}
    \Pr[|\sum_{i=1}^{n} X_i - \mathbb{E}[\sum_{i=1}^{n} X_i]| \geq t] \leq 2 \exp{(-\frac{2 t^2}{\sum_{i=1}^n(b_i-a_i)^2})}\\
\end{equation}
\end{fact}

\begin{lemma}[Additive error bound of Shapley value estimator]\label{lemma:shapley_error_bound}
    For any $\epsilon > 0$ and $\delta > 0$, the following holds for all $\seed \in \seedset$:
\begin{equation}\label{eq:total_error}
    \Pr\left( |\widehat{\shap(\seed)} - \shap(\seed)| \geq \epsilon \right) \leq \delta.
\end{equation}
if the number of permutations satisfies
\begin{equation}\label{eq:n_size}
    n_{\pi} \geq \frac{8 |\nodeset|^2}{\epsilon^2} \ln\left( \frac{4 |\seedset|}{\delta} \right)
\end{equation}
and the number of Monte Carlo simulations satisfies
\begin{equation}\label{eq:m_size}
    n_{\text{MC}} \geq \frac{8 |\nodeset|^2}{\epsilon^2} \ln ( \frac{4 n_{\pi} |\seedset|}{\delta} )
\end{equation}
\end{lemma}

\begin{proof}

The total error in estimating $\shap(\seed)$ can be decomposed as two terms as follows: (1) the error from estimating marginal contributions $E_1(\seed)$ and (2) the error from sampling a finite number of permutations $E_2(\seed)$. For simplicity, we denote the marginal contribution $\Delta_{\seed}^{\pi_i} = \val(\Sset_{\pi_i, \seed} \cup \{\seed\}) - \val(\Sset_{\pi_i, \seed})$, and thus $\hat{\Delta}_{\seed}^{\pi_i} = \valsim(\Sset_{\pi_i, \seed} \cup \{\seed\}) - \valsim(\Sset_{\pi_i, \seed})$ in the proof.

\begin{equation}
\begin{aligned}
|\widehat{\shap(\seed)} - \shap(\seed)| &= \left|\frac{1}{n_{\pi}}\sum_{i=1}^{n_{\pi}} \hat{\Delta}_{\seed}^{\pi_i} - \shap(\seed)\right| \\
&= \left|\frac{1}{n_{\pi}}\sum_{i=1}^{n_{\pi}} \hat{\Delta}_{\seed}^{\pi_i} - \frac{1}{n_{\pi}}\sum_{i=1}^{n_{\pi}} \Delta_{\seed}^{\pi_i} + \frac{1}{n_{\pi}}\sum_{i=1}^{n_{\pi}} \Delta_{\seed}^{\pi_i} - \shap(\seed)\right| \\
&\leq \underbrace{\left|\frac{1}{n_{\pi}}\sum_{i=1}^{n_{\pi}} (\hat{\Delta}_{\seed}^{\pi_i} - \Delta_{\seed}^{\pi_i})\right|}_{E_1(\seed)} + \underbrace{\left|\frac{1}{n_{\pi}}\sum_{i=1}^{n_{\pi}} \Delta_{\seed}^{\pi_i} - \shap(\seed)\right|}_{E_2(\seed)}
\end{aligned}
\end{equation}

We bound the probability of $E_1(\seed)$ and $E_2(\seed)$ separately using Hoeffding's inequality and then take the union bound.

\paragraph{Error in marginal contribution estimation}
For each marginal contribution $\Delta_{\seed}^{\pi_i}$, the estimation error is:
\begin{equation}
    \begin{aligned}
    &|\hat{\Delta}_{\seed}^{\pi_i} - \Delta_{\seed}^{\pi_i}| = |\valsim(\Sset_{\pi_i, \seed} \cup \{\seed\}) - \valsim(\Sset_{\pi_i, \seed}) - (\val(\Sset_{\pi_i, \seed} \cup \{\seed\}) - \val(\Sset_{\pi_i, \seed}))|\\
    &\leq |\hat{\val}(\Sset_{\pi_i, \seed} \cup \{\seed\}) - \val(\Sset_{\pi_i, \seed} \cup \{\seed\})| + |\hat{\val}(\Sset_{\pi_i, \seed}) - \val(\Sset_{\pi_i, \seed})|
    \end{aligned}
\end{equation}

Since both $\valsim(\Sset_{\pi_i, \seed} \cup \{\seed\})$ and $\valsim(\Sset_{\pi_i, \seed})$ are estimated via Monte Carlo simulations with $|\nodeset|$ as the upper bound, applying Hoeffding's inequality for each $\hat{\Delta}_{\seed}^{\pi_i}$ and set $\epsilon_{\Delta}$ as the maximum error in estimating $\Delta_{\seed}^{\pi_i}$:

\begin{equation}
\begin{aligned}
\Pr\left( \left| \hat{\Delta}_{\seed}^{\pi_i} - \Delta_{\seed}^{\pi_i} \right| \geq \epsilon_{\Delta}\right) &\leq 2 \exp\left( -\frac{2 n_{\text{MC}} \epsilon_{\Delta}^2}{(2 |\nodeset|)^2} \right)\\
&= 2 \exp\left( -\frac{n_{\text{MC}} \epsilon_{\Delta}^2}{2 |\nodeset|^2} \right)\\
\end{aligned}
\end{equation}

Then we apply this result to bound the probability of $E_1(\seed)$ exceeding $\epsilon_{\Delta}$:

\begin{equation}
\begin{aligned}
\Pr(E_1(\seed) \geq \epsilon_{\Delta}) &= \Pr\left(\left|\frac{1}{n}\sum_{i=1}^n (\hat{\Delta}_{\seed}^{\pi_i} - \Delta_{\seed}^{\pi_i})\right| \geq \epsilon_{\Delta}\right) \\
&\leq \Pr\left(\exists i: |\hat{\Delta}_{\seed}^{\pi_i} - \Delta_{\seed}^{\pi_i}| \geq \epsilon_{\Delta}\right)
\end{aligned}
\end{equation}
Because we estimated the marginal contribution of a given seed node for all $n$ permutations, we have:

\begin{equation}
    \begin{aligned}
    \Pr(E_1(\seed) \geq \epsilon_{\Delta}) &\leq \Pr\left(\exists i: |\hat{\Delta}_{\seed}^{\pi_i} - \Delta_{\seed}^{\pi_i}| \geq \epsilon_{\Delta}\right) \\
    &\leq n \cdot \Pr\left( \left| \hat{\Delta}_{\seed}^{\pi_i} - \Delta_{\seed}^{\pi_i} \right| \geq \epsilon_{\Delta}\right)\\
    &\leq n \cdot 2 \exp\left( -\frac{m \epsilon_{\Delta}^2}{2 |\nodeset|^2} \right)
    \end{aligned}
    \end{equation}
Applying the union bound over all seeds:

\begin{equation}
\begin{aligned}
\Pr(\exists \seed \in \seedset: E_1(\seed) \geq \epsilon_{\Delta}) &\leq |\seedset| \cdot n \cdot 2 \exp\left( -\frac{m \epsilon_{\Delta}^2}{2 |\nodeset|^2} \right)
\end{aligned}
\end{equation}
Then we set $\epsilon_{\Delta} = \frac{\epsilon}{2}$ and replace $m$ using \Cref{eq:m_size}:
\begin{equation}
    \begin{aligned}
    &\Pr(\exists \seed \in \seedset: E_1(\seed) \geq \frac{\epsilon}{2})\\
    &\leq |\seedset| \cdot n \cdot 2 \exp\left( -\frac{\epsilon^2}{8 |\nodeset|^2} \cdot  \frac{8 |\nodeset|^2}{\epsilon^2} \ln ( \frac{4 n |\seedset|}{\delta} ) \right)\\
    & = \frac{\delta}{2}
    \end{aligned}
    \end{equation}
Therefore, the probability that $\forall \seed \in \seedset$, $E_1(\seed)$ at most $\frac{\epsilon}{2}$ is at least $1-\frac{\delta}{2}$.

\paragraph{Error from sampling permutations}
Then we consider $E_2(\seed)$, the error from sampling permutations.
Both $\val(\Sset_{\pi_i, \seed} \cup \{\seed\})$ and $\val(\Sset_{\pi_i, \seed})$ are bounded between $0$ and $|\nodeset|$, so we can apply Hoeffding's inequality as before and set $\epsilon_{\text{perm}}$ as the maximum error for $E_2(\seed)$:

\begin{equation}
\begin{aligned}
&\Pr\left( E_2(\seed) \geq \epsilon_{\text{perm}} \right) = \Pr\left( \left| \frac{1}{n_{\pi}} \sum_{i=1}^{n_{\pi}} \Delta_{\seed}^{\pi_i} - \shap(\seed) \right| \geq \epsilon_{\text{perm}} \right) \\
&\leq 2 \exp\left( -\frac{2 n_{\pi} \epsilon_{\text{perm}}^2 }{ (2 |\nodeset|)^2 } \right) = 2 \exp\left( -\frac{n_{\pi} \epsilon_{\text{perm}}^2}{2 |\nodeset|^2} \right)
\end{aligned}
\end{equation}

Similarly, applying the union bound over all seeds:

\begin{equation}
\begin{aligned}
\Pr\left( \exists \seed \in \seedset : E_2(\seed) \geq \epsilon_{\text{perm}} \right) &\leq  |\seedset| \cdot 2 \exp\left( -\frac{ n \epsilon_{\text{perm}}^2 }{ 2 |\nodeset|^2 } \right)
\end{aligned}
\end{equation}
We set $\epsilon_{\text{perm}} = \frac{\epsilon}{2}$ and replace $n$ using \Cref{eq:n_size}:

\begin{equation}
\begin{aligned}
&\Pr\left( \exists \seed \in \seedset : E_2(\seed) \geq \frac{\epsilon}{2} \right)\\
&\leq  |\seedset| \cdot 2 \exp\left( -\frac{\epsilon^2}{8 |\nodeset|^2} \cdot  \frac{8 |\nodeset|^2}{\epsilon^2} \ln ( \frac{4 |\seedset|}{\delta} ) \right)\\
& = \frac{\delta}{2}
\end{aligned}
\end{equation}
Hence, the probability that $\forall \seed \in \seedset$, $E_2(\seed)$ at most $\frac{\epsilon}{2}$ is at least $1-\frac{\delta}{2}$.

\paragraph{Combining the errors}
Combine the probability of $E_1(\seed)$ and $E_2(\seed)$ exceeding their respective errors:
\begin{equation}
\begin{aligned}
&\Pr\left( \exists \seed \in \seedset : \left| \hat{\shap}(\seed) - \shap(\seed) \right| \geq \epsilon \right)\\
\leq& \Pr\left( \exists \seed \in \seedset : E_1(\seed) + E_2(\seed) \geq \epsilon \right)\\
\leq& \Pr\left( \exists \seed \in \seedset : E_1(\seed) \geq \frac{\epsilon}{2} \right) + \Pr\left( \exists \seed \in \seedset : E_2(\seed) \geq \frac{\epsilon}{2} \right) \\
\leq& \frac{\delta}{2} + \frac{\delta}{2} = \delta.
\end{aligned}
\end{equation}

This shows that with probability at least $1-\delta$, the total error for all seeds is bounded by $\epsilon$, completing the proof of the approximation guarantee.

\end{proof}
\textbf{(3) Runtime Complexity.~}
Last, we prove the polynomial runtime of $\permuteMC$:
\begin{proof}[Proof for \Cref{prop:permute_mc_property} (3)]
    For each of the $n_{\pi}$ permutations and each seed node in $\seedset$, we perform $n_{\text{MC}}$ Monte Carlo simulations. So the total number of simulations is $O(n_{\pi} \cdot |\seedset| \cdot n_{\text{MC}})$. Since each simulation runs a BFS traversal that takes $O(|E|)$ time, the runtime per simulation is $O(|E|)$. therefore, the total runtime is $O(n_{\pi} \cdot |\seedset| \cdot n_{\text{MC}} \cdot |E|)$
    
    As analyzed in \Cref{lemma:shapley_error_bound}, the sampling sizes $n_{\pi}$ and $n_{\text{MC}}$ are polynomial:
    \begin{align}
        n_{\pi} &\geq \frac{8 |\nodeset|^2}{\epsilon^2} \ln ( \frac{4 n_{\pi} |\seedset|}{\delta} ) = O\left( \frac{|V|^2}{\epsilon^2} \ln ( \frac{|V|^2 |\seedset|^2}{\epsilon^2 \delta} ) \right) \\
        n_{\text{MC}} &\geq \frac{8 |V|^2}{\epsilon^2} \ln ( \frac{4 |\seedset|}{\delta} ) = O\left( \frac{|V|^2}{\epsilon^2} \ln ( \frac{|\seedset|}{\delta} ) \right)
        \end{align}
    %Since $|\seedset|$ is bounded by $|\nodeset|$, each $m$ and $n$ is polynomial in $|V|$, $|\seedset|$, $1/\epsilon$, and $\ln(1/\delta)$.

    As $|E| = O(|V|^2)$ and $|\seedset| \leq |V|$, substituting the values of $n$ and $m$, the total runtime is polynomial in $|V|$, $\frac{1}{\epsilon}$, and $\ln(\frac{1}{\delta})$.
\end{proof}

\section{Details of $\algliveedge$}\label{sec:algliveedge_appendix}

\subsection{\textbf{Pseudocode}}
We present $\algliveedge$ in \Cref{alg:live_edge_graph}. It first samples $n=O\left( \frac{|V|^2}{\epsilon^2} \ln( \frac{|\seedset|}{\delta}) \right)$ live-edge graphs in Line~\ref{line:sample_graph}, where each edge in $g_i$ is sampled based on its activation probability.
For each sampled $g_i$, it operates in two phases: (1) first, it performs a single multi-source Breadth-First Search (BFS) starting from all seed nodes to construct a reachability map, $B_{g_i}$, which records the set of seed nodes that can reach each non-seed node (Line~\ref{line:msbfs}). (2)Then the algorithm iterates through all non-seed nodes (Line~\ref{line:node_loop_start})
and distributes its unit value equally among all seeds in its reachability set $B_{g_i}[v]$ by adding $1/|B_{g_i}[v]|$ to each seed's estimated contribution $est_{g_i}[\seed]$ (Lines~\ref{line:seed_loop_start} to~\ref{line:seed_loop_end}).
Finally, after all sampled live-edge graphs are processed, the algorithm returns the averaged contributions as estimated final Shapley values (Line~\ref{line:average}).

\begin{algorithm}[!htbp]
    \caption{$\algliveedge(G, T, p, \epsilon, \delta)$}
    \label{alg:live_edge_graph}
    \begin{algorithmic}[1]
    \Require Network $G(V,E)$, seed set $T$, activation probabilities $p$, approximation parameters $\epsilon >0$, $\delta \in (0,1)$
    \Ensure Estimated Shapley values $\widehat{\text{Shap}}(t)$ for all $t \in T$
    \State Number of sampled live-edge graphs $n \gets O\left( \frac{|V|^2}{\epsilon^2} \ln( \frac{|\seedset|}{\delta}) \right)$ \label{ine:num_samples}
    \For{$i = 1$ to $n$} \label{line:outer_loop_start}
        \State Sample a live-edge graph $g_i$ from $G$ with $p$ \label{line:sample_graph}
        \State \textbf{Phase 1: Construct bipartite reachability map $B_i$} \label{line:phase1_start}
        \State $B_{g_i} \gets \textsc{MultiSourceBFS}(g_i, \seedset)$ \label{line:msbfs}
        %\State \textit{\textcolor{gray}{// Perform a multi-source BFS on $g_i$ from all seed nodes to get a map $B_{g_i}$ that for each non-seed node $v \in V \setminus T$, $B_{g_i}[v]$ stores the set of seeds that can reach $v$ in $g_i$}} \label{line:msbfs_comment}
        \State \textit{\textcolor{gray}{// $B_{g_i}[v]$ stores seed nodes that can reach non-seed $v$ in $g_i$}} \label{line:msbfs_comment}
        % \State Initialize a queue $Q$
        % \For{each $t \in T$}
        %     \State Enqueue $(t, t)$ into $Q$ \Comment{(current\_node, originating\_seed)}
        % \EndFor
        % \While{$Q$ is not empty}
        %     \State Pop $(u, s)$ from $Q$ 
        %     \For{each out-neighbor $v$ of $u$ in $g_i$}
        %         \If{$v \notin T$ \textbf{and} $s \notin B_i[v]$}
        %             \State Add $s$ to $B_i[v]$ \Comment{as $s$ can reach $v$}
        %             \State Enqueue $(v, s)$ into $Q$
        %         \EndIf
        %     \EndFor
        % \EndWhile
        %\State \Comment{Get a bipartite graph $B_i$ that stores for each edge $(s, v)$ in $B$, $s \in T$ can reach $v \in V \setminus T$ in $g_i$}
        \State \textbf{Phase 2: Credit Distribution} \label{line:phase2_start}
        \State Initialize $\textit{est}_{g_i}[t] \gets 0$ for all $t \in T$ \label{line:init_est}
        \For{each non-seed node $v \in V \setminus T$ with $B_{g_i}[v] \not= \emptyset$} \label{line:node_loop_start}
            %\State contribution $\gets \frac{1}{|B_{g_i}[v]|}$ \label{line:calc_contribution}
            \For{each seed $t \in B_{g_i}[v]$} \label{line:seed_loop_start}
                \State $\textit{est}_{g_i}[t] \gets \textit{est}_{g_i}[t] +  \frac{1}{|B_{g_i}[v]|}$ \label{line:add_contribution}
            \EndFor \label{line:seed_loop_end}
        \EndFor \label{line:node_loop_end}
    \EndFor \label{line:outer_loop_end}
    \State \textbf{Phase 3: Final Averaging} \label{line:phase3_start}
        \State $\widehat{\text{Shap}}(t) \gets \frac{1}{n} \sum_{i = 1}^{n} \textit{est}_{g_i}[t]$ for all $t \in T$ \label{line:average}
    \State \Return $\widehat{\text{Shap}}(t)$ for all $t \in T$ \label{line:return}
    \end{algorithmic}
\end{algorithm}

% \begin{proposition}[Runtime of $\algliveedge$]\label{prop:liveedge_runtime}
% For any $\epsilon > 0$ and $\delta \in (0, 1)$, the total running time of $\algliveedge$ is $O\left(\frac{|V|^2}{\epsilon^2} \ln(\frac{|\seedset|}{\delta}) \cdot |\seedset||E|\right)$.
% \end{proposition}

\subsection{\textbf{Credit-splitting identity.}}
We first establish the connection between the Shapley value and then live-edge graph realizations, which is the key
tool underlying $\algliveedge$.

Recall that for a live-edge graph $g$
and a non-seed node $v \in V \setminus \seedset$, $B_g[v] \subseteq \seedset$
denotes the set of seed nodes that can reach $v$ in $g$, as computed by
the multi-source BFS in Algorithm~3.

\begin{lemma}\label{lem:credit-splitting}
For any fixed live-edge graph $g$, any non-seed node $v \in V \setminus \seedset$,
and any seed $\seed \in \seedset$,
\begin{equation}
    \mathbb{E}_{\pi}\left[\mathbbm{1}[\seed \text{ activates } v \text{ in } (g, \pi)]\right]
    = \frac{\mathbbm{1}[\seed \in B_g[v]]}{|B_g[v]|}
\end{equation}
where $\pi$ is a uniformly random permutation of $\seedset$, and we adopt the convention $\frac{0}{0} = 0$ when $B_g[v] = \emptyset$.
\end{lemma}

\begin{proof}
Fix a live-edge graph $g$ and a non-seed node $v$.
In the influence propagation process under $(g, \pi)$, node $v$ is activated
by the first seed in $\pi$ that can reach $v$ in $g$.

If $\seed \notin B_g[v]$, then $\seed$ cannot reach $v$ in $g$, so
$\mathbbm{1}[\seed \text{ activates } v \text{ in } (g, \pi)] = 0$ for every permutation $\pi$.
The right-hand side is also $0$ since $\mathbbm{1}[\seed \in B_g[v]] = 0$.
If $\seed \in B_g[v]$, let $k = |B_g[v]| \geq 1$.
Among all $|\seedset|!$ permutations of $\seedset$, the relative order of the $k$ seeds
in $B_g[v]$ is uniformly distributed over all $k!$ orderings, regardless of the
positions of seeds outside $B_g[v]$.
Therefore, each of the $k$ seeds in $B_g[v]$ appears first among them
with probability exactly $1/k$.
Hence,
\begin{equation}
    \mathbb{E}_{\pi}\left[\mathbbm{1}[\seed \text{ activates } v \text{ in } (g, \pi)]\right]
    = \frac{1}{k}
    = \frac{\mathbbm{1}[\seed \in B_g[v]]}{|B_g[v]|}. \qedhere
\end{equation}
\end{proof}

%Therefore, for any fixed live-edge graph $g$, the contribution assigned to seed $\seed$ for node $v$ is $\frac{\mathbbm{1}[\seed \in B_g[v]]}{|B_g[v]|}$, which equals the probability that $\seed$ is the first seed to activate $v$ under a uniformly random permutation of seeds.
\subsection{Theoretical guarantees}
Then we present the approximation guarantees and runtime complexity of $\algliveedge$ in the following proposition, along with the complete proof.

\begin{restatable}{proposition}{liveedgeproperty}\label{prop:liveedge_property}
    Given $(G(\nodeset, E), \seedset, \actprob)$, $\algliveedge$ returns $\widehat{\shap}(\seed)$ for all $\seed \in \seedset$ such that: (1) $\mathbb{E}[\widehat{\shap}(\seed)] = \shap(\seed)$, $\forall \seed \in \seedset$; (2) for any $\epsilon > 0$ and $\delta \in (0,1)$, with probability at least $1 - \delta$, $|\widehat{\shap}(\seed) - \shap(\seed)| \le \epsilon$ holds for all $\seed \in \seedset$; (3) the running time is $O\left(\frac{|V|^2}{\epsilon^2} \ln(\frac{|\seedset|}{\delta}) \cdot |\seedset| \cdot |E|\right)$.
\end{restatable}

\begin{proof}
    We prove each property separately.

\textbf{(1) Unbiasedness.~}
By the live-edge graph equivalence, for any seed set $S \subseteq \seedset$, we have
$U(S) = \mathbb{E}_{g}[u_g(S)]$, where $u_g(S)$ counts the number of non-seed nodes
reachable from $S$ in a random live-edge graph $g$.
Therefore, the marginal contribution of $\seed$ under permutation $\pi$ satisfies:
\begin{equation}
    U(S_{\pi,\seed} \cup \{\seed\}) - U(S_{\pi,\seed})
    = \mathbb{E}_{g}\left[u_g(S_{\pi,\seed} \cup \{\seed\}) - u_g(S_{\pi,\seed})\right]
\end{equation}
For a fixed live-edge graph $g$ and permutation $\pi$, the quantity
$u_g(S_{\pi,\seed} \cup \{\seed\}) - u_g(S_{\pi,\seed})$ counts exactly the non-seed nodes
that are reachable from $\seed$ in $g$ but not reachable from any seed preceding $\seed$ in $\pi$.
This is precisely the number of nodes that $\seed$ activates in $(g, \pi)$.
Substituting into the permutation-based definition of Shapley value (\Cref{eq:shapley_value_1}) and exchanging the sum over permutations
with the expectation over $g$ (by linearity of expectation):
\begin{equation}\label{eq:shap-perm}
    \shap(\seed) = \mathbb{E}_{g}\left[\mathbb{E}_{\pi}\left[\sum_{v \in V \setminus \seedset} \mathbbm{1}[\seed \text{ activates } v \text{ in } (g, \pi)]\right]\right]
\end{equation}
where the outer expectation is over random live-edge graphs $g$ sampled from $G$ according to activation probabilities $p$,
and the inner expectation is over uniformly random permutations $\pi$ of $\seedset$.

Applying \Cref{lem:credit-splitting}:
\begin{equation}\label{eq:shap-credit}
    \shap(\seed)
    = \mathbb{E}_{g}\left[\sum_{v \in V \setminus \seedset} \frac{\mathbbm{1}[\seed \in B_g[v]]}{|B_g[v]|}\right]
\end{equation}

Now consider what $\algliveedge$ computes.
For each independently sampled live-edge graph $g_i$ ($i = 1, \ldots, n$), the algorithm
performs a multi-source BFS from all seeds to obtain $B_{g_i}[v]$
for every non-seed node $v$, and then sets:
\begin{equation}
    \textit{est}_{g_i}[\seed]
    = \sum_{\substack{v \in V \setminus \seedset \\ B_{g_i}[v] \neq \emptyset}}
      \frac{\mathbbm{1}[\seed \in B_{g_i}[v]]}{|B_{g_i}[v]|}
\end{equation}
Each $\textit{est}_{g_i}[\seed]$ therefore has the same distribution as
$\sum_{v \in V \setminus \seedset} \frac{\mathbbm{1}[\seed \in B_g[v]]}{|B_g[v]|}$.
The final output is:
\begin{equation}
    \widehat{\shap}(\seed) = \frac{1}{n}\sum_{i=1}^{n} \textit{est}_{g_i}[\seed]
\end{equation}
By linearity of expectation and Eq.~\eqref{eq:shap-credit}:
\begin{equation}
\begin{aligned}
    &\mathbb{E}\left[\widehat{\shap}(\seed)\right]
    = \frac{1}{n}\sum_{i=1}^{n} \mathbb{E}\left[\textit{est}_{g_i}[\seed]\right]\\
    &= \mathbb{E}_{g}\left[\sum_{v \in V \setminus \seedset} \frac{\mathbbm{1}[\seed \in B_g[v]]}{|B_g[v]|}\right]
    = \shap(\seed).
\end{aligned}
\end{equation}

\noindent\textbf{(2) Accuracy.}
For each sampled live-edge graph $g_i$, each non-seed node $v$ contributes
at most $1$ to $\textit{est}_{g_i}[\seed]$
(since $\frac{\mathbbm{1}[\seed \in B_{g_i}[v]]}{|B_{g_i}[v]|} \leq 1$).
Summing over at most $|V \setminus \seedset| \leq |V|$ non-seed nodes,
we obtain $\textit{est}_{g_i}[\seed] \in [0, |V|]$.

Since $g_1, \ldots, g_n$ are sampled independently,
$\textit{est}_{g_1}[\seed], \ldots, \textit{est}_{g_n}[\seed]$ are independent random variables,
each bounded in $[0, |V|]$.
Applying Hoeffding's inequality to $\widehat{\shap}(\seed) = \frac{1}{n}\sum_{i=1}^{n} \textit{est}_{g_i}[\seed]$:
\begin{equation}
    \Pr\left[|\widehat{\shap}(\seed) - \shap(\seed)| > \epsilon\right]
    \leq 2\exp\left(-\frac{2n\epsilon^2}{|V|^2}\right).
\end{equation}
Applying a union bound over all $|\seedset|$ seeds, we get:
\begin{equation}
    \Pr\left[\exists \seed \in \seedset: |\widehat{\shap}(\seed) - \shap(\seed)| > \epsilon\right]
    \leq |\seedset| \cdot 2\exp\left(-\frac{2n\epsilon^2}{|V|^2}\right)
\end{equation}

Setting the right-hand side to at most $\delta$ and solving for $n$:
\begin{equation}
\begin{aligned}
    &|\seedset| \cdot 2\exp\left(-\frac{2n\epsilon^2}{|V|^2}\right) \leq \delta
    \quad \Longleftrightarrow \quad\\
    &n \geq \frac{|V|^2}{2\epsilon^2} \ln\left(\frac{2|\seedset|}{\delta}\right)
    = O\left(\frac{|V|^2}{\epsilon^2} \ln\left(\frac{|\seedset|}{\delta}\right)\right)
    \end{aligned}
\end{equation}

Therefore, with probability at least $1 - \delta$,
$|\widehat{\shap}(\seed) - \shap(\seed)| \leq \epsilon$ holds for all $\seed \in \seedset$.

\noindent\textbf{(3) Runtime.}
For each of the $n$ sampled live-edge graphs, the algorithm performs three operations.
First, sampling a live-edge graph $g_i$ requires one pass over all edges, taking $O(|E|)$ time.
Second, the multi-source BFS starting from all seeds in $\seedset$ constructs the
reachability map $B_{g_i}$.
In the worst case, each edge $(u, v)$ in $g_i$ may be traversed once per seed
that reaches $u$, giving a cost of $O(|\seedset| \cdot |E|)$.
Third, the credit distribution phase iterates over all non-seed nodes and their
reachability sets; since each node $v$ has $|B_{g_i}[v]| \leq |\seedset|$ entries,
this takes $O(|V| \cdot |\seedset|)$ time, which is dominated by the BFS cost.

Therefore, the per-iteration cost is $O(|\seedset| \cdot |E|)$.
With $n = O\left(\frac{|V|^2}{\epsilon^2} \ln\left(\frac{|\seedset|}{\delta}\right)\right)$
iterations, the total runtime is:
\begin{equation}
    O\left(\frac{|V|^2}{\epsilon^2} \ln\left(\frac{|\seedset|}{\delta}\right) \cdot |\seedset| \cdot |E|\right). \qedhere
\end{equation}
\end{proof}

\section{Details of $\algrrset$}\label{sec:rr_analysis_appendix}
We first present $\algrrset$ in \Cref{alg:shapley_rr} along with its subroutine $\textsc{EstimateThreshold}$ in \Cref{subsec:rrset_pseudocode}, and then provide the proof for Lemma~\ref{lem:shapley_rr} in \Cref{app:proof_shapley_rr}. , and then present results and rigorous analysis of approximation guarantees in \Cref{subsec:rrset_approximation_guarantees} and runtime complexity in \Cref{subsec:rr_time_complexity_appendix}.
The analysis adapts the framework from \cite{chen2017interplay,tang2015influence} to our specific problem of a fixed seed set and a modified value function.

\subsection{\textbf{Pseudocode}}\label{subsec:rrset_pseudocode}

\begin{algorithm}[!htbp]
\caption{$\algrrset(G, \seedset, \actprob, \varepsilon, \ell, k)$}
\label{alg:shapley_rr}
\begin{algorithmic}[1]
\Require Network $G(V, E)$, seed set $\seedset$, activation probabilities $\actprob$, approximation parameters $\varepsilon > 0$, $\ell > 0$, $k \in [|\seedset|]$
\Ensure Estimated Shapley values $\widehat{\shap(\seed)}$ for all $\seed \in \seedset$
\State \textbf{Phase 0: Graph Modification}
\State $E^{\prime} \gets \{(u,v) \in E \mid v \notin \seedset\}$; $G^{\prime} \gets (V, E^{\prime})$ \label{line:rr_modify_graph}
\State \textit{\textcolor{gray}{// Remove all incoming edges to seed nodes}}

\State $\rrsize \gets |V \setminus \seedset|$ \label{line:rr_n_prime}
\State \textbf{Phase 1: Parameter Estimation}
\State $\boldsymbol{LB} \gets \textsc{EstimateThreshold}(G^{\prime}, \seedset, \varepsilon, \ell, k)$ \label{line:rr_phase1_lb}
\State \textit{\textcolor{gray}{// Find lower bound of the $k$-th largest Shapley value}}
\State $\theta \gets \left\lceil \frac{ \rrsize \left( 2 + \frac{2}{3} \varepsilon \right)}{ \varepsilon^2 \cdot \boldsymbol{LB}} \left( \ell \ln \rrsize + \ln |\seedset| + \ln 4 \right)  \right\rceil$ \label{line:rr_phase2_theta}
\State \textit{\textcolor{gray}{// Calculate required number of RR sets}}
\State \textbf{Phase 2: Shapley Value Estimation}
\State Reset $\boldsymbol{est}_t \gets 0$ for all $t \in \seedset$
\For{$j = 1$ to $\theta$} \label{line:rr_phase2_loop_start}
    \State Generate a random RR set $R_j$ by selecting a node from $V \setminus \seedset$ uniformly at random \label{line:sample_rrset}
    \If{$R_j \cap \seedset \neq \emptyset$}
        \For{each $t \in R_j \cap \seedset$}\label{line:rr_phase2_update_est}
            \State $\boldsymbol{est}_t \gets \boldsymbol{est}_t + \frac{1}{| R_j \cap \seedset |}$ \label{line:rr_phase2_update_est_inner}
        \EndFor
    \EndIf
\EndFor \label{line:rr_phase2_loop_end}
\State \Return $\widehat{\shap}(t) \gets \rrsize \cdot \frac{\boldsymbol{est}_t}{\theta}$ for all $t \in \seedset$ \label{line:rr_phase2_return}

\end{algorithmic}
\end{algorithm}

\begin{algorithm}[!htbp]
    \caption{$\textsc{EstimateThreshold}(G^{\prime}, \seedset, \varepsilon, \ell, k)$}
    \label{alg:estimate_threshold_complete}
    \begin{algorithmic}[1]
    \Require modified graph $G^{\prime} = (V, E^{\prime})$, seed set $\seedset$, activation probabilities $\actprob$, approximation parameters $\varepsilon > 0$, $\ell > 0$, $k \in [|\seedset|]$
    \Ensure $\boldsymbol{LB}$ \Comment{A lower bound estimate for the $k$-th largest Shapley value among all nodes in $\seedset$}
    \State Initialize $\boldsymbol{LB} \gets 1$, $\varepsilon' \gets \sqrt{2} \varepsilon$, $\theta_0 \gets 0$, $\boldsymbol{est}_t \gets 0$ for all $t \in \seedset$ \label{line:et_init}
    \For{$i = 1$ to $\left\lfloor \log_2 \rrsize \right\rfloor - 1$} \label{line:et_outer_loop}
        \State $x_i \gets \frac{\rrsize}{2^i}$ \label{line:et_xi}
        \State $\theta_i \gets \left\lceil \frac{ \rrsize\left( 2 + \frac{2}{3} \varepsilon' \right)}{ {\varepsilon'}^2 x_i } \left(\ell \ln \rrsize + \ln |\seedset| + \ln \log_2 \rrsize + \ln 2 \right) \right\rceil$ \label{line:et_theta_i}
        \For{$j = 1$ to $\theta_i - \theta_{i-1}$} \label{line:et_inner_loop}
            \State Generate a random RR set $R_j$ in $G^{\prime}$ by selecting a node from $V \setminus \seedset$ uniformly at random \label{line:et_sample_rr}
            \If{$R_j \cap \seedset \neq \emptyset$}
                \For{each $t \in R_j \cap \seedset$}
                    \State $\boldsymbol{est}_t \gets \boldsymbol{est}_t + \frac{1}{| R_j \cap \seedset |}$ \label{line:et_update_est}
                \EndFor
            \EndIf
        \EndFor
        \State $\boldsymbol{est}^{(k)} \gets $ the $k$-th largest value in $\left\{\boldsymbol{est}_{\seed}\right\}_{\seed \in \seedset}$ \label{line:et_est_k}
        \If{$\rrsize \cdot \frac{\boldsymbol{est}^{(k)}}{\theta_i} \geq (1 + \varepsilon') x_i$} \label{line:et_check}
            \State $\boldsymbol{LB} \gets \rrsize \cdot \frac{\boldsymbol{est}^{(k)}}{\theta_i(1 + \varepsilon')} $ \label{line:et_set_lb}
            \State \textbf{break}
        \EndIf
    \EndFor
    \State \Return $\boldsymbol{LB}$ \label{line:et_return}
    \end{algorithmic}
\end{algorithm}

\textbf{Description of the pseudocode.}
\Cref{alg:shapley_rr} takes parameters $\varepsilon$ (multiplicative error), $\ell$ (confidence $1 - 1/n^\ell$), and $k$ (the error bound holds for the $k$ largest Shapley values).
It runs in three phases.
\textbf{Phase 0 (Lines~\ref{line:rr_modify_graph}--\ref{line:rr_n_prime}):} Build the modified graph $G^{\prime} = (V, E^{\prime})$ by removing all incoming edges to seed nodes, and set $\rrsize = |V \setminus \seedset|$.
\textbf{Phase 1 (Lines~\ref{line:rr_phase1_lb}--\ref{line:rr_phase2_theta}):} Call $\textsc{EstimateThreshold}$ (\Cref{alg:estimate_threshold_complete}) to obtain a lower bound $\boldsymbol{LB}$ on the $k$-th largest Shapley value, then compute the required number of RR sets $\theta$ from the formula at Line~\ref{line:rr_phase2_theta}.
The subroutine $\textsc{EstimateThreshold}$ (\Cref{alg:estimate_threshold_complete}) follows the adaptive-sampling framework of~\cite{chen2017interplay}: it initializes $\boldsymbol{LB}$ and $\boldsymbol{est}_t$ (Line~\ref{line:et_init}), then over rounds $i = 1, \ldots, \lfloor \log_2 \rrsize \rfloor - 1$ (Line~\ref{line:et_outer_loop}) sets $x_i = \rrsize/2^i$ and $\theta_i$ (Lines~\ref{line:et_xi}--\ref{line:et_theta_i}), generates $\theta_i - \theta_{i-1}$ new RR sets and updates $\boldsymbol{est}_t$ for each seed in $R_j \cap \seedset$ by $1/|R_j \cap \seedset|$ (Lines~\ref{line:et_inner_loop}--\ref{line:et_update_est}), takes the $k$-th largest $\boldsymbol{est}^{(k)}$ (Line~\ref{line:et_est_k}), and if $\rrsize \cdot \boldsymbol{est}^{(k)}/\theta_i \geq (1+\varepsilon') x_i$ sets $\boldsymbol{LB}$ and breaks (Lines~\ref{line:et_check}--\ref{line:et_set_lb}); it returns $\boldsymbol{LB}$ (Line~\ref{line:et_return}).
\textbf{Phase 2 (Lines~\ref{line:rr_phase2_loop_start}--\ref{line:rr_phase2_return}):} Reset $\boldsymbol{est}_t \gets 0$ for all $t \in \seedset$.
For $j = 1$ to $\theta$, generate a random RR set $R_j$ by sampling a root from $V \setminus \seedset$ uniformly (Line~\ref{line:sample_rrset}); if $R_j \cap \seedset \neq \emptyset$, add $1/|R_j \cap \seedset|$ to $\boldsymbol{est}_t$ for each $t \in R_j \cap \seedset$ (Lines~\ref{line:rr_phase2_update_est}--\ref{line:rr_phase2_update_est_inner}).
Finally, return $\widehat{\shap}(t) = \rrsize \cdot \boldsymbol{est}_t / \theta$ for all $t \in \seedset$ (Line~\ref{line:rr_phase2_return}).

\subsection{Proof of Lemma~\ref{lem:shapley_rr}}\label{app:proof_shapley_rr}
\shapleyRr*

The above \Cref{lem:shapley_rr} shows the connection between the Shapley value of a seed node $\seed$ and the probability that $\seed$ appears in a random RR set, which is the foundamental principle of $\algrrset$. In this section, we first prove the connection between RR sets and the influence propagation process on $G^{\prime}$ in \Cref{lem:rr_equivalence}, then the relationship between RR sets and the value function and the marginal contribution in \Cref{lem:rr_valuefunction,lem:marginal_rr}. Finally, we use these results to complete the proof of \Cref{lem:shapley_rr}. The proof structure is similar to \cite{chen2017interplay} but with adjustment to our problem setting.

\begin{lemma}\label{lem:rr_equivalence}
    For any $\Sset \subseteq \seedset$ and any $v \in \nodeset \setminus \seedset$, denote $\rrset_v$ as an RR set for $v$ in $G^{\prime}$. The probability that $\rrset_v$ intersects $\Sset$ equals to the probability that $\Sset$ activates $v$ in an influence propagation process on $G$:
    \begin{equation}\label{eq:rr_equivalence}
        \Pr[\Sset \cap \rrset_v \neq \emptyset] = \Pr[\Sset \text{ activates } v \text{ in } G]
    \end{equation}
\end{lemma}

\begin{proof}
    Let $g$ be a live-edge graph obtained by removing each edge $e$ with $1-p(e)$ probability from $G^{\prime}$. We generate an RR set $\rrset_v$ for $v$ in $g$ by traversing the graph in reverse starting from $v$. Then the event $\Sset$ intersects $\rrset_v$ is equivalent to that there exists a path from $\Sset$ to $v$ in $g$. Moreover, the event that $\Sset$ activates $v$ in $g$ is also equivalent to that there exists a path from $\Sset$ to $v$ in $g$. Therefore, the two events are equivalent in any realization of $g$. Then their probabilities over the distribution of $g$ must be equal:
    \begin{equation}
    \begin{aligned}
        \Pr_{g \sim G^{\prime}}[\Sset \cap \rrset_v \neq \emptyset] & = \Pr_{g \sim G^{\prime}}[\text{there exists a path from } \Sset \text{ to } v \text{ in } g] \\
        & = \Pr_{g \sim G^{\prime}}[v \text{ is activated by } \Sset \text{ in } G^{\prime}]
    \end{aligned}
    \end{equation}

    Moreover, since $G'$ is the modified graph obtained by removing incoming edges to the seed nodes in $\Sset$ from $G$. For all non-seed nodes, the activation probability is unchanged:
    \begin{equation}
        \begin{aligned}
            \Pr_{g \sim G^{\prime}}[S \cap R_v \neq \emptyset]& = \Pr_{g \sim G^{\prime}}[v \text{ is activated by } S \text{ in } G^{\prime}]\\
            & = \Pr_{g \sim G}[v \text{ is activated by } S \text{ in } G]
        \end{aligned}
        \end{equation}

    This establishes the desired connection between the RR sets and the influence propagation process, completing the proof.
\end{proof}

Then we can establish the following relationship between the value function of any subset of seed nodes and the probability that a random RR set intersects with this subset.

\begin{restatable}{lemma}{lemrrvfunction}\label{lem:rr_valuefunction}
    Given ($G(V,E)$, $\seedset$,$\actprob$), let $\rrset$ be a random RR set generated in $G^{\prime}$ by selecting a node $v$ uniformly at random from $\nodeset \setminus \seedset$ and performing reverse reachability from $v$. Then for any $\Sset \subseteq \seedset$, its value function is proportional to the probability that $\Sset$ intersects $\rrset$:
    \begin{equation} \label{eq:rr_valuefunction}
        \val_{G,\seedset}(\Sset) = \rrsize \cdot \Pr[\Sset \cap \rrset \neq \emptyset]    
    \end{equation}
    where $\rrsize = |\nodeset \setminus \seedset|$.
    
\end{restatable}

\begin{proof}
    First, by \Cref{def:value_function}, the value function $\val(\Sset)$ is the expected number of nodes in $V \setminus \seedset$ activated by $\Sset$:
    \begin{equation} 
        \begin{aligned}
            \val_{G,\seedset}(\Sset)
            & = \sum_{v \in V \setminus \seedset} \Pr[ v \text{ is activated by } \Sset ]
        \end{aligned}
        \end{equation}
    Then, By \Cref{lem:rr_equivalence}, for each  $v \in V \setminus \seedset$:
\begin{equation}    
\Pr[ v \text{ is activated by } \Sset ] = \Pr[ \Sset \cap \rrset_v \neq \emptyset ]
\end{equation}
Given that each node $v \in V \setminus \seedset$ has an equal probability of being chosen when generating a random RR set $\rrset$, the probability over all $v \in V \setminus \seedset$ is:
\begin{equation}
\begin{aligned}
    \Pr[\Sset \cap \rrset \neq \emptyset] 
    &= \sum_{v \in \nodeset \setminus \seedset} \Pr[v \text{ is selected}] \cdot \Pr[\Sset \cap \rrset_v \neq \emptyset] \\
    &= \sum_{v \in \nodeset \setminus \seedset} \frac{1}{\rrsize} \cdot \Pr[\Sset \cap \rrset_v \neq \emptyset] \\
    &= \frac{1}{\rrsize} \sum_{v \in \nodeset \setminus \seedset} \Pr[v \text{ is activated by } \Sset] \\
    &= \frac{\val_{G,\seedset}(\Sset)}{\rrsize}
\end{aligned}
\end{equation}
Rearranging the final equation, we obtain \Cref{eq:rr_valuefunction} and complete the proof.
\end{proof}

Then we connect the marginal contribution of a seed node to the probability that the node appears in a random RR set:
\begin{lemma}\label{lem:marginal_rr}
Given ($G(V,E)$, $\seedset$,$\actprob$), for any $\Sset \subseteq \seedset$ and any $\seed \in \seedset \setminus \Sset$, the marginal contribution of $\seed$ to $\Sset$ is:
\begin{equation}\label{eq:marginal_rr}
\val(\Sset \cup { \seed }) - \val(\Sset) = \rrsize \cdot \Pr[ \seed \in \rrset \land \Sset \cap \rrset = \emptyset ]
\end{equation}
\end{lemma}

\begin{proof}
    From Lemma \ref{lem:rr_valuefunction}, we have:
    \begin{align}
    \val(\Sset) & = \rrsize \cdot \Pr[ \Sset \cap \rrset \ne \emptyset ]\\
    \val( \Sset \cup \{ \seed \} ) & = \rrsize \cdot \Pr[ ( \Sset \cup \{ \seed \} ) \cap \rrset \ne \emptyset ]
    \end{align}
    
    Moreover, since $\seed \in \seedset \setminus \Sset$, the events $\Sset \cap \rrset \ne \emptyset$ and $\seed \in \rrset \land \Sset \cap \rrset = \emptyset$ are disjoint. Then:
    \begin{equation}
        \Pr[ ( \Sset \cup \{ \seed \} ) \cap \rrset \ne \emptyset ] = \Pr[ \Sset \cap \rrset \ne \emptyset ] + \Pr[ \seed \in \rrset \land \Sset \cap \rrset = \emptyset ]
    \end{equation}

    Therefore, the marginal contribution is:
    
    \begin{equation}
    \begin{aligned}
    \val( \Sset \cup \{ \seed \} ) - \val(\Sset) &= \rrsize \left( \Pr[ ( \Sset \cup \{ \seed \} ) \cap \rrset \ne \emptyset ] - \Pr[ \Sset \cap \rrset \ne \emptyset ] \right) \\
    &= \rrsize \cdot \Pr[ \seed \in \rrset \land \Sset \cap \rrset = \emptyset ]
    \end{aligned}
    \end{equation}
    
\end{proof}

Using the above results, we prove \Cref{lem:shapley_rr} as follows:
\begin{proof}[Proof of \Cref{lem:shapley_rr}]
    Since seed node are in the intersetction of $\rrset$ and $\seedset$, we replace the RR set $\rrset$ with $\rrset^{\prime} = \rrset \cap \seedset$ in the expression of marginal contribution:
    \begin{equation}
        \Pr[\seed \in \rrset \land S_{\pi,\seed} \cap \rrset = \emptyset] = \Pr[\seed \in \rrset^{\prime} \land S_{\pi,\seed} \cap \rrset^{\prime} = \emptyset] 
    \end{equation}
    
    Applying Lemma \ref{lem:marginal_rr} to the definition of Shapley value (\Cref{def:shapley_influence}), we have:
    \begin{equation}
    \begin{aligned}
    \shap(\seed) & = \mathbb{E}_{\pi \sim \Pi(\seedset)} [\val(S_{\pi,\seed} \cup \{\seed\}) - \val(S_{\pi,\seed})] \\
    & = \mathbb{E}_{\pi \sim \Pi(\seedset)} [ \rrsize  \Pr [\seed \in \rrset^{\prime} \land S_{\pi,\seed} \cap \rrset^{\prime} = \emptyset] ]\\
    &= \mathbb{E}_{\pi}[\rrsize \mathbb{E}_{\rrset^{\prime}}[\mathbb{I}\{t \in \rrset^{\prime} \land S_{\pi,t} \cap \rrset^{\prime} = \emptyset\}]]\\
    &= \rrsize \mathbb{E}_{\rrset^{\prime}}[\mathbb{E}_{\pi}[\mathbb{I}\{t \in \rrset^{\prime} \land S_{\pi,t} \cap \rrset^{\prime} = \emptyset\}]]
    \end{aligned}
    \end{equation}
    
    And the probability that $S_{\pi,\seed} \cap \rrset^{\prime} = \emptyset$ (i.e., $\seed$ is the first node from $\rrset^{\prime}$ in $\pi$) is $\frac{1}{|\rrset^{\prime}|}$ when $\seed \in \rrset^{\prime}$, and $0$ otherwise, because $\pi$ is a uniformly random permutation of $\seedset$ and $\seed$ must be the first node from $\rrset^{\prime}$ in $\pi$ for the event to happen. Therefore,
    \begin{equation}
    \begin{aligned}   
    & \mathbb{E}_{\pi}[\mathbb{I}\{t \in \rrset^{\prime} \land S_{\pi,t} \cap \rrset^{\prime} = \emptyset\}]\\
    &= \mathbb{I}\{\seed \in \rrset^{\prime}\} \cdot \mathbb{E}_{\pi}\left[\mathbb{I}\{S_{\pi,\seed} \cap \rrset^{\prime} = \emptyset\} \mid \seed \in \rrset^{\prime}\right] \\
    &= \frac{\mathbb{I}\{\seed \in \rrset^{\prime}\} }{|\rrset^{\prime}|}
\end{aligned}
    \end{equation}
    
    Combining, we get:
    
    \begin{equation}
    \shap(\seed) = \rrsize \cdot \mathbb{E}_{\rrset^{\prime}}\left[\frac{\mathbb{I}\{\seed \in \rrset^{\prime}\}}{|\rrset^{\prime}|}\right] 
    \end{equation}
\end{proof}

\subsection{Approximation Guarantees}\label{subsec:rrset_approximation_guarantees}
After establishing the relationship between RR sets and the Shapley value, we show that the final estimator $\widehat{\shap}^{(k)}$ of $\algrrset$ (\Cref{alg:shapley_rr}) serves as a good estimate of $\shap^{(k)}$ with high probability, as shown in \Cref{thm:rr_approx}:
\begin{restatable}{proposition}{rrguarantee}\label{thm:rr_approx}
    For any $\varepsilon > 0$, $\ell > 0$ and $k \in [|\seedset|]$, suppose that $\shap^{(k)} \geq 1$, with probability at least $1 - 1/n^{\ell}$, $\algrrset$ returns $\widehat{\shap}(\seed)$ for all $\seed \in \seedset$ such that:
    \begin{small}
    \begin{equation}\label{eq:final_guarantee}
    \begin{cases}
    |\widehat{\shap}(\seed) - \shap(\seed)| \leq \varepsilon \shap(\seed), & \forall \seed \text{ with } \shap(\seed) > \shap^{(k)}, \\
    |\widehat{\shap}(\seed) - \shap(\seed)| \leq \varepsilon \shap^{(k)}, & \forall \seed \text{ with } \shap(\seed) \leq \shap^{(k)}.
    \end{cases}
    \end{equation}
    \end{small}
\end{restatable}

In this section, we first prove that our estimator, which uses the frequency of a seed node's appearance in RR sets, provides an unbiased estimate of the Shapley value (see \Cref{subsec:unbiasedness_of_estimator}).
Then in \Cref{subsec:approximation_guarantees}, we derive concentration bounds to show that our estimator converges to the true Shapley value with high probability, given a sufficient number of RR sets. Since the number of RR sets required depends on the true value of the $k$-th largest Shapley value among all seed nodes, which is unknown, we estimate a lower bound for the $k$-th largest Shapley value using a martingale approach in \Cref{subsec:estimating_lower_bound_for_shap_k}. Finally, we combine the results to show the proof of \Cref{thm:rr_approx} in \Cref{subsec:rr_approx_proof}.

\subsubsection{\textbf{Unbiasedness of the Estimator}}\label{subsec:unbiasedness_of_estimator}

By \Cref{alg:shapley_rr}, our estimator for the Shapley value of each $\seed \in \seedset$ is:

\begin{equation}
\widehat{\shap(\seed)} = \rrsize \cdot \frac{1}{\theta} \sum_{i=1}^{\theta} X_{\rrset_i}(\seed)
\end{equation}

where $\rrsize = |V \setminus \seedset|$, $\theta$ is the total number of RR sets generated, and $\rrset_1, \rrset_2, \dots, \rrset_{\theta}$ are independent random RR sets generated in $G^{\prime}$. The random variable $X_{\rrset_i}(\seed)$ equals to $\frac{1}{|\rrset_i^{\prime}|}$ if $\seed \in \rrset_i^{\prime}$, and $0$ otherwise. We establish the following result:
\begin{restatable}{lemma}{unbiasedestimator}\label{lem:unbiased_estimator}
For any $\seed \in \seedset$, $\algrrset$ returns an unbiased estimator of $\shap(\seed)$, i.e.,
\begin{equation}
    \mathbb{E}[\widehat{\shap(\seed)}] = \shap(\seed)
\end{equation}
\end{restatable}

\begin{proof}
    \sloppy
    Since $X_{\rrset_i}(\seed)$ are independent and identically distributed random variables with expectation $\mathbb{E}[X_{\rrset_i}(\seed)] = \frac{\shap(\seed)}{n'}$ (from \Cref{lem:shapley_rr}) and linearity of expectation, we have:
    
    \begin{equation}
    \begin{aligned}
        \mathbb{E}[\widehat{\shap(\seed)}] & = \mathbb{E}[\rrsize \cdot \frac{1}{\theta} \sum_{i=1}^{\theta} X_{\rrset_i}(\seed)]\\
        & = \rrsize \cdot \frac{1}{\theta}\sum_{i=1}^{\theta} \mathbb{E}[X_{\rrset_i}(\seed)]\\
        & = \rrsize \cdot \frac{1}{\theta} \sum_{i=1}^{\theta} \frac{\shap(\seed)}{n^{\prime}}\\
        & = \shap(\seed)
    \end{aligned}
    \end{equation}
\end{proof}

\subsubsection{\textbf{Approximation Bound}}\label{subsec:approximation_guarantees} The estimator provides accurate approximations with high probability when a sufficient number of RR sets are generated, as shown in the following \Cref{lem:approximation_guarantee}:
\begin{restatable}{lemma}{approxguarantee}\label{lem:approximation_guarantee}
    For any $\varepsilon > 0$, $\ell > 0$, and $k \in [ | \seedset | ]$. If the number of RR sets $\theta$ satisfies:
    %Suppose \Cref{alg:shapley_rr} generates
    \begin{equation}\label{eq:theta_value}
        \theta \geq \frac{ \rrsize \left( \ell \ln \rrsize + \ln |\seedset| + \ln 4 \right) \left( 2 + \frac{2}{3} \varepsilon \right) }{ \varepsilon^2 \shap^{(k)} }
    \end{equation}
    \sloppy
    where $\shap^{(k)}$ is the $k$-th largest Shapley value among all $\{ \shap(\seed) \}_{ \seed \in \seedset }$. Then, the following holds with probability at least $1 - \frac{1}{2 {\rrsize}^\ell}$:
    \begin{enumerate}
    \item For all $\seed \in \seedset$ with $\shap(\seed) > \shap^{(k)}$,
    \begin{equation}
    \left| \widehat{ \shap(\seed) } - \shap(\seed) \right| \leq \varepsilon \shap(\seed)
    \end{equation}
    \item For all $\seed \in \seedset$ with $\shap(\seed) \leq \shap^{(k)}$,
    \begin{equation}
    \left| \widehat{ \shap(\seed) } - \shap(\seed) \right| \leq \varepsilon \shap^{(k)}
    \end{equation}
    \end{enumerate}
\end{restatable}

First, we present the following Chernoff bounds \cite{chung2006concentration} that will be used in our analysis:

\begin{fact}[Chernoff Bounds]
Let $\boldsymbol{Y}$ be the sum of $t$ i.i.d. random variables with mean $\mu$ and within the range $[0, 1]$. For any $\delta > 0$, the following upper tail bound holds:
\begin{equation}
\Pr\{\boldsymbol{Y} - t\mu \geq \delta \cdot t\mu \} \leq \exp \left( - \frac{\delta^2}{2 + \frac{2}{3} \delta} t\mu \right).
\end{equation}
For any $0 < \delta < 1$, the lower tail bound is given by:
\begin{equation}
\Pr\{\boldsymbol{Y} - t\mu \leq - \delta \cdot t\mu \} \leq \exp \left( - \frac{\delta^2}{2} t\mu \right).
\end{equation}
\end{fact}

\begin{proof}[Proof of \Cref{lem:approximation_guarantee}]\label{app:proof_approximation_guarantee}
    Since $\mathbb{E}[X_{\rrset_i}(\seed)] = \frac{\shap(\seed)}{n^{\prime}}$ and the value range of $X_{\rrset_i}(\seed)$ is $[0,1]$ by \Cref{lem:shapley_rr}, we can apply the Chernoff bounds for all $\seed \in \seedset$. We discuss by two cases: $\shap(\seed) > \shap^{(k)}$ and $\shap(\seed) \leq \shap^{(k)}$.
    First, for every $\seed \in \seedset$ with $\shap(\seed) > \shap^{(k)}$:
    %\begin{equation}
    \begin{align}
    &\Pr[|\widehat{\shap}(\seed) - \shap(\seed)| \geq \varepsilon \cdot \shap(\seed)] \\
    %= & Pr[|n' \cdot \frac{\mathbf{est}_{\theta}^\theta}{\theta} - \shap(\seed)| \geq \varepsilon \cdot \shap(\seed)] \\
    = & Pr[|\frac{\rrsize}{\theta} \sum_{j=1}^{\theta} X_{\rrset_j}(\seed) -  \shap(\seed)| \geq \varepsilon  \cdot \shap(\seed)] \\
    = & Pr[|\sum_{j=1}^{\theta} X_{\rrset_j}(\seed) - \theta \cdot \frac{\shap(\seed)}{\rrsize} | \geq  \varepsilon \cdot \theta \cdot \frac{\shap(\seed)}{\rrsize}] \\
    = & Pr[|\sum_{j=1}^{\theta} X_{\rrset_j}(\seed) - \theta \cdot \mathbb{E}[X_{\rrset_j}(\seed)]| \geq  \varepsilon \cdot \theta \cdot \mathbb{E}[X_{\rrset_j}(\seed)]] \\ %\quad (\text{By \Cref{lem:shapley_rr}})\\
    \leq& 2 \exp(-\frac{\varepsilon^2}{2+\frac{2}{3}\varepsilon} \cdot \frac{\theta}{\rrsize} \cdot \shap(\seed)) \quad (\text{By the Chernoff Bound})\\
    \leq& 2 \exp(-\frac{\varepsilon^2}{2+\frac{2}{3}\varepsilon} \cdot \frac{\shap(\seed)}{\rrsize} \cdot \frac{ \rrsize \left( \ell \ln \rrsize + \ln |\seedset| + \ln 4 \right) \left( 2 + \frac{2}{3} \varepsilon \right) }{ \varepsilon^2 \shap^{(k)} })\\
    & \quad (\text{By \Cref{eq:theta_value}}) \nonumber\\
    =& 2\exp(-\frac{\shap(\seed)}{\shap^{(k)}}\cdot (\ell \ln \rrsize + \ln |\seedset| + \ln 4))\\
    \leq & 2 \exp(- \ell \ln \rrsize - \ln |\seedset| - \ln 4)  \quad (\text{By }\shap(\seed)>\shap^{(k)})\\
    = & \frac{1}{2|\seedset|{\rrsize}^{\ell}}
    \end{align}
    %\end{equation}
    
    Then, for every $\seed \in \seedset$ with $\shap(\seed) \leq \shap^{(k)}$, we have:
    %\begin{equation}
    \begin{align}
    &\Pr[|\hat{\shap}(\seed) - \shap(\seed)| \geq \varepsilon \cdot \shap^{(k)}] \\
    %= & Pr[|n' \cdot \frac{\mathbf{est}_{\theta}^\theta}{\theta} - \shap(\seed)| \geq \varepsilon \cdot \shap^{(k)}] \\
    = & Pr[|\sum_{j=1}^{\theta} X_{\rrset_j}(\seed) - \theta \cdot \frac{\shap(\seed)}{\rrsize} | \geq  (\varepsilon \cdot \frac{\shap^{(k)}}{\shap(\seed)}) \cdot (\theta \cdot \frac{\shap(\seed)}{\rrsize})]\\
    %= & Pr[|\sum_{j=1}^{\theta} \mathbf{X}_{\mathbf{R}_j}(\seed) - \theta \cdot \mathbb{E}[X_{R_i}(\seed)]| \geq  \varepsilon \cdot \theta \cdot \mathbb{E}[X_{R_i}(\seed)]] \quad (\text{By \Cref{lem:shapley_rr}})\\
    \leq& 2 \exp(-\frac{(\varepsilon \cdot \frac{\shap^{(k)}}{\shap(\seed)})^2}{2+\frac{2}{3}(\varepsilon \cdot \frac{\shap^{(k)}}{\shap(\seed)})} \cdot \frac{\theta}{\rrsize} \cdot \shap(\seed)\\
    & (\text{By the Chernoff Bound}) \nonumber\\
    = & 2\exp(- \frac{\varepsilon^2(\shap^{(k)})^2}{\rrsize(2\shap(\seed)+\frac{2}{3}\varepsilon \shap^{(k)})})\cdot \theta)\\
    \leq & 2\exp(- \frac{\varepsilon^2\shap^{(k)}}{\rrsize(2+\frac{2}{3}\varepsilon)} \cdot \theta) \quad (\shap(\seed) \leq \shap^{(k)})\\
    \leq & 2\exp(- \frac{\varepsilon^2\shap^{(k)}}{\rrsize(2+\frac{2}{3}\varepsilon)} \cdot \frac{ \rrsize \left( \ell \ln \rrsize + \ln |\seedset| + \ln 4 \right) \left( 2 + \frac{2}{3} \varepsilon \right) }{ \varepsilon^2 \shap^{(k)} })\\
    & (\text{By \Cref{eq:theta_value}}) \nonumber\\
    \leq & 2 \exp(- \ell \ln \rrsize - \ln |\seedset| - \ln 4)\\
    = & \frac{1}{2|\seedset|{\rrsize}^{\ell}}
    \end{align}
    %\end{equation}

    Last, we apply the union bound to the two cases to obtain the final result that for all $\seed \in \seedset$:
    \begin{align}
        \Pr&\left[ \text{Approximation guarantee fails for some } t \in T \right] \nonumber \\
        &\leq |\seedset| \cdot \dfrac{ 1 }{ 2 |\seedset|{\rrsize}^{\ell} } = \dfrac{ 1 }{ 2 {\rrsize}^{\ell} }
    \end{align}
    %Therefore, we obatin the number of samples $\theta$ required using the $k$-th largest Shapley value among all Shapley values of seed set as
\end{proof}

\subsubsection{\textbf{Estimating a Lower Bound for \texorpdfstring{$\shap^{(k)}$}{shap(k)}}}\label{subsec:estimating_lower_bound_for_shap_k}
Since $\shap^{(k)}$ is not known in advance, we estimate a lower bound $\boldsymbol{LB}$ for $\shap^{(k)}$ in $\textsc{EstimateThreshold}$ (\Cref{alg:estimate_threshold_complete}). 

%\textbf{Martingale Definition and Properties.} 
First, we construct a martingale sequence.
\begin{definition}[Martingale]
    A sequence of random variables $\{Y_i(t)\}_{i \geq 1}$ is a martingale if and only if that for all $i \geq 1$, (1) $\mathbb{E}[|Y_i|]<\infty$; (2) $\mathbb{E}\left[Y_i(t) \,|\, Y_1(t), Y_2(t), \cdots, Y_{i-1}(t)\right] = Y_{i-1}(t)$.
\end{definition}
 
\begin{lemma}
    \label{lem:martingale}
    Let $\theta^{\prime}$ be the number of RR sets generated in Phase 1, and let $\rrset^{(1)}_1, \rrset^{(1)}_2, \dots, \rrset^{(1)}_{\theta^{\prime}}$ be these RR sets. For every $\seed \in \seedset$ and every $i \geq 1$,
    \begin{equation}
    \mathbb{E}\left[X_{{\rrset^{(1)}_i}}(t) \,|\,  X_{\rrset^{(1)}_2}(t), \cdots, X_{\rrset^{(1)}_{i-1}}(t)\right] = \dfrac{\shap(t)}{n'},
    \end{equation}
    
    Define the sequence $\{Y_i(t)\}_{i \geq 1}$ as the following:
    \begin{equation}
    \boldsymbol{Y}_i(t) = \sum_{j = 1}^{i} \left(X_{\rrset^{(1)}_j}(t) - \dfrac{\shap(t)}{n'}\right).
    \end{equation}
    Then for every $\seed \in \seedset$, $\{Y_i(t), i \geq 1\}$ is a martingale.
    \end{lemma}
    
\begin{proof}

Since the generation process of each RR set independent, each $X_{\rrset^{(1)}_i}(\seed)$ is independent of the previous RR sets. Then for every $\seed \in \seedset$ and $i \geq 1$, we have:
\begin{equation}\label{eq:martingale_x}
\mathbb{E}\left[X_{\rrset^{(1)}_i}(\seed) \,|\,X_{\rrset^{(1)}_2}(\seed), \cdots, X_{\rrset^{(1)}_{i-1}}(\seed) \right] = \dfrac{\shap(\seed)}{\rrsize}
\end{equation}

By the definition of $\boldsymbol{Y}_i(\seed)$, we know that first, the value range of $\boldsymbol{Y}_i(\seed)$ is $[-i,i]$. Second, 
    \begin{equation}
    \boldsymbol{Y}_i(\seed) = \boldsymbol{Y}_{i-1}(\seed) + \left(X_{\rrset^{(1)}_i}(\seed) - \dfrac{\shap(\seed)}{\rrsize}\right)
    \end{equation}
    
By \Cref{eq:martingale_x}, we have:
\begin{equation}
\begin{aligned}
    &\mathbb{E}\left[\boldsymbol{Y}_i(\seed) \,|\, \boldsymbol{Y}_{i-1}(\seed), \cdots, \boldsymbol{Y}_1(\seed)\right] \\
    &= \boldsymbol{Y}_{i-1}(\seed) + \mathbb{E}\left[X_{\rrset^{(1)}_i}(\seed) \,|\, X_{\rrset^{(1)}_2}(\seed), \cdots, X_{\rrset^{(1)}_{i-1}}(\seed)\right] - \dfrac{\shap(\seed)}{\rrsize} \\
    &= \boldsymbol{Y}_{i-1}(\seed) + \dfrac{\shap(\seed)}{\rrsize} - \dfrac{\shap(\seed)}{\rrsize} \\
    &= \boldsymbol{Y}_{i-1}(\seed)
\end{aligned}
\end{equation}
Therefore, we prove that $\{Y_i(\seed),i \geq 1\}$ is a martingale.    
\end{proof}

%\textbf{Martingale Concentration Inequality.~} 
Since $X_{R^{(1)}_i}(t) \in [0,1]$ and $\dfrac{\shap(t)}{n'} \in [0,1]$, the martingale sequence satisfies $|X_{R^{(1)}_i}(t) - \dfrac{\shap(t)}{n'}| \leq 1$. So we can apply the following tail bounds:

\begin{fact}[Martingale Tail Bounds]
    \label{fact:martingale_tail_bounds}
    Let $\boldsymbol{X}_1, \boldsymbol{X}_2, \ldots, \boldsymbol{X}_t$ be a sequence of random variables such that (1) the value range is $[0,1]$ for each $X_i$; (2) for some $\mu \in[0,1], \mathbb{E}\left[\boldsymbol{X}_i \mid \boldsymbol{X}_1, \boldsymbol{X}_2, \ldots, \boldsymbol{X}_{i-1}\right]=\mu$ for every $i \in[t]$. Let $\boldsymbol{Y}=\sum_{i=1}^t \boldsymbol{X}_i$. For any $\delta>0$, we have:

\begin{equation}
\operatorname{Pr}\{\boldsymbol{Y}-t \mu \geq \delta \cdot t \mu\} \leq \exp \left(-\frac{\delta^2}{2+\frac{2}{3} \delta} t \mu\right)
\end{equation}

For any $0<\delta<1$, we have
\begin{equation}
\operatorname{Pr}\{\boldsymbol{Y}-t \mu \leq-\delta \cdot t \mu\} \leq \exp \left(-\frac{\delta^2}{2} t \mu\right)
\end{equation}

\end{fact}

%\subsubsection{Analysis of the Estimation Algorithm}

Then, the following lemma shows that the estimator $\boldsymbol{est}_i^{(k)}$ is a good estimate of $\shap^{(k)}$ for both case that $x_i$ is greater than or less than $\shap^{(k)}$.

\begin{restatable}{lemma}{lowerboundestimation}\label{lem:lower_bound_estimation}
    For each $i = 1, 2, \cdots, \lfloor\log_2 \rrsize\rfloor-1$, provided that $\theta_i$ satisfies:
    \begin{equation} \label{eq:theta_i}
     \theta_i \geq \frac{\rrsize(2+\frac{2}{3}\varepsilon')}{{\varepsilon'}^2\cdot x_i}\left(\ell \ln \rrsize + \ln |\seedset| + \ln (\log_2 \rrsize) + \ln 2 \right)
    \end{equation}
    \begin{enumerate}[leftmargin=*]
        \item  If $x_i = \frac{\rrsize}{2^i} > \shap^{(k)}$, then with probability at least $1 -\frac{1}{2 {\rrsize}^{\ell} \log_2 \rrsize}$, $\frac{\rrsize \cdot \boldsymbol{est}_i^{(k)}}{ \theta_i} < (1 + \varepsilon^{\prime}) \cdot x_i$, i.e., $\widehat{\shap}_i^{(k)} < (1 + \varepsilon^{\prime}) x_i$.
        \item  If $x_i = \frac{\rrsize}{2^i} \leq \shap^{(k)}$, then with probability at least $1 -\frac{1}{2{\rrsize}^{\ell} \log_2 \rrsize}$, $\frac{\rrsize \cdot \boldsymbol{est}_i^{(k)}}{ \theta_i} < (1 + \varepsilon^{\prime}) \cdot \shap^{(k)}$, i.e., $\widehat{\shap}_i^{(k)} \geq (1 - \varepsilon') x_i$
    \end{enumerate}
\end{restatable}
  
\begin{proof}

    Let $\boldsymbol{R}_1^{(1)}, \boldsymbol{R}_2^{(1)}, \ldots, \boldsymbol{R}_{\theta_i}^{(1)}$ be the $\theta_i$ generated RR sets by the end of the $i$-th iteration of the for-loop at $textsc{EstimateThreshold}$. By \Cref{lem:martingale}, we can apply the martingale tail bound of Fact 29 on the sequence $\{X_{\rrset_1^{(1)}}(\seed), X_{\rrset_2^{(1)}}(\seed), \cdots, X_{\rrset_i^{(1)}}(\seed)\}$ for each $\seed \in \seedset$.
    
    Denote $\boldsymbol{est}_{\seed,i}$ as the value of $\boldsymbol{est}_{\seed}$ in the i-th iteration of the same for-loop. Then we have $\boldsymbol{est}_{\seed,i} = \sum_{j=1}^{\theta_i} X_{\rrset_j^{(1)}}(\seed)$. Denote $\widehat{\shap}_i(\seed) = \dfrac{\rrsize \cdot \boldsymbol{est}_{\seed,i}}{ \theta_i}$ as the estimator of $\shap(\seed)$ at the end of the $i$-th iteration.

\paragraph{\textbf{Case 1:}} Consider $x_i > \shap^{(k)}$. First, for every $\seed \in \seedset$ such that $\shap(\seed) \leq \shap^{(k)}$, we have:
\vspace{0em}
\begin{equation} \label{eq:lower_bound_estimation_case1}
    %\small
 \begin{aligned}
    & \Pr\left\{\frac{\rrsize \cdot \boldsymbol{est}_{\seed,i}}{ \theta_i} \geq (1 + \varepsilon^{\prime}) \cdot x_i\right\} \\
    & =\Pr\left\{\boldsymbol{est}_{\seed,i} \geq (1 + \varepsilon^{\prime}) \frac{\theta_i \cdot x_i}{\rrsize}\right\} \\
    & = \Pr\left\{\boldsymbol{est}_{\seed,i} - \theta_i \cdot \frac{x_i}{\rrsize} \geq \varepsilon^{\prime}\frac{\theta_i \cdot x_i}{\rrsize}\right\} \\
    & \leq \Pr\left\{\boldsymbol{est}_{\seed,i} - \theta_i \cdot \frac{\shap(\seed)}{\rrsize} \geq (\varepsilon^{\prime}\cdot \frac{x_i}{\shap(\seed)}) \frac{\theta_i \cdot \shap(\seed)}{\rrsize}\right\} (\text{By }x_i > \shap(\seed)) \\
    & \leq \exp \left(-\frac{\left(\varepsilon^{\prime} \cdot \frac{x_i}{\shap(\seed)}\right)^2}{2+\frac{2}{3}\left(\varepsilon^{\prime} \cdot \frac{x_i}{\shap(\seed)} \right)} \cdot \theta_i \cdot \frac{\shap(\seed)}{\rrsize}\right) \quad (\text{By \Cref{fact:martingale_tail_bounds}}) \\
    & =\exp \left(-\frac{{\varepsilon^{\prime}}^2 \cdot x_i^2}{2 \shap(\seed)+\frac{2}{3}\varepsilon^{\prime} x_i} \cdot \frac{\theta_i}{\rrsize}\right) \\
    & \leq \exp \left(-\frac{\varepsilon^{\prime 2} \cdot x_i}{2+\frac{2}{3} \varepsilon^{\prime}} \cdot \frac{\theta_i}{\rrsize}\right) \quad (\text {By } x_i > \shap^{(k)} \geq \shap(\seed)) \\
    & \leq \frac{1}{2|\seedset| \cdot {\rrsize}^{\ell} \log_2 \rrsize} \quad (\text{By \Cref{eq:theta_i}})\\
\end{aligned}
\end{equation}

%\normalsize
Note that $\boldsymbol{est}_i^{(k)}$ is the $k$-th largest value among $\{\boldsymbol{est}_{\seed,i}\}_{\seed \in \seedset}$. There are at most $|\seedset|-k$ nodes $\seed$ with $\shap(\seed) < \shap^{(k)}$. This implies that there is at least one node $\seed$ with $\shap(\seed) \geq \shap^{(k)}$ and $\boldsymbol{est}_{\seed,i} \leq \boldsymbol{est}_i^{(k)}$. More precisely, there are at most $k$ such nodes $\seed$ with $\shap(\seed)$ ranked at or above $\shap^{(k)}$. Therefore, we have:

\begin{equation}
\begin{aligned}
&\Pr\left\{\frac{\rrsize \cdot \boldsymbol{est}_{i}^{(k)}}{\theta_i} \leq (1 + \varepsilon^{\prime}) \cdot x_i\right\} \\
& \leq \Pr\left\{\exists \seed \in \seedset, \shap(\seed) \leq \shap^{(k)}, \frac{\rrsize \cdot \boldsymbol{est}_{\seed,i}}{\theta_i} \geq (1 + \varepsilon^{\prime}) \cdot x_i\right\} \\
& \leq \frac{k}{2 |\seedset| {\rrsize}^{\ell} \log_2 \rrsize}
\end{aligned}
\end{equation}

\paragraph{\textbf{Case 2:}} Consider $x_i \leq \shap^{(k)}$. First, for every $\seed \in \seedset$ such that $\shap(\seed) \leq \shap^{(k)}$, we have:

\begin{equation} \label{eq:lower_bound_estimation_case2}
\begin{aligned}
& \Pr\left\{\frac{\rrsize \cdot \boldsymbol{est}_{\seed,i}}{\theta_i} \geq (1+ \varepsilon^{\prime}) \cdot \shap^{(k)}\right\} \\
& = \Pr\left\{\boldsymbol{est}_{\seed,i} - \theta_i \cdot \frac{\shap^{(k)}}{\rrsize} \leq \varepsilon^{\prime}\frac{\theta_i \cdot \shap^{(k)}}{\rrsize}\right\} \\
& \leq \Pr\left\{\boldsymbol{est}_{\seed,i} - \theta_i \cdot \frac{\shap(\seed)}{\rrsize} \leq (\varepsilon^{\prime}\cdot \frac{\shap^{(k)}}{\shap(\seed)}) \frac{\theta_i \cdot \shap(\seed)}{\rrsize}\right\} \\
& \quad (\text{By } \shap(\seed) \leq \shap^{(k)}) \\
& \leq \exp \left(-\frac{\left(\varepsilon^{\prime} \cdot \frac{\shap^{(k)}}{\shap(\seed)}\right)^2}{2+\frac{2}{3}\left(\varepsilon^{\prime} \cdot \frac{\shap^{(k)}}{\shap(\seed)} \right)} \cdot \theta_i \cdot \frac{\shap(\seed)}{\rrsize}\right) \quad (\text{By \Cref{fact:martingale_tail_bounds}}) \\
& = \exp \left(-\frac{\varepsilon^{\prime 2} \cdot (\shap^{(k)})^2}{2 \shap(\seed)+\frac{2}{3}\varepsilon^{\prime} \shap^{(k)}} \cdot \frac{\theta_i}{\rrsize}\right) \\
& \leq \exp \left(-\frac{\varepsilon^{\prime 2} \cdot \shap^{(k)}}{2+\frac{2}{3}\varepsilon^{\prime}} \cdot \frac{\theta_i}{\rrsize}\right) \quad{(\text{By } \shap(\seed) \leq \shap^{(k)})} \\
& \leq exp \left(-\frac{\varepsilon^{\prime 2} \cdot x_i}{2+\frac{2}{3}\varepsilon^{\prime}} \cdot \frac{\theta_i}{\rrsize}\right) \quad{(\text{By } x_i \leq \shap^{(k)})} \\
& \leq \frac{1}{2|\seedset| \cdot {\rrsize}^{\ell} \log_2 \rrsize} \quad (\text{By \Cref{eq:theta_i}})\\
\end{aligned}
\end{equation}                             
Similarly, by taking the union bound, we have

\begin{equation}
\begin{aligned}
& \Pr\left\{\frac{\rrsize \cdot \boldsymbol{est}_{i}^{(k)}}{\theta_i} \geq (1+ \varepsilon^{\prime}) \cdot \shap^{(k)}\right\} \\
& \leq \Pr\left\{\exists \seed \in \seedset, \shap(\seed) \leq \shap^{(k)}, \frac{\rrsize \cdot \boldsymbol{est}_{\seed,i}}{\theta_i} \geq (1+ \varepsilon^{\prime}) \cdot \shap^{(k)}\right\} \\
& \leq \frac{1}{2 {\rrsize}^{\ell} \log_2 \rrsize}
\end{aligned}
\end{equation}

\end{proof}

%\subsubsection{Establishing the Lower Bound $\boldsymbol{LB}$}
Next, we establish the lower bound $\boldsymbol{LB}$ in $\textsc{EstimateThreshold}$ (\Cref{alg:estimate_threshold_complete}).

\begin{restatable}{lemma}{lbcorrectness}\label{lem:lb_correctness}
    Suppose that $\shap^{(k)} \geq 1$. With probability at least $1 - \frac{1}{2{\rrsize}^{\ell}}$, $\boldsymbol{LB}$ computed at the end of Phase 1 satisfies $\boldsymbol{LB} \leq \shap^{(k)}$.
\end{restatable}

\begin{proof}
    Let $\boldsymbol{LB}_i = \rrsize \cdot \frac{\boldsymbol{est}_i^{(k)}}{\theta_i \cdot\left(1+\varepsilon^{\prime}\right)}$. \\
    \textbf{Case 1:} If $\shap^{(k)} \geq x_{\lfloor\log_2 \rrsize\rfloor-1}$, let $i^*$ be the smallest index such that $x_{i^*} \leq \shap^{(k)}$. Thus for all iterations $i < i^*$, we have $x_i > \shap^{(k)}$. Applying \Cref{lem:lower_bound_estimation} (1) on each $i \leq i^*$:
    \begin{equation}
        \Pr\left\{\frac{\rrsize \cdot \boldsymbol{est}_{i}^{(k)}}{\theta_i} < (1 + \varepsilon^{\prime}) \cdot x_i\right\}  \geq 1- \frac{1}{2 {\rrsize}^{\ell} \log_2 \rrsize}
    \end{equation}
    Taking the union bound over all iterations $i < i^*$, we have:
    \begin{equation}
        \Pr\left\{\frac{\rrsize \cdot \boldsymbol{est}_{i}^{(k)}}{\theta_i} < (1 + \varepsilon^{\prime}) \cdot x_i \text{ for all } i < i^*\right\} \geq 1- \frac{i-1}{2 {\rrsize}^{\ell} \log_2 \rrsize}
    \end{equation}
    This means that with probability at least $1 - \frac{i-1}{2 {\rrsize}^{\ell} \log_2 \rrsize}$, the algorithm will not break before the $i^*$-th iteration. Therefore, $\boldsymbol{LB} = \boldsymbol{LB}_{i}$ for some $i \geq i^*$ or $\boldsymbol{LB} = 1$.

    Then for every $i \geq i^*$, we have $x_i \leq \shap^{(k)}$. Applying \Cref{lem:lower_bound_estimation} (Case 2) on each $i \geq i^*$:
    \begin{equation}
    \begin{aligned}
        &\Pr\left\{\frac{\rrsize \cdot \boldsymbol{est}_{i}^{(k)}}{\theta_i} < (1 + \varepsilon^{\prime}) \cdot \shap^{(k)}\right\} \\
        &= \Pr\left\{\boldsymbol{LB}_i < \shap^{(k)}\right\} \\
        &\geq 1- \frac{1}{2 {\rrsize}^{\ell} \log_2 \rrsize}
    \end{aligned}
    \end{equation}
    Taking the union bound again, we have:
    \begin{equation}
        \Pr\left\{\boldsymbol{LB} < \shap^{(k)} \right\} \geq 1- \frac{1}{2 {\rrsize}^{\ell}}
    \end{equation}

    \textbf{Case 2:} If $\shap^{(k)} < x_{\lfloor\log_2 \rrsize\rfloor-1}$, we use the similar argument as the above and we can show that with probability at least $1 - \frac{1}{2 {\rrsize}^{\ell}}$, the for-loop would not break at any iteration and thus $\boldsymbol{LB} = 1$ as the initail value. Since $\shap^{(k)} \geq 1$, we still have $\boldsymbol{LB} \leq \shap^{(k)}$.
\end{proof}

\subsubsection{\textbf{Proof of \Cref{thm:rr_approx}}}\label{subsec:rr_approx_proof}
%Finally, we conclude the proof for \Cref{thm:rr_approx}:
%\rrguarantee*
\begin{proof}[Proof of \Cref{thm:rr_approx}]
    By Lemma~\ref{lem:lb_correctness}, with probability at least $1 - \frac{1}{2{\rrsize}^{\ell}}$, we have $\boldsymbol{LB} \leq \shap^{(k)}$. From the approximation guarantee established in \Cref{lem:approximation_guarantee}, with $\theta$ set appropriately using $\boldsymbol{LB}$, we have that, with probability at least  $1 - \frac{1}{2{\rrsize}^{\ell}}$, the estimates $\widehat{\shap}(\seed)$ satisfy the desired bounds \Cref{eq:final_guarantee}. 
    Then, we use $\boldsymbol{LB}$ in place of $\shap^{(k)}$ to determine the number of RR sets $\theta$ in Phase 2 of the algorithm, as per Equation~\eqref{eq:theta_value}. By the union bound, we know that with probability at least $1 - \frac{1}{{\rrsize}^{\ell}}$, the estimates $\widehat{\shap}(\seed)$ satisfy \Cref{eq:final_guarantee}.

\end{proof}

%\subsection{Proof of \Cref{prop:overall_expected_running_time}}\label{subsec:rr_time_complexity_appendix}

\subsection{Runtime Complexity}\label{subsec:rr_time_complexity_appendix}
%The following %proposition 
%shows that expected runtime complexity of $\algrrset$ is the product of the total number of RR sets generated in two phases, and the expected cost of generating one RR set.
Last, we analyze the time complexity of $\algrrset$. 
%We establish upper bounds on the expected running time by analyzing the number of RR sets generated in both Phase 1 and Phase 2, and the computational cost associated with generating and processing each RR set. The proof follows \cite{tang2015influence,chen2017interplay}, with adjustment to our problem.%\fs{The analysis follows the similar structure as \cite{tang2015influence,chen2017interplay}, but since the different expression of the number of RR sets generated in each phase, the results are slightly different.}
We first establish \Cref{lem:expected_running_time} that relates the expected running time of our algorithm to the expected number of RR sets generated and the expected time to generate an RR set. Next, as the time complexity of generating a single RR set is related to the expected width $EPT$ of a random RR set, we relate $EPT$ of a random RR set to the influence spread of a single node in \Cref{lem:ept}. Then we derive the bounds on the expected number of RR sets generated in both Phase 1 ($\boldsymbol{\theta}^{\prime}$) and Phase 2 ($\boldsymbol{\theta}$) in \Cref{lem:theta_bound}. Combining the results from above lemmas, we obtain the overall expected running time of our algorithm in \Cref{prop:overall_expected_running_time}. The proof follows \cite{tang2015influence,chen2017interplay}, with adjustment to our problem.

A random variable $\boldsymbol{\tau}$ is a \emph{stopping time} for martingale $\left\{\boldsymbol{Y}_i, i \geq 1\right\}$ if $\tau$ takes positive integer values, and the event $\boldsymbol{\tau}=i$ depends only on the values of $\boldsymbol{Y}_1, \boldsymbol{Y}_2, \ldots, \boldsymbol{Y}_i$.

\begin{fact}[Martingale Stopping Theorem]\label{fact:martingale_stopping_theorem}
Suppose that $\left\{\boldsymbol{Y}_i, i \geq 1\right\}$ is a martingale and $\boldsymbol{\tau}$ is a stopping time for $\left\{\boldsymbol{Y}_i, i \geq 1\right\}$. If $\tau \leq c$ for some constant $c$ independent of $\left\{\boldsymbol{Y}_i, i \geq 1\right\}$, then $\mathbb{E}\left[\boldsymbol{Y}_\tau\right]=\mathbb{E}\left[\boldsymbol{Y}_1\right]$.
\end{fact}

%For the Independent Cascade model, the expected time to generate a random RR set is O(m'/n'), where m' is the number of edges and n' is the number of nodes in G' (excluding the seed set).
Given a fixed subset $R \subseteq V$, let the width of $R$, denoted $\omega(R)$, be the total in-degrees of nodes in $R$. The time complexity to generate the random $R R$ set $\boldsymbol{R}$ is $\Theta(\omega(\boldsymbol{R})+1)$. We leave the constant 1 in the above formula because $\omega(\boldsymbol{R})$ could be less than 1 or even $o(1)$ when $m<n$, while $\Theta(1)$ time is needed just to select a random root. The expected time complexity to generate a random RR set is $\Theta(\mathbb{E}[\omega(\boldsymbol{R})]+1)$.

Let $EPT=\mathbb{E}[\omega(\boldsymbol{R})]$ be the expected width of a random $R R$ set. Let $\boldsymbol{\theta}^{\prime}$ be the random variable denoting the number of RR sets generated in Phase 1.

\begin{restatable}{lemma}{expectedrunningtime}\label{lem:expected_running_time}
The expected running time of \Cref{alg:shapley_rr} is:
\begin{equation}
\Theta\left( \left( \mathbb{E}[ \boldsymbol{\theta}^{\prime} ] + \mathbb{E}[ \boldsymbol{\theta} ] \right) \cdot ( \text{EPT} + 1 ) \right)
\end{equation}
where:
\begin{itemize}
\item $\boldsymbol{\theta}^{\prime}$ is the number of RR sets generated in Phase 1.
\item $\boldsymbol{\theta}$ is the number of RR sets generated in Phase 2.
\item $\text{EPT} = \mathbb{E}[ \omega( \boldsymbol{R} ) ]$ is the expected width of a random RR set $\boldsymbol{R}$, with $\omega( \boldsymbol{R} )$ being the total in-degree of nodes in $\boldsymbol{R}$.
\end{itemize}
\end{restatable}

\begin{proof}
The proof contains the following 4 steps:

\textbf{(1) Time Complexity of Generating an RR Set}

For each RR set $\rrset$, the time to generate $\rrset$ is $\Theta( \omega( \rrset ) + 1 )$, as the time to select a random root node and perform BFS to find the RR set. Note that the constant term $\Theta(1)$ is not absorbed by $\omega( \rrset )$ since $\omega( \rrset )$ because the width of the RR set could be less than $1$.

\textbf{(2) Analysis of $\textsc{EstimateThreshold}$}

Let ${\rrset}^{(1)}_1, {\rrset}^{(1)}_2, \dots, {\rrset}^{(1)}_{\boldsymbol{\theta}^{\prime}}$ be the RR sets generated in Phase 1. For each RR set ${\rrset}^{(1)}_j$, we need to update $\boldsymbol{est}_{\seed}$ for each $\seed \in \seedset$ that appears in ${\rrset}^{(1)}_j$, which takes $\Theta(|{\rrset}^{(1)}_j|)$ time. Since $\seedset \subseteq \nodeset$, and the size of ${\rrset}^{(1)}_j$ is at most $\omega( {\rrset}^{(1)}_j ) + 1$ (because the induced subgraph is weakly connected), the algorithm takes $\Theta(\omega( {\rrset}^{(1)}_j ) + 1 + |{\rrset}^{(1)}_j|) = \Theta(\omega( {\rrset}^{(1)}_j ) + 1)$ for each RR set. Therefore, summing up for all $\boldsymbol{\theta}^{\prime}$ RR sets, the total expected running time of Phase 1 is:

\begin{equation}
\Theta\left( \sum_{j = 1}^{\boldsymbol{\theta}^{\prime}} \left( \omega( {\rrset}^{(1)}_j ) + 1 \right) \right)
\end{equation}

Define $\boldsymbol{W}_i = \sum_{j = 1}^{i} \left( \omega( {\rrset}^{(1)}_j ) - EPT \right)$ for $i \geq 1$. By an argument similar to that in \Cref{lem:martingale}, $\{ \boldsymbol{W}_i, i \geq 1 \}$ is a martingale. The stopping time $\boldsymbol{\theta}^{\prime}$ is upper bounded by a constant $\theta_{\left\lfloor \log_2 \rrsize \right\rfloor - 1}$, and depends only on the RR sets already generated. Therefore, by \Cref{fact:martingale_stopping_theorem}:

\begin{equation}
\mathbb{E}[ \boldsymbol{W}_{\boldsymbol{\theta}^{\prime}} ] = \mathbb{E}[ \boldsymbol{W}_1 ] = 0
\end{equation}

Thus:

\begin{equation}
\mathbb{E}\left[ \sum_{j = 1}^{\boldsymbol{\theta}^{\prime}} \omega( {\rrset}^{(1)}_j ) \right] - \mathbb{E}[ \boldsymbol{\theta}^{\prime} ] \cdot EPT = 0
\end{equation}

Rearranging:

\begin{equation}
\mathbb{E}\left[ \sum_{j = 1}^{\boldsymbol{\theta}^{\prime}} \omega( {\rrset}^{(1)}_j ) \right] = \mathbb{E}[ \boldsymbol{\theta}^{\prime} ] \cdot EPT
\end{equation}

Therefore, the expected running time of Phase 1 is:

\begin{equation}
\Theta\left( \mathbb{E}[ \boldsymbol{\theta}^{\prime} ] \cdot ( EPT + 1 ) \right).
\end{equation}

\textbf{(3) Analysis of Phase 2}

Similarly, in Phase 2, we generate $\boldsymbol{\theta}$ RR sets independently. The expected running time of Phase 2 is:

\begin{equation}
\Theta\left( \mathbb{E}[ \boldsymbol{\theta} ] \cdot ( EPT + 1 ) \right)
\end{equation}

\textbf{(4) Total Expected Running Time}

Combining both phases, the total expected running time is:

\begin{equation}
\Theta\left( \left( \mathbb{E}[ \boldsymbol{\theta}^{\prime} ] + \mathbb{E}[ \boldsymbol{\theta} ] \right) \cdot ( EPT + 1 ) \right)
\end{equation}

\end{proof}

\begin{restatable}{lemma}{ept}\label{lem:ept}
Let $\tilde{\boldsymbol{v}}$ be a random node drawn from $V \setminus \seedset$ with probability proportional to the in-degree of $\tilde{\boldsymbol{v}}$ in $G^{\prime}$. Let $\boldsymbol{R}$ be a random RR set generated in $G^{\prime}$. Then:

\begin{equation}
EPT = \mathbb{E}_{\boldsymbol{R}}[ \omega( \boldsymbol{R} ) ] = \dfrac{m^{\prime}}{n^{\prime}} \cdot \mathbb{E}_{\tilde{\boldsymbol{v}}}[ \val( \tilde{\boldsymbol{v}} ) ],
\end{equation}

where:
\begin{itemize}
\item $m^{\prime}$ is the number of edges in $G^{\prime}$. % (excluding edges to $\seedset$).
\item $n^{\prime} = |V \setminus \seedset|$.
\item $\val(\tilde{\boldsymbol{v}})$ is the expected value function of $\tilde{\boldsymbol{v}}$. % in $G^{\prime}$.
\end{itemize}
\end{restatable}

\begin{proof}
For a fixed set $R \subseteq V$, let $p(R)$ be the probability that a randomly selected edge in $G^{\prime}$ points to a node in $R$. Since $R$ has $\omega(R)$ edges pointing to nodes in $R$, and the total number of edges in $G^{\prime}$ is $m^{\prime}$, we have that $p(R) = \dfrac{\omega(R)}{m^{\prime}}$.

Denote $d_v$ as the in-degree of a node $v$. We have:
\begin{equation}
\begin{aligned}
p(R) & =\sum_{(u, v) \in E^{\prime}} \frac{1}{m^{\prime}} \cdot \mathbb{I}\{v \in R\} \\
& =\sum_{v \in V} \frac{d_v}{m^{\prime}} \cdot \mathbb{I}\{v \in R\}=\mathbb{E}_{\tilde{\boldsymbol{v}}}[\mathbb{I}\{\tilde{\boldsymbol{v}} \in R\}]
\end{aligned}
\end{equation}

Then for a random RR set $\rrset$, we have:

\begin{equation}
\begin{aligned}
\mathbb{E}_{\rrset}[\omega(\rrset)] & =m^{\prime} \cdot \mathbb{E}_{\rrset}[p(\rrset)] \\
& =m^{\prime} \cdot \mathbb{E}_{\rrset}\left[\mathbb{E}_{\tilde{\boldsymbol{v}}}[\mathbb{I}\{\tilde{\boldsymbol{v}} \in \rrset\}]\right] \\
& =m^{\prime} \cdot \mathbb{E}_{\hat{\boldsymbol{v}}}\left[\mathbb{E}_{\rrset}[\mathbb{I}\{\tilde{\boldsymbol{v}} \in \rrset\}]\right] \\
& =m^{\prime} \cdot \mathbb{E}_{\hat{\boldsymbol{v}}}\left[\Pr(\tilde{\boldsymbol{v}} \in \boldsymbol{R})\right] \\
\end{aligned}
\end{equation}
%By \Cref{lem:rr_valuefunction}, 
Recall that $\val(\Sset) = \rrsize \cdot \Pr[\Sset \cap \rrset \neq \emptyset]$. Therefore, when the subset $\Sset$ is a single node $\tilde{\boldsymbol{v}}$, we have $\Pr(\tilde{\boldsymbol{v}} \in \boldsymbol{R}) = \frac{\val(\tilde{\boldsymbol{v}})}{\rrsize}$. Therefore, we have:
\begin{equation}
    EPT = \mathbb{E}_{\rrset}[\omega(\rrset)] = m^{\prime} \cdot \mathbb{E}_{\hat{\boldsymbol{v}}}[\frac{\val(\tilde{\boldsymbol{v}})}{\rrsize}] = \frac{m^{\prime}}{\rrsize} \cdot \mathbb{E}_{\hat{\boldsymbol{v}}}[\val(\tilde{\boldsymbol{v}})]
\end{equation}
\end{proof}

%\subsubsection{Bounding $\mathbb{E}[ \boldsymbol{\theta}^{\prime} ]$ and $\mathbb{E}[ \boldsymbol{\theta} ]$}

We now derive bounds on $\mathbb{E}[ \boldsymbol{\theta}^{\prime} ]$ and $\mathbb{E}[ \boldsymbol{\theta} ]$, the expected numbers of RR sets generated in Phases 1 and 2, respectively. First, we prove the following two inequalities:

\begin{lemma}\label{lem:est_bound}
For each $i=1,2, \ldots,\left\lfloor\log _2 \rrsize\right\rfloor-1$, if $\shap^{(k)} \geq \left(1+\varepsilon^{\prime}\right)^2 \cdot x_i$, then the following holds with probability at least $1-\frac{k}{2 |\seedset| {\rrsize}^{\ell} \log _2 \rrsize}$:
\begin{align}
    \frac{\rrsize \cdot \boldsymbol{est}^{(k)}_i}{\theta_i} &> \frac{\shap^{(k)}}{1+\varepsilon^{\prime}} \label{eq:est_bound_1}\\
    \frac{ \rrsize \cdot \boldsymbol{est}^{(k)}_i}{\theta_i} &> \left(1+\varepsilon^{\prime}\right) \cdot x_i \label{eq:est_bound_2}
\end{align}
\end{lemma}

\begin{proof}
Let $\rrset_1^{(1)}, \rrset_2^{(1)}, \ldots, \rrset_{\theta_i}^{(1)}$ be the $\theta_i$ generated RR sets by the end of the $i$-th iteration of the for-loop in Phase 1. Recall that $\boldsymbol{est}_{\seed,i} = \sum_{j = 1}^{\theta_i} X_{\rrset_j^{1}}(\seed)$.

If $\shap^{(k)} \geq\left(1+\varepsilon^{\prime}\right)^2 \cdot x_i$, then we can apply \Cref{fact:martingale_tail_bounds} and have:

\begin{equation}
\begin{aligned}
    &\Pr \left[\rrsize \cdot \frac{\boldsymbol{est}_{\seed,i}}{\theta_i} \leq \frac{\shap(\seed)}{1+\varepsilon^{\prime}}\right]\\
    &= \Pr \left[\boldsymbol{est}_{\seed,i} - \theta_i \cdot \frac{\shap(\seed)}{\rrsize} \leq -\frac{\varepsilon^{\prime}}{1+\varepsilon^{\prime}} \cdot \theta_i \cdot \frac{\shap(\seed)}{\rrsize}\right]\\
    &\leq \exp \left[-\frac{\varepsilon^{\prime 2}}{2\left(1+\varepsilon^{\prime}\right)^2} \cdot \theta_i \cdot \frac{\shap(\seed)}{\rrsize}\right] \quad \text{(By \Cref{fact:martingale_tail_bounds})}\\
    &\leq \exp \left[-\frac{\varepsilon^{\prime 2} \cdot x_i}{2 \rrsize} \cdot \theta_i\right] \quad \left(\text{By } x_i \leq \frac{\shap^{(k)}}{\left(1+e^{\prime}\right)^2} \leq \frac{\shap(\seed)}{\left(1+\varepsilon^{\prime}\right)^2}\right)\\
    & \leq \frac{1}{2 |\seedset| {\rrsize}^{\ell} \log _2 \rrsize}
\end{aligned}
\end{equation}

Then we take the union bound to obtain \Cref{eq:est_bound_1}. Note that there is at least one node $\seed$ with $\shap(\seed) \geq \shap^{(k)}$ and $\boldsymbol{est}_{\seed,i} \leq \boldsymbol{est}_i^{(k)}$. More precisely, there are at most $k$ such nodes $\seed$ with $\shap(\seed)$ ranked at or above $\shap^{(k)}$. Therefore, we have:

\begin{equation}
\begin{aligned}
    &\Pr \left[\rrsize \cdot \frac{\boldsymbol{est}^{(k)}_{i}}{\theta_i} \leq \frac{\shap^{(k)}}{1+\varepsilon^{\prime}}\right]\\
    & \leq \Pr\left[\exists \seed \in \seedset, \shap(\seed) \geq \shap^{(k)}, \rrsize \cdot \frac{\boldsymbol{est}_{\seed,i}}{\theta_i} \leq \frac{\shap(\seed) }{\left(1+\varepsilon^{\prime}\right)}\right] \\
    & \leq k \Pr\left[\rrsize \cdot \frac{\boldsymbol{est}_{\seed,i}}{\theta_i} \leq \frac{\shap(\seed) }{\left(1+\varepsilon^{\prime}\right)}\right] \\
    & \leq \frac{k}{2 |\seedset| {\rrsize}^{\ell} \log _2 \rrsize}
\end{aligned}
\end{equation}

Last, by applying $\shap^{(k)} \geq\left(1+\varepsilon^{\prime}\right)^2 \cdot x_i$ to \Cref{eq:est_bound_1}, we can also obtain \Cref{eq:est_bound_2}.
\end{proof}

Then we show the expected number of RR sets generated in both Phase 1 and Phase 2.

\begin{restatable}{lemma}{thetabound}\label{lem:theta_bound}

Under our algorithm, the expected number of RR sets generated satisfies:
\begin{align}
    \mathbb{E}[ \boldsymbol{\theta}^{\prime} ] &= O\left(\frac{n \ell \log n }{\shap^{(k)} \varepsilon^2}\right)\\
    \mathbb{E}[ \boldsymbol{\theta} ] &= O\left(\frac{n \ell \log n }{\shap^{(k)} \varepsilon^2}\right)
\end{align}

provided that $\ell \geq \dfrac{\log_2 k - \log_2 \log_2 \rrsize + \log_2 \rrsize - \log_2 |\seedset|}{\log_2 \rrsize}$.

\end{restatable}

\begin{proof}
We first prove for $\boldsymbol{\theta}^{\prime}$. We consider two cases seperately and then combine them: $\shap^{(k)}<\left(1+\varepsilon^{\prime}\right)^2 \cdot x_{\lfloor\log _2 \rrsize\rfloor-1}$ and $\shap^{(k)} \geq \left(1+\varepsilon^{\prime}\right)^2 \cdot x_{\lfloor\log _2 \rrsize\rfloor-1}$.

\textbf{Case 1:} $\shap^{(k)}<\left(1+\varepsilon^{\prime}\right)^2 \cdot x_{\lfloor\log _2 \rrsize\rfloor-1}$.
Recall that $\theta_i$ in Phase 1 is:
\begin{equation}    
    \theta_i = \left\lceil \frac{\rrsize(2+\frac{2}{3}\varepsilon')}{{\varepsilon'}^2\cdot x_i}\left(\ell \ln \rrsize + \ln |\seedset| + \ln (\log_2 \rrsize) + \ln 2 \right) \right\rceil
\end{equation}
Since $x_{\lfloor\log _2 \rrsize\rfloor-1} = \frac{\rrsize}{2^{\lfloor\log _2 \rrsize\rfloor-1}} \leq 2$, we have $\shap^{(k)}<4(1+ \varepsilon^{\prime})^2$. Thus in the worst case we can bound $\boldsymbol{\theta}^{\prime}$ based on $\shap^{(k)}$ as:

\begin{align}
& \boldsymbol{\theta}^{\prime}=\theta_{\lfloor \log _2 \rrsize\rfloor-1} \\
& \leq \left\lceil \frac{n^{\prime} \left( \ell \ln n^{\prime} + \ln |\seedset| + \ln \log_2 n^{\prime} + \ln 2 \right) \left( 2 + \frac{2}{3} \varepsilon^{\prime} \right)}{\varepsilon^{\prime 2}} \right\rceil \\
& \leq \left\lceil \frac{n^{\prime} \left( \ell \ln n^{\prime} + \ln |\seedset| + \ln \log_2 n^{\prime} + \ln 2 \right) \left( 2 + \frac{2}{3} \varepsilon^{\prime} \right) \cdot 4(1+ \varepsilon^{\prime})^2}{\varepsilon^{\prime 2} \cdot \shap^{(k)}} \right\rceil \label{eq:theta_prime_bound_1}\\ 
& =O\left( \frac{\rrsize (\ell \log \rrsize +\log |\seedset|) }{\shap^{(k)} \varepsilon^2}\right) \label{eq:theta_prime_bound_2}\\ 
& =O\left(\frac{n \ell \log n }{\shap^{(k)} \varepsilon^2}\right) \label{eq:theta_prime_bound_3}
\end{align}

From \Cref{eq:theta_prime_bound_1} to \Cref{eq:theta_prime_bound_2} follows the fact that $\varepsilon^{\prime}=\sqrt{2} \cdot \varepsilon$ and $\varepsilon$ is a sufficiently small positive parameter, so $\left( 2 + \frac{2}{3} \varepsilon^{\prime} \right) \cdot 4(1+ \varepsilon^{\prime})^2$ remains a constant. Moreover, we can relax the bound from \Cref{eq:theta_prime_bound_2} to \Cref{eq:theta_prime_bound_3} in terms of the size of the network $n$.

Similarly, for $\boldsymbol{\theta}$, since $\boldsymbol{LB} \geq 1$, we have:

\begin{equation}
\begin{aligned}
\boldsymbol{\theta} & \leq\left\lceil\frac{ \rrsize \left( \ell \ln \rrsize + \ln |\seedset| + \ln 4 \right) \left( 2 + \frac{2}{3} \varepsilon^{\prime} \right) }{\varepsilon^2}\right\rceil \\
& \leq\left\lceil\frac{ \rrsize \left( \ell \ln \rrsize + \ln |\seedset| + \ln 4 \right) \left( 2 + \frac{2}{3} \varepsilon^{\prime} \right)4(1+ \varepsilon^{\prime})^2 }{\varepsilon^2 \cdot \shap^{(k)}}\right\rceil \\
& =O\left( \frac{\rrsize (\ell \log \rrsize +\log |\seedset|) }{\shap^{(k)} \varepsilon^2}\right)\\ 
& =O\left(\frac{n \ell \log n }{\shap^{(k)} \varepsilon^2}\right) 
\end{aligned}
\end{equation}

\textbf{Case 2:} $\shap^{(k)} \geq \left(1+\varepsilon^{\prime}\right)^2 \cdot x_{\lfloor\log _2 \rrsize\rfloor-1}$

Let $i^*$ be the smallest index such that $(1+\varepsilon^{\prime})^2 \cdot x_{i^*} \leq \shap^{(k)}$. Therefore, $\shap^{(k)} < (1+\varepsilon^{\prime})^2 \cdot x_{i^*-1}$. Here we denote $x_0 = \rrsize$.

By the \Cref{eq:est_bound_2} in \Cref{lem:est_bound}, we know that with probability at least $1-\frac{k}{2 |\seedset| {\rrsize}^{\ell} \log _2 \rrsize}$, Phase 1 will stop at the $i^*$-th iteration. Similarly, we can bound $\boldsymbol{\theta}^{\prime}$ as:

\begin{equation}\label{eq:theta_prime_bound_case2}
\begin{aligned}
    \boldsymbol{\theta}^{\prime} &  = \theta_{i^*}  =  \left\lceil \frac{\rrsize \left( \ell \ln \rrsize + \ln |\seedset| + \ln \log_2 \rrsize + \ln 2 \right) \left( 2 + \frac{2}{3} \varepsilon^{\prime} \right)}{\varepsilon^{\prime 2} x_{i^*}} \right\rceil \\
    & \leq \left\lceil \frac{\rrsize \left( \ell \ln \rrsize + \ln |\seedset| + \ln \log_2 \rrsize + \ln 2 \right) \left( 2 + \frac{2}{3} \varepsilon^{\prime} \right) (1+\varepsilon^{\prime})^2}{\varepsilon^{\prime 2} \shap^{(k)}} \right\rceil \\
    & = O\left( \frac{\rrsize (\ell \log \rrsize +\log |\seedset|) }{\shap^{(k)} \varepsilon^2}\right) \\
    & =O\left(\frac{n \ell \log n }{\shap^{(k)} \varepsilon^2}\right) 
\end{aligned}
\end{equation}

By the \Cref{eq:est_bound_1} in \Cref{lem:est_bound}, when Phase 1 stops at the $i^*$-th iteration, the $\boldsymbol{LB}$ satisfies:
\begin{equation}
    \boldsymbol{LB} = \frac{\rrsize \cdot \boldsymbol{est}^{(k)}_{i^*}}{\theta_{i^*} \cdot (1+\varepsilon^{\prime})} \geq \frac{\shap^{(k)}}{(1+\varepsilon^{\prime})^2}
\end{equation}

Then we can bound $\boldsymbol{\theta}$ as:

\begin{equation}\label{eq:theta_bound_case2}
\begin{aligned}
    \boldsymbol{\theta} &\leq \left\lceil \frac{ \rrsize \left( \ell \ln \rrsize + \ln |\seedset| + \ln 4 \right) \left( 2 + \frac{2}{3} \varepsilon \right) (1+\varepsilon^{\prime})^2 }{ \varepsilon^2 \shap^{(k)} }  \right\rceil \\
    &= O\left( \frac{\rrsize (\ell \log \rrsize +\log |\seedset|) }{\shap^{(k)} \varepsilon^2}\right) \\
    &= O\left(\frac{n \ell \log n }{\shap^{(k)} \varepsilon^2}\right) 
\end{aligned}
\end{equation}

Moreover, Phase 1 will not stop at the $i^*$-th iteration with probability at most $\frac{k}{2 |\seedset| {\rrsize}^{\ell} \log _2 \rrsize}$. In the worst case, it continues to iteration $\lfloor \log_2 \rrsize \rfloor -1$, and $\boldsymbol{\theta^{\prime}}  =O\left( \frac{\rrsize (\ell \log \rrsize +\log |\seedset|) }{\varepsilon^2}\right) = O\left(\frac{n \ell \log n }{\varepsilon^2}\right)$.
Combine with \Cref{eq:theta_prime_bound_case2}, because the fact that $\shap^{(k)} \leq \rrsize$, and the condition that $\ell \geq \dfrac{\log_2 k - \log_2 \log_2 \rrsize + \log_2 \rrsize - \log_2 |\seedset|}{\log_2 \rrsize}$, we have:
\begin{equation}
        \begin{aligned}
            \mathbb{E}[\boldsymbol{\theta}^{\prime}] &= O\left(\frac{n \ell \log n }{\shap^{(k)} \varepsilon^2}\right) + \frac{k}{2 |\seedset| {\rrsize}^{\ell} \log _2 \rrsize} \cdot O\left(\frac{n \ell \log n }{\varepsilon^2}\right) \\
            &= O\left(\frac{n \ell \log n }{\shap^{(k)} \varepsilon^2}\right)
        \end{aligned}
    \end{equation}

Similarly, the worst case in Phase 2 is that $\boldsymbol{LB} = 1$, and we have: $\boldsymbol{\theta} = O\left( \frac{\rrsize (\ell \log \rrsize +\log |\seedset|) }{\varepsilon^2}\right) = O\left(\frac{n \ell \log n }{\varepsilon^2}\right)$. Combine with \Cref{eq:theta_bound_case2}, we have:
\begin{equation}
    \begin{aligned} 
        \mathbb{E}[\boldsymbol{\theta}] &= O\left(\frac{n \ell \log n }{\shap^{(k)} \varepsilon^2}\right) + \frac{k}{2 |\seedset| {\rrsize}^{\ell} \log _2 \rrsize} \cdot O\left(\frac{n \ell \log n }{\shap^{(k)} \varepsilon^2}\right) \\
        &= O\left(\frac{n \ell \log n }{\shap^{(k)} \varepsilon^2}\right)
    \end{aligned}
\end{equation}

Note that we set $\varepsilon^{\prime} = \sqrt{2} \cdot \varepsilon$ as suggested in \cite{tang2015influence}. Moreover, the condition on $\ell$ ensures the probability that Phase 1 will not stop at the $i^*$-th iteration is at most $\frac{1}{2\rrsize}$.

\end{proof}

%\subsubsection{Final Time Complexity}
Combining the results from above lemmas, we obtain the overall expected running time of our algorithm:
%\overallexpectedrunningtime*
\begin{restatable}{proposition}{overallexpectedrunningtime}\label{prop:overall_expected_running_time}
    Let $\tilde{\boldsymbol{v}}$ be a random node drawn from $V \setminus \seedset$ with probability proportional to the in-degree of $\tilde{\boldsymbol{v}}$ in $G^{\prime}$. Denote $m^{\prime}$ is the number of edges in $G^{\prime}$ and $\mathbb{E}_{\tilde{v}}[ \val( \tilde{v} ) ]$ is the expected value function of $\tilde{\boldsymbol{v}}$. Then the expected runtime complexity of $\algrrset$ is:  
 \begin{small}
    \begin{equation}
        O\left(\frac{n \ell \log n }{\shap^{(k)} \varepsilon^2} \cdot \dfrac{m^{\prime}}{n^{\prime}} \cdot \mathbb{E}_{\tilde{v}}[ \val( \tilde{v} ) ]\right),
    % O\left( \dfrac{n^{\prime} \left( \ell \ln n^{\prime} + \ln |T| \right) (\ln n^{\prime})}{\varepsilon^{2} \shap^{(k)}} \cdot \left( \dfrac{m^{\prime}}{n^{\prime}} \cdot \mathbb{E}_{\tilde{v}}[ \sigma^{\prime}( \{ \tilde{v} \} ) ] + 1 \right) \right)
    \end{equation}
    \end{small}
   % under the condition that
    
    $$
    \text{under condition} \ \ \begin{small}\ell \geq \dfrac{\log_2 k - \log_2 \log_2 \rrsize + \log_2 \rrsize - \log_2 |\seedset|}{\log_2 \rrsize}  \end{small}.
    $$
  
    %where $n = |V|$, $n^{\prime} = |V \setminus \seedset|$, $m^{\prime} = |E^{\prime}|$, and $\mathbb{E}_{\tilde{v}}[ \val( \tilde{v} ) ]$ is the expected value function of a random node in $V \setminus \seedset$.
\end{restatable}

\begin{proof}
From Lemma~\ref{lem:expected_running_time}, the expected running time is:

\begin{equation}
\Theta\left( \left( \mathbb{E}[ \boldsymbol{\theta}^{\prime} ] + \mathbb{E}[ \boldsymbol{\theta} ] \right) \cdot (EPT + 1 ) \right)
\end{equation}

Substituting $\mathbb{E}[ \boldsymbol{\theta}^{\prime} ]$ and $\mathbb{E}[ \boldsymbol{\theta} ]$ from \ref{lem:theta_bound}, we have:

\begin{equation}
\mathbb{E}[ \boldsymbol{\theta}^{\prime} ] + \mathbb{E}[ \boldsymbol{\theta} ] = O\left(\frac{n \ell \log n }{\shap^{(k)} \varepsilon^2}\right)
\end{equation}

From Lemma~\ref{lem:ept}, we have:
\begin{equation}
\text{EPT} = \dfrac{m^{\prime}}{n^{\prime}} \cdot \mathbb{E}_{\tilde{v}}[ \val( \tilde{v} ) ]
\end{equation}

Combining these, the expected running time is:

\begin{equation}
O\left(\frac{n \ell \log n }{\shap^{(k)} \varepsilon^2} \cdot \dfrac{m^{\prime}}{n^{\prime}} \cdot \mathbb{E}_{\tilde{v}}[ \val( \tilde{v} ) ]\right)
% O\left( \dfrac{ (\ell \ln \rrsize + \ln |\seedset|)\cdot (m^{\prime}\mathbb{E}_{\tilde{v}}[ \sigma^{\prime}( \{ \tilde{v} \} ) ] + \rrsize)}{\varepsilon^{2} \shap^{(k)}} \right)
\end{equation}

This completes the proof.
\end{proof}

\balance
\section{Details from Section~\ref{sec:experiments}}\label{sec:appendix_experiments}
\subsection{Parameter Sensitivity Analysis}\label{subsec:param_tuning}

\begin{figure}[!htbp]
    \centering
    \begin{subfigure}[b]{0.22\textwidth}
        \centering
        \includegraphics[width=\textwidth]{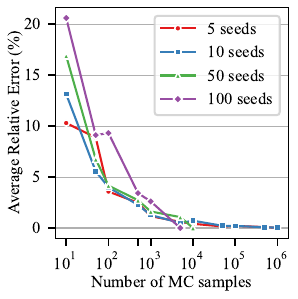}
        \caption{$\permuteMC$: number of MC samples $n_{\text{MC}}$ ($n_{\pi} = 500$)}
        \label{fig:param_permute_m_n500}
    \end{subfigure}
    \hfill
    \begin{subfigure}[b]{0.22\textwidth}
        \centering
        \includegraphics[width=\textwidth]{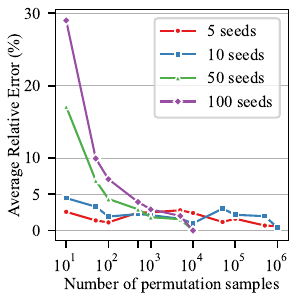}
        \caption{$\permuteMC$: number of permutations $n_{\pi}$ ($n_{\text{MC}} = 500$)}
        \label{fig:param_permute_n_m500}
    \end{subfigure}
    \vfill
    \begin{subfigure}[b]{0.22\textwidth}
        \centering
        \includegraphics[width=\textwidth]{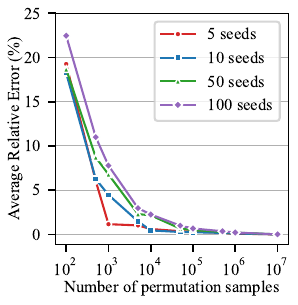}
        \caption{$\algliveedge$: number of sampled live edge graphs $n$}
        \label{fig:param_permute2_n}
    \end{subfigure}
    \hfill
    \begin{subfigure}[b]{0.22\textwidth}
        \centering
        \includegraphics[width=\textwidth]{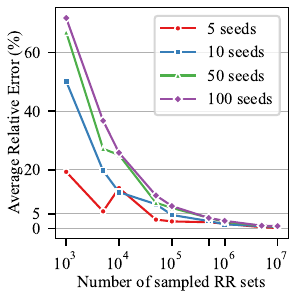}
        \caption{$\algrrset$: number of sampled RR sets $\theta$}
        \label{fig:param_rr_epsilon}
    \end{subfigure}
    \caption{Average relative error varying the parameters of different algorithms on \congressdata network}
    \label{fig:tune_parameters_congress}
\end{figure}

To investigate the effect of number of samples and select appropriate sample sizes to balance approximation accuracy and computational efficiency on approximation algorithms $\permuteMC$, $\algliveedge$ and $\algrrset$, we vary the sampling parameters of each algorithm on the \congressdata dataset under the \fullstep termination. \Cref{fig:tune_parameters_congress} shows average relative error for different seed set sizes (5, 10, 50, and 100), which is calculated against the result obtained using the largest number of samples for each varying parameters.
%For above experiments in \Cref{subsec:quality_analysis,subsec:runtime_analysis}, we set the sampling parameters for each algorithm to achieve average relative error less than $5\%$ to balance the accuracy and efficiency.

For $\permuteMC$, we vary two sampling parameters: the number of seed permutation samples $n_{\pi}$ (\Cref{fig:param_permute_n_m500}), and the number of Monte Carlo samples $n_{\text{MC}}$ that estimates value functions (\Cref{fig:param_permute_m_n500}). 
%As expected, increasing $n_{\text{MC}}$ reduces the error by improving the accuracy of influence spread estimation, and increasing $n_{\pi}$ reduces the error by better approximating the true expectation over all permutations. 
When $n_{\text{MC}} = 500$, the relative error compared with the benchmark $n_{\text{MC}}=10^6$ is less than $5\%$, and when $n_{\pi}=500$ the relative error compared with the benchmark $n_{\pi}=10^6$ is less than $4\%$. Thus we set $n_{\text{MC}}=500$ and $n_{\pi}=500$ for $\permuteMC$ for other experiments.
For $\algliveedge$, we vary the number of sampled live-edge graphs $n$. \Cref{fig:param_permute2_n} shows that the relative error converges faster when the seed size is smaller. Moreover, the relative error is decreasing to $\leq 5\%$ and becomes stable when $n=5,000$ for all seed sizes. Thus we set $n=5,000$.
For the \rrsetexp algorithm, we vary the number of sampled RR sets $\theta$. \Cref{fig:param_rr_epsilon} shows the relative error drops below $4\%$ when $\theta = 50,000$. Thus we set $\theta = 50,000$.

\subsection{\textbf{Impact of termination step $\timeconst$}}\label{subsec:impact_k_appendix}
%We investigate how the termination step $\timeconst$ affects accuracy and runtime on synthetic graphs. 
Figure~\ref{fig:vary_step_shapley} shows increasing $\timeconst$ yields diminishing accuracy returns, as most activations occur in early steps ($\timeconst \leq 16$). While both algorithms converge similarly, \algliveedgeexp's error increases while \rrsetexp's decreases as $\timeconst$ grows. This occurs because \algliveedgeexp's forward graph traversal incurs higher variance with longer propagation paths, while \rrsetexp's backward graph traversal benefits from deeper paths. Figure~\ref{fig:vary_step_runtime} shows \rrsetexp's runtime grows with $\timeconst$ due to deeper reverse sampling, while \algliveedgeexp remains more stable as the cost of sampling live-edge graphs is independent of $\timeconst$. Thus, \algliveedgeexp excels when diffusion terminates quickly, while \rrsetexp maintains superior efficiency for all $\timeconst$.

\begin{figure}[!htbp]
\begin{minipage}{\linewidth}
\begin{subfigure}{0.47\columnwidth}
    \centering
    \includegraphics[width=\columnwidth]{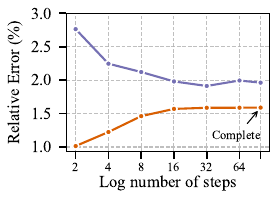}
    \caption{Relative error}
    \label{fig:vary_step_shapley}
\end{subfigure}%
\hfill
\begin{subfigure}{0.5\columnwidth}
    \centering
    \includegraphics[width=\columnwidth]{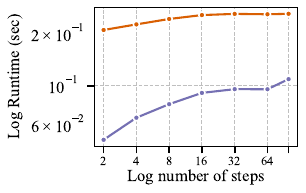}
    \caption{Runtime}
    \label{fig:vary_step_runtime}
\end{subfigure}
\end{minipage}

\begin{minipage}{\columnwidth}
    \centering
    \includegraphics[width=\columnwidth]{figures/experiments_new/synthetic_quality/algorithm_legend_fixed_point.pdf}
\end{minipage}
\caption{Varying termination step $\timeconst$ on synthetic graph (5K nodes, 10 avg. degree, 500 seeds).}
\label{fig:vary_step_experiment}
\end{figure}

% \clearpage

\end{document}